\documentclass[10pt,conference]{IEEEtran}
\IEEEoverridecommandlockouts

\usepackage{cite}
\usepackage{amsmath,amssymb,amsfonts}
\usepackage{algorithmic}
\usepackage{algorithm}
\usepackage{graphicx}
\usepackage{textcomp}
\usepackage{xcolor}
\usepackage[hidelinks]{hyperref}
\usepackage{booktabs}
\usepackage{multirow}
\usepackage{tabularx}
\usepackage{amsthm}
\usepackage{subcaption}
\usepackage{float}
\usepackage{placeins}
\usepackage{balance}

\usepackage{tikz}
\usetikzlibrary{positioning, arrows, calc, fit, backgrounds, decorations.pathreplacing, shapes.geometric}

\newtheorem{theorem}{Theorem}

\newtheorem{definition}{Definition}

\newtheorem{proposition}{Proposition}

\theoremstyle{definition}
\newtheorem{assumption}{Assumption}

\newtheorem{corollary}{Corollary}[theorem]

\theoremstyle{remark}
\newtheorem{remark}{Remark}

\graphicspath{{figures/}}

\def\BibTeX{{\rm B\kern-.05em{\sc i\kern-.025em b}\kern-.08em
    T\kern-.1667em\lower.7ex\hbox{E}\kern-.125emX}}

\begin{document}

\title{FairWave: A Fairness-Aware Asynchronous DAG-BFT Consensus}

\author{\IEEEauthorblockN{Syariful Mujaddiq}
	\IEEEauthorblockA{Independent Researcher} \\
	\texttt{syariful.mujaddiq@coconut.or.id}}

\maketitle

\begin{abstract}
Proof-of-Stake DAG-BFT consensus faces a trilemma between Sybil resistance, reward fairness, and plutocracy. Existing protocols prioritize liveness over fair stake-based selection, driving longitudinal centralization.

FairWave is a dual-channel DAG-BFT protocol that separates anchor selection from reward distribution. The selection channel is super-linear in stake, guaranteeing Sybil gain $<1$ for $K>1$; the reward channel is sub-linear via square-root stake normalization. DAG-derived uptime and latency metrics eliminate external oracles, and lagged reputation breaks circular dependency between selection outcomes and weights.

Evaluated through approximately 550,000 Monte Carlo rounds against eight baselines, FairWave shows Gini 0.140 (vs.~Pure-PoS 0.490), monotone HHI reduction from 0.039 to 0.020 over 50,000 epochs, and optimal Sybil split $K^{*}=1$. Safety follows unconditionally from the $2f+1$ commit rule; the liveness model predicts monotone degradation from $94.0\%$ at $b=0.20$ to $74.0\%$ at $b=1/3$, consistent with the architectural expectation of no discontinuous cliff.
\end{abstract}


\section{Introduction}
\label{sec:intro}

\subsection{Motivation}

Long-running Proof-of-Stake blockchains exhibit empirical centralization despite formally claiming decentralized governance~\cite{motepalli2024, centralization_pos2024, buterin2020}. Pure stake-based consensus protocols suffer from three well-characterized pathologies:

\begin{enumerate}
	\item \textbf{Rich-get-richer dynamics.} When rewards are proportional to stake, accumulated rewards compound into stake growth at a rate proportional to current holdings, inducing exponential divergence of the Herfindahl-Hirschman Index (HHI) toward one.
	\item \textbf{Sybil vulnerability under sub-linear weighting.} Alternative rules such as \emph{Square Root Stake Weighting} (SRSW) $\sqrt{\text{stake}}$ or \emph{Logarithmic Stake Weighting} (LSW) $\ln(1+\text{stake})$ that attempt to achieve fairness instead introduce Sybil gain proportional to $\sqrt{K}$ or $K$ respectively, where $K$ denotes the split factor.
	\item \textbf{Plutocratic convergence.} Under proportional rewards, the share of power of the top-1 validator increases monotonically, eventually violating the BFT safety threshold $f < n/3$.
\end{enumerate}

These pathologies are not merely theoretical---empirical measurements on Ethereum~2.0, Cosmos Hub, and Solana mainnet show Nakamoto coefficients as low as 3--7 for the top staking pools~\cite{decentralization_measurement2022, eth_beacon2023, daian2020}.

\subsection{Contributions}

This work presents the following contributions:

\begin{itemize}
	\item \textbf{C1.} A novel dual-channel consensus architecture in which selection weights and reward weights exhibit \emph{opposite curvature} with respect to stake-mitigates the Sybil-fairness-plutocracy trilemma.
	\item \textbf{C2.} A Closed-form analytical derivation of Sybil gain bounds for four reward rules (Pure-PoS, SRSW, LSW, FairWave), validated to machine precision ($\epsilon = 0.0$) against $20,160$ numerical evaluation cells.
	\item \textbf{C3.} Empirical evaluation of nine tasks spanning approximately $550,000$ Monte Carlo rounds, compared against protocol models representing Narwhal+Tusk~\cite{narwhal2022}, Bullshark~\cite{bullshark2022}, Mysticeti~\cite{mysticeti2024}, HotStuff~\cite{hotstuff2019}, HotStuff-2~\cite{hotstuff2_2023}, PBFT~\cite{pbft1999}, Tendermint~\cite{tendermint2014}, and Algorand~\cite{algorand2017} under a unified simulation framework.
	\item \textbf{C4.} Quantitative characterization of robustness via One-At-a-Time (OAT) elasticity analysis, Sobol variance decomposition, and combined perturbation mapping---To the best of our knowledge, this work presents the first multi-method input sensitivity analysis (combining OAT, Sobol, and combined perturbation) applied to DAG-BFT consensus parameters.
\end{itemize}

\subsection{Paper Organization}

Section~\ref{sec:related} surveys related work. Section~\ref{sec:model} formalizes the threat model and system model. Section~\ref{sec:design} presents the FairWave protocol design. Section~\ref{sec:math} derives the mathematical foundations. Section~\ref{sec:methodology} describes the empirical methodology. Section~\ref{sec:results} presents results from nine evaluation tasks. Section~\ref{sec:discussion} discusses trade-offs, and Section~\ref{sec:conclusion} concludes.

\section{Related Work and Background}
\label{sec:related}

\subsection{Asynchronous BFT Consensus}

Classical PBFT~\cite{pbft1999} achieves Byzantine agreement under partial synchrony with message complexity $O(n^2)$ per view. HotStuff~\cite{hotstuff2019} reduces view-change complexity to $O(n)$ via threshold signatures and pipelined chain structure. HotStuff-2~\cite{hotstuff2_2023} further optimizes to two-phase commit with optimal responsiveness. The DAG-BFT paradigm, initiated by DAG-Rider~\cite{dagrider2022} and Narwhal+Tusk~\cite{narwhal2022}, separate data dissemination from consensus ordering by embedding transactions in a structured Directed Acyclic Graph (DAG). Subsequent refinements---Bullshark~\cite{bullshark2022} (reducing commit latency to 2 waves), Shoal~\cite{shoal2023} (pipelining and leader reputation), Mysticeti~\cite{mysticeti2024} (achieving fast-path commit sub-wave), and Shoal++~\cite{shoalpp2024} (near-optimal latency)---further optimize the latency-throughput frontier while maintaining Byzantine tolerance $f < n/3$. The asynchronous BFT landscape was pioneered by the HoneyBadgerBFT protocol~\cite{miller2016}, which demonstrated that practical asynchronous consensus is achievable without timing assumptions, building on the theoretical foundation of Ben-Or's agreement protocol~\cite{benor1983}.

\subsection{Proof-of-Stake Reward Rules}

The choice of stake-to-weight mapping fundamentally determines the Sybil resistance and fairness properties of a protocol. Table~\ref{tab:rules_comparison} summarizes four representative rules.

\begin{table}[!tb]
	\centering
	\caption{Comparison of PoS reward rules and Sybil gain characteristics. $P$ denotes the real-time performance score (Eq.~\ref{eq:wsel}); $S$ denotes the normalized stake factor (Eq.~\ref{eq:stake_factor}).}
	\label{tab:rules_comparison}
	\begin{tabular}{lll}
		\toprule
		\textbf{Rule}                & \textbf{Weight function} & \textbf{Sybil gain}       \\
		\midrule
		Pure-PoS~\cite{algorand2017} & $w = s$                  & $g(K) = 1$                \\
		SRSW~\cite{cosmos2019}       & $w = \sqrt{s}$           & $g(K) = \sqrt{K}$         \\
		LSW                          & $w = \ln(1+s)$           & $g(K) \approx K$          \\
		\textbf{FairWave}            & $w = P \cdot S^2$        & $g(K) < 1\;\forall K > 1$ \\
		\bottomrule
	\end{tabular}
\end{table}

Pure-PoS (Algorand, Tendermint) is Sybil-neutral: stake splitting yields neither gain nor loss. SRSW (used conceptually in Cosmos governance discussions~\cite{cosmos2019}) introduces sub-linear fairness at the cost of Sybil gain $\sqrt{K}$. LSW worsens vulnerability to approximately linear gain relative to split factor, rendering it catastrophically exploitable. The Ouroboros protocol~\cite{ouroboros2017} introduced provably secure PoS with formal reward analysis, and Ethereum~2.0's Casper FFG~\cite{ethereumpos2022} adopts stake-based weighting with inactivity penalties. As a baseline for divergent concentration, we also define \emph{PoS$^2$}, a strawman protocol where rewards are proportional to $s_i^2$, which exacerbates plutocratic dynamics. Alternative approaches to address wealth compounding have been proposed~\cite{frd2022, pos_reward_governance2024, pos_concentration2022}, but none simultaneously mitigates the Sybil-fairness-plutocracy trilemma.

\subsection{Decentralization Metrics}

We employ three complementary decentralization metrics: the \emph{Nakamoto coefficient}~\cite{decentralization_measurement2022} (minimum coalition to achieve a $1/3$ threshold, Fig.~\ref{fig:nakamoto_baseline}); the \emph{Gini coefficient}~\cite{gini1921} (inequality of reward distribution); and the \emph{Herfindahl-Hirschman Index} (HHI, sum of squared market shares). These metrics capture different aspects of centralization: coalition vulnerability, distributional inequality, and concentration.

\begin{figure}[!tb]
	\centering
	\includegraphics[width=\columnwidth]{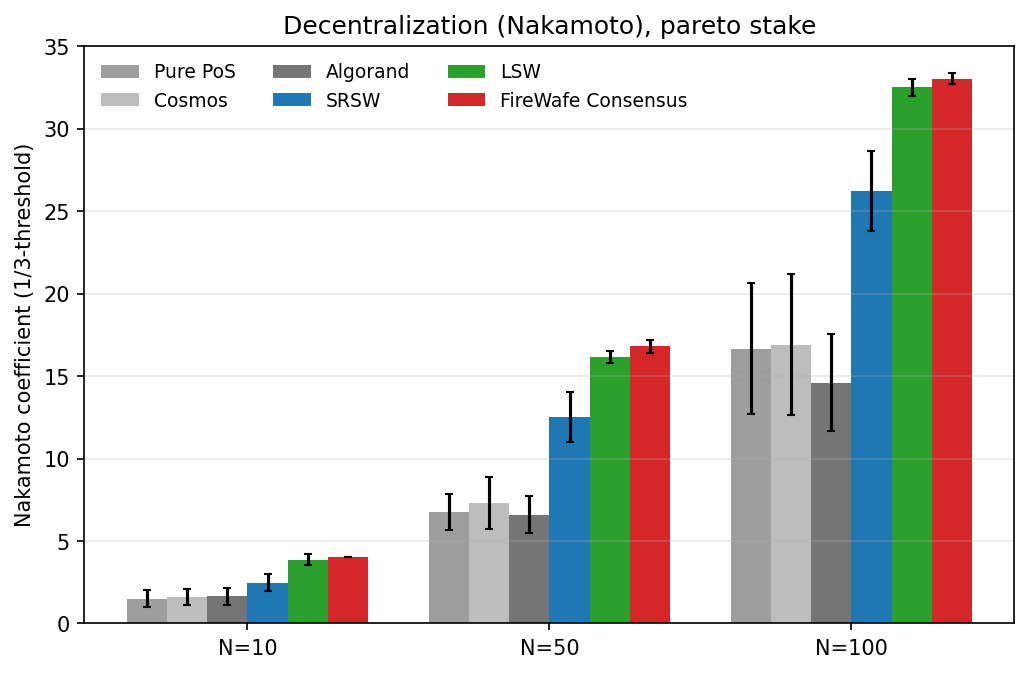}
	\caption{The Nakamoto coefficient, which represents the minimum coalition needed for 1/3 power, is evaluated across six reward rules and three network sizes. FairWave consistently shows a higher Nakamoto coefficient than Pure-PoS across all scales, indicating better decentralization resilience.}
	\label{fig:nakamoto_baseline}
\end{figure}

\subsection{Sensitivity Analysis in Distributed Systems}

Individual sensitivity methods (OAT elasticity~\cite{cukier1973}, Sobol variance decomposition~\cite{saltelli2010,sobol2001}) have been applied to distributed systems in prior work. To the best of our knowledge, this work is the first to combine all three modalities---OAT, Sobol, and combined perturbation mapping---within a unified input sensitivity characterization for a DAG-BFT consensus protocol, providing quantitative evidence of parametric stability under joint input variation.

\section{Threat Model and System Model}
\label{sec:model}

\subsection{System Model}

We consider a set of $n$ validators $\mathcal{V} = \{v_1, \ldots, v_n\}$ with $n \geq 4$, each possessing attributes $(s_i, u_i, \text{rep}_i, \ell_i)$ representing stake, uptime, reputation score, and latency score. Our system model follows the standard distributed systems formalism~\cite{cachin2011}: processes communicate via authenticated point-to-point channels, and we assume a computationally bounded adversary.

\textbf{Stake bounds.} Validator stakes are constrained: $s_i \in [S_{\min}, S_{\max}]$ with $S_{\min} = 5.000$ and $S_{\max} = 50.000$ token units. The normalized stake factor is defined as:
\begin{equation}
	S(i) = \sqrt{s_i / S_{\max}}, \quad S(i) \in [\sqrt{S_{\min}/S_{\max}}, 1]
	\label{eq:stake_factor}
\end{equation}

\textbf{Reward score.} The multi-factor composite reward score for validator $i$ is:
\begin{equation}
	R(i) = \alpha \cdot S(i) + \beta \cdot U(i) + \gamma \cdot \text{Rep}_{\text{active}}(i, e) + \delta \cdot L(i)
	\label{eq:reward_score}
\end{equation}
with simplex constraint $\alpha + \beta + \gamma + \delta = 1$ and default parameterization $(\alpha, \beta, \gamma, \delta) = (0.20;\; 0.25;\; 0.40;\; 0.15)$.

\textbf{Lagged reputation.} The reputation factor $\text{Rep}(i)$ evolves throughout each epoch via \emph{increment}, \emph{decay}, and \emph{slashing} (see Algorithm~\ref{alg:fairwave} and Section~\ref{sec:discussion}.C). To prevent circular dependency between epoch $e$ selection results and epoch $e$ selection weights, FairWave introduces a temporal separation between \emph{observed} and \emph{active} reputations. Let $\text{Rep}_{\text{final}}(i, e)$ denote the final reputation value of validator~$i$ at epoch closure (i.e., after all \emph{increments}, \emph{decay}, and \emph{slashing} throughout epoch~$e$ have been applied). Then the active reputation used for $R(i)$ computation at epoch~$e$ is:
\begin{equation}
	\text{Rep}_{\text{active}}(i, e) = \text{Rep}_{\text{final}}(i, e - 1)
	\label{eq:rep_active}
\end{equation}
In other words, the reputation value used in $W_{\text{rew}}$ at epoch~$e$ is the value \emph{frozen} at epoch~$e-1$ closure, and reputation updates occurring throughout epoch~$e$ only affect $W_{\text{rew}}$ at epoch~$e+1$. This is analogous to \emph{delayed validator activation} practice in some production PoS protocols but applied here specifically to the reputation component to ensure \emph{forward-causality} property in the dual-channel architecture (see Proposition~\ref{prop:no_circular} in Section~\ref{sec:math}).

Fig.~\ref{fig:lagged_rep_timing} illustrates the lagged-reputation timeline and the SNAPSHOT phase that freezes $\mathrm{Rep}_{\mathrm{final}}(e)$ for use at epoch $e{+}1$.
\begin{figure*}[!tb]
	\centering
	\includegraphics[width=0.85\textwidth,keepaspectratio]{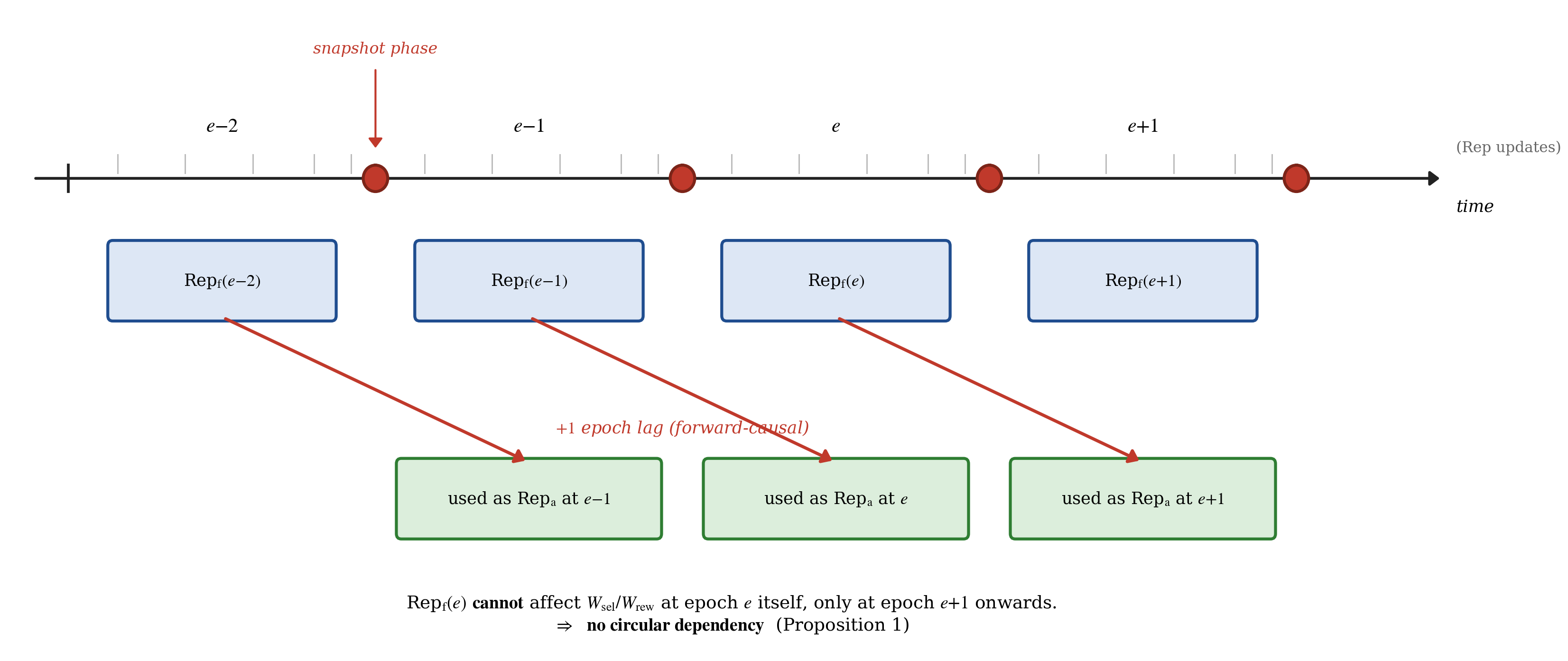}
	\caption{Lagged reputation timing. Reputation updates (small ticks) accrue continuously through each epoch. At every epoch boundary, the \emph{SNAPSHOT} phase (red dot) freezes $\mathrm{Rep}_{\mathrm{final}}(e)$. That frozen value is then used as $\mathrm{Rep}_{\mathrm{active}}$ \emph{one epoch later} (red arrows). This forward-causal arrangement eliminates the circular dependency between epoch-$e$ selection results and epoch-$e$ selection weights (Proposition~\ref{prop:no_circular}).}
	\label{fig:lagged_rep_timing}
\end{figure*}

\textbf{Validator registry.} The validator set $\mathcal{V}_e$ for epoch $e$ is determined atomically at epoch $e-1$ boundary via a seven-phase ceremony: \texttt{FREEZE $\to$ SLASH $\to$ REGISTRY $\to$ SEED $\to$ SNAPSHOT $\to$ RESET $\to$ BARRIER}. The set is locked per-epoch; no stake modifications occur intra-epoch.

Fig.~\ref{fig:epoch_ceremony} illustrates the atomically-executed seven-phase ceremony that establishes the next-epoch validator set and the frozen reputation snapshot.
\begin{figure*}[!tb]
	\centering
	\includegraphics[width=0.85\textwidth,keepaspectratio]{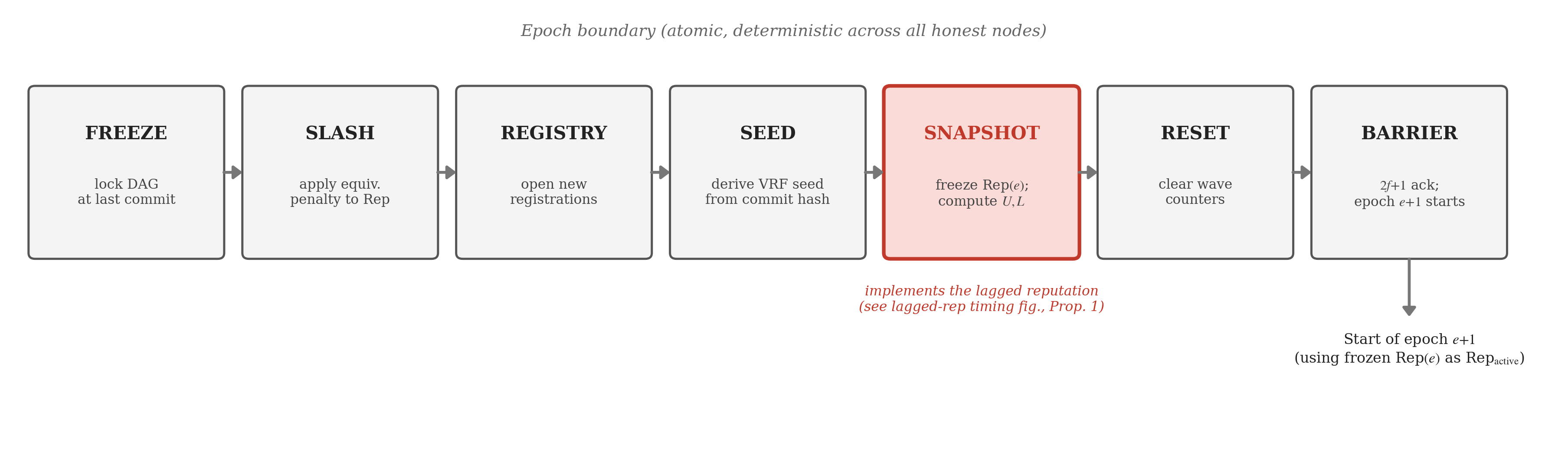}
	\caption{Per-epoch seven-phase ceremony executed atomically at every epoch boundary. The pipeline ensures all honest validators converge on the next-epoch validator set $\mathcal{V}_{e+1}$, the VRF seed, the frozen reputation snapshot $\mathrm{Rep}(e)$, and the new weight vectors $W_{\mathrm{sel}}^{e+1}, W_{\mathrm{rew}}^{e+1}$ \emph{without an external aggregator}. The \emph{SNAPSHOT} phase (highlighted red) implements the lagged-reputation freezing visualized in Fig.~\ref{fig:lagged_rep_timing}.}
	\label{fig:epoch_ceremony}
\end{figure*}

\subsection{Network Model}

For \emph{safety}, the protocol maintains safety under full asynchrony. An anchor commit requires observing strong support from $\geq 2f+1$ vertices that are causally reachable---not from the delivery deadline.

For \emph{liveness}, we assume partial synchrony (Dwork-Lynch-Stockmeyer~\cite{dls1988}): there exists a Global Stabilization Time (GST) after which all messages are delivered within delay bound $\Delta$. Before GST, liveness is not guaranteed, but safety is maintained unconditionally.

\subsection{Adversary Model}

\textbf{Static Byzantine adversary.} The adversary controls at most $f = \lfloor(n-1)/3\rfloor$ validators (by weighted stake). Byzantine validators can equivocate (produce conflicting vertices for the same round), selectively withhold messages, and coordinate timing within delay bound $\Delta$.

\textbf{Sybil's adversary.} An attacker controlling aggregate stake $X$ can split into $K$ Sybil identities~\cite{douceur2002}, each holding $X/K$ stake. We analyze selection gain $g(K)$ as the ratio of adversarial power after splitting to power before splitting.

\textbf{Equivocation detection.} Conflicting vertices from the same validator in the same round are detected with probability $p_{\text{detect}} = 0.95$ via VRF proof validation and cross-referencing DAG parent sets.

\subsection{Derivation of Non-Stake Metrics from DAG Structure}
\label{sec:dag_metrics}

The composite reward score (Eq.~\ref{eq:reward_score}) combines four factors: stake $S(i)$, uptime $U(i)$, reputation $\text{Rep}_{\text{active}}(i, e)$, and latency $L(i)$. Stake is external to consensus (locked in validator registry per epoch), and reputation is updated endogenously from commit results (see Algorithm~\ref{alg:fairwave} and Eq.~\ref{eq:rep_active}). This subsection formalizes that the remaining two factors---uptime and latency---are \emph{deterministically derived from the finalized DAG structure}, so honest validators agree on their values without requiring an \emph{external oracle}.

\textbf{Notation.} For epoch~$e$, let $\mathcal{R}_e = \{r_1, r_2, \ldots, r_{R_e}\}$ denote the set of DAG rounds in that epoch, with $R_e = |\mathcal{R}_e|$. Let $\mathcal{D}_e$ denote the finalized DAG throughout epoch~$e$, i.e., the sub-DAG all of whose vertices have entered the \emph{causal history} of at least one committed anchor. For validator~$i$, define $\mathcal{V}_i(e) = \{v \in \mathcal{D}_e : \text{proposer}(v) = i\}$ as the set of vertices proposed by~$i$ and entered into the finalized DAG at epoch~$e$.

\begin{definition}[DAG-Derived Uptime]
	\label{def:uptime}
	Uptime of validator~$i$ at epoch~$e$ is defined as the ratio of rounds in which vertex~$i$ entered the finalized DAG to total epoch rounds:
	\begin{equation}
		U(i, e) = \frac{|\mathcal{V}_i(e)|}{R_e} \in [0, 1].
		\label{eq:uptime}
	\end{equation}
\end{definition}

In practice, DAG metrics can be used to flag suspicious validators.

Since in DAG-BFT each validator proposes at most one vertex per round, $|\mathcal{V}_i(e)| \le R_e$ and $U(i,e)$ directly measures continuous participation: value $U = 1$ indicates the validator is present each round throughout the epoch; value $U = 0$ indicates the validator is perpetually offline.

\begin{definition}[DAG-Derived Latency]
	\label{def:latency}
	For vertex $v \in \mathcal{V}_i(e)$, let
	\begin{multline}
		\label{eq:tau_2f1}
		\tau_{2f+1}(v) = \min\Big\{\, r > \text{round}(v) \;\Big|\\
		\big|\{\, j : \exists\, v_j \in \mathcal{D}_e,\, \text{round}(v_j) \le r, \\
		v \in \text{anc}^{*}(v_j) \,\}\big| \ge 2f+1 \,\Big\},
	\end{multline}
	the first round where at least $2f+1$ distinct validators have vertices that reference $v$ causally (directly or transitively), with $\text{anc}^{*}(\cdot)$ denoting \emph{transitive parent closure}. The round-offset of vertex~$v$ is $\delta(v) = \tau_{2f+1}(v) - \text{round}(v) \ge 1$. The raw latency of validator~$i$ at epoch~$e$ is the median round-offset across all its vertices:
	\begin{equation}
		\tilde{\ell}(i, e) = \operatorname*{median}_{v \in \mathcal{V}_i(e)} \delta(v).
		\label{eq:latency_raw}
	\end{equation}
	For composition into the reward score, raw values are normalized to $[0,1]$ via a monotone-decreasing mapping:
	\begin{equation}
		L(i, e) = 1 - \min\!\left( \frac{\tilde{\ell}(i, e) - 1}{\ell_{\max} - 1},\; 1 \right),
		\label{eq:latency_score}
	\end{equation}
	with $\ell_{\max}$ a protocol parameter (default $\ell_{\max} = 2f+1$).
\end{definition}

Validators whose vertices are consistently referenced in the next round ($\delta = 1$) obtain $L = 1$; validators whose vertices are covered only after $\ge \ell_{\max}$ rounds obtain $L = 0$.

\textbf{Determinism and Byzantine tolerance.} Equations~\eqref{eq:uptime}--\eqref{eq:latency_score} reference only (i)~the finalized DAG structure $\mathcal{D}_e$ and (ii)~the validator set $\mathcal{V}_e$ at that epoch. Both objects are already agreed upon via the DAG-BFT commit mechanism itself: any \emph{honest node} running the protocol stores $\mathcal{D}_e$ that is identical (modulo \emph{causal history} from the same committed anchor) and validator registry $\mathcal{V}_e$ identical (locked at epoch~$e-1$ boundary). Consequently, functions $U(i, e)$ and $L(i, e)$ are \emph{deterministic functions of consensus state}, and any two \emph{honest nodes} computing them at epoch~$e$ boundary will obtain \textbf{bit-exact identical} numerical values. Byzantine tolerance over $U$ and $L$ values is \textbf{inherited from DAG-BFT}: no new attack surface is introduced, because Byzantine validators claiming different $U$ or $L$ values would be detected via local recomputation by any honest node.

\textbf{Eliminating the external oracle.} Unlike quality-based reward schemes commonly implemented on production PoS protocols (e.g.,\ Polkadot era points~\cite{polkadot2020} or Solana skip-rate penalties), which partially rely on external aggregators or oracles to report availability and latency, FairWave \emph{requires no component outside the consensus protocol set}. neither a separate committee, nor a trusted reporter, nor an additional communication channel for agreeing on $U$ and $L$ values. Consequently, FairWave's \emph{threat model} for reward computation is identical to the \emph{threat model} of the underlying DAG-BFT: if DAG consensus is safe under $f < n/3$, then reward score computation is also safe under the same assumption.

\section{FairWave Protocol Design}
\label{sec:design}

\subsection{Design Rationale}
\label{sec:design_rationale}

Three fundamental architectural decisions underpin FairWave: asynchronous security, DAG-based data dissemination, and VRF-based leader selection. We justify each in sequence.

\textbf{Why asynchronous BFT, not partially synchronous BFT?}
Classical leader-based protocols (PBFT~\cite{pbft1999}, HotStuff~\cite{hotstuff2019}, and Tendermint~\cite{tendermint2014}) depend on the \emph{Global Stabilization Time} (GST) assumption: after an unknown time, all messages arrive within a known bound $\Delta$. Before GST, liveness is \emph{not} guaranteed---if the network remains unstable, the protocol stalls until a timeout expires and a view change selects a new leader. On wide-area deployments with heterogeneous ISPs, transient partitions can persist for minutes, triggering cascading timeouts of many multiples of $\Delta$. Asynchronous BFT protocols, in contrast, guarantee safety and liveness under \emph{arbitrary} message delays and require only eventual delivery.

FairWave inherits this asynchronous property through its DAG-BFT structure: when a VRF-selected anchor is unreachable, the protocol executes an \emph{instant wave skip} with zero timeout delay, in contrast to partially synchronous protocols that block the entire timeout window ($\geq 8$~s in a typical Tendermint configuration). The \emph{safety} property (agreement and validity) is a formal property that follows from the $2f{+}1$ strong-support commit rule for any correct DAG-BFT protocol under $f < n/3$, and holds unconditionally; consequently, adversarial evaluation at Task~F (Section~\ref{sec:bft_validation}) does not present safety as a differential claim, but instead focuses FairWave's differentiation on the \emph{liveness-degradation curve} around the Byzantine boundary.

\textbf{Why DAG-based consensus, not a linear chain?}
In linear-chain BFT, a single leader proposes one block per round, creating a throughput bottleneck proportional to leader bandwidth. The DAG paradigm, pioneered by Narwhal+Tusk~\cite{narwhal2022}, separates data dissemination from ordering: \emph{all} $N$ validators produce vertices in parallel each round, yielding $N\times$ data availability throughput. Consensus ordering is then applied retroactively by selecting anchor vertices and committing their causal history. This eliminates three pathologies of leader-based protocols: (i)~single-leader bandwidth bottleneck, (ii)~explicit multi-phase voting (DAG edges function as implicit votes), and (iii)~view-change mechanism (replaced by costless wave skip). Empirically, DAG-BFT protocols achieve throughput scaling linearly with $N$ (see Fig.~\ref{fig:throughput_scaling} in Section~\ref{sec:results}), whereas leader-BFT plateaus or degrades above $N = 50$ due to $O(N^2)$ message complexity.

\textbf{Why VRF-based anchor selection?}
Deterministic leader schedules (round-robin or stake-proportional) are \emph{predictable}: an adversary can target the next leader with denial-of-service attacks before their slot begins. Verifiable Random Functions (VRF, RFC~9381 ECVRF Edwards25519-SHA512-Elligator2) provide three properties that address this: (i)~\emph{unpredictability}---only the holder of the secret key can compute the VRF output, so no party can predict the anchor before the selection round; (ii)~\emph{verifiability}---any observer can verify the VRF proof against the public key, ensuring integrity without trusted setup; (iii)~\emph{non-interactivity}---each validator computes the racing key $k_i = -\ln(u_i)/W_{\text{sel}}(i)$ independently, requiring no additional communication round. The exponential race mechanism ensures that selection is distributionally equivalent to weighted sampling proportional to $W_{\text{sel}}$, preserving the security properties of the dual-channel architecture. The marginal overhead is $+2$~ms commit latency versus Bullshark (Table~\ref{tab:performance}), an acceptable trade-off for DoS-resistant leader selection.

\subsection{DAG-BFT Wave Structure}

FairWave operates on a DAG-BFT wave protocol. Each wave consists of $r = 3$ rounds. Vertices are cryptographically signed messages that reference $\geq 2f+1$ vertices from the previous round as parents. Anchors are selected per wave via weighted VRF sortition using $W_{\text{sel}}$. A Commit occurs when $2f+1$ support vertices referencing the anchor are observed.

Fig.~\ref{fig:dag_wave_structure} illustrates the DAG-wave and anchor-selection structure and the $2f{+}1$ strong-support commit condition.
\begin{figure}[!tb]
	\centering
	\includegraphics[width=\columnwidth,keepaspectratio]{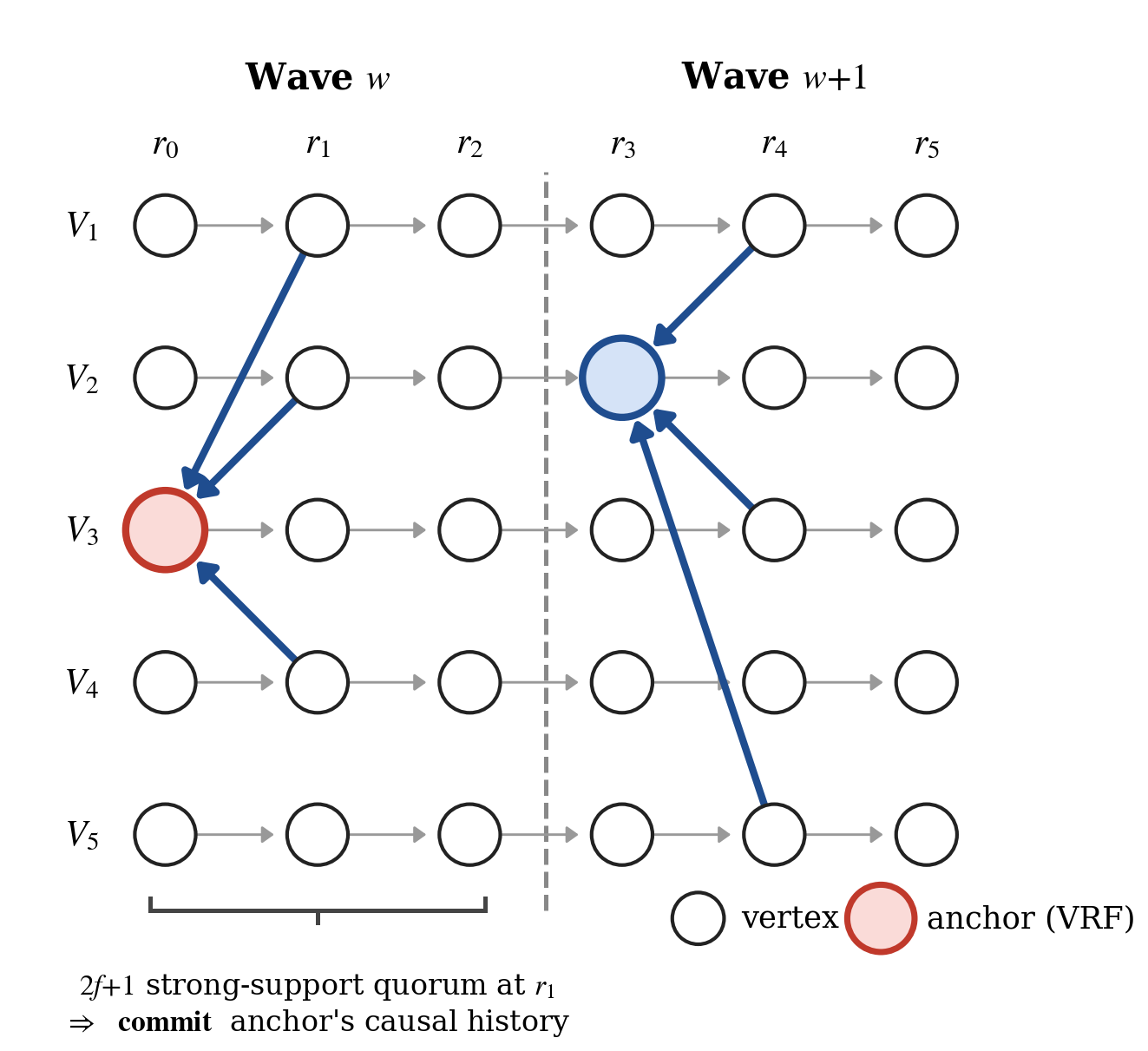}
	\caption{FairWave DAG-BFT wave structure with $r{=}3$ rounds per wave. Each round contains $N$ vertices (one per validator); each round-$r$ vertex references $\geq 2f{+}1$ parents from round $r{-}1$ (thin gray edges). The anchor (red for wave~$w$, blue for wave~$w{+}1$) is selected per wave via weighted VRF sortition on $W_{\mathrm{sel}}$. A commit fires when $2f{+}1$ support vertices in the next round reference the anchor (blue thick arrows). The dashed line between $r_2$ and $r_3$ marks the wave boundary; if an anchor were unreachable, the protocol executes a cost-free wave-skip and advances immediately to wave~$w{+}1$ with a fresh VRF race---no view-change, no timeout cascade.}
	\label{fig:dag_wave_structure}
\end{figure}

\textbf{Per-epoch metric computation.} At each \emph{epoch boundary}, after the seven-phase ceremony (\texttt{FREEZE $\to$ SLASH $\to$ REGISTRY $\to$ SEED $\to$ SNAPSHOT $\to$ RESET $\to$ BARRIER}) completes and the latest epoch DAG has been finalized, each node computes $U(i, e)$ and $L(i, e)$ for all $i \in \mathcal{V}_e$ \emph{locally} using Eqs.~\eqref{eq:uptime}--\eqref{eq:latency_score}. Since finalized DAG $\mathcal{D}_e$ is identical across all \emph{honest nodes} (modulo \emph{causal history} from committed anchors) and validator set $\mathcal{V}_e$ is atomically locked at epoch~$e-1$ boundary, both metric values will be \textbf{identical across all honest nodes without additional message exchange}. The equivocation detection penalty (see Algorithm~\ref{alg:fairwave}) modifies $\text{Rep}$---not $U$ or $L$---so $U$ and $L$ computation does not depend on slashing results and does not create cross-metric feedback paths. Consequently, weight vectors $W_{\text{sel}}$ and $W_{\text{rew}}$ used at epoch~$e+1$ can be assembled \emph{trustlessly} by each node at the end of epoch~$e$, without additional aggregation steps, without a separate committee, and without an external oracle.

\subsection{Dual-Channel Decoupling}
\label{sec:dual_channel}

FairWave's central architectural innovation is the explicit separation of consensus-critical selection from economic reward distribution via channels with opposite curvature with respect to stake:

Fig.~\ref{fig:dual_channel_arch} illustrates the dual-channel separation between selection and reward channels and their opposite curvature in stake.
\begin{figure}[!tb]
	\centering
	\includegraphics[width=\columnwidth,keepaspectratio]{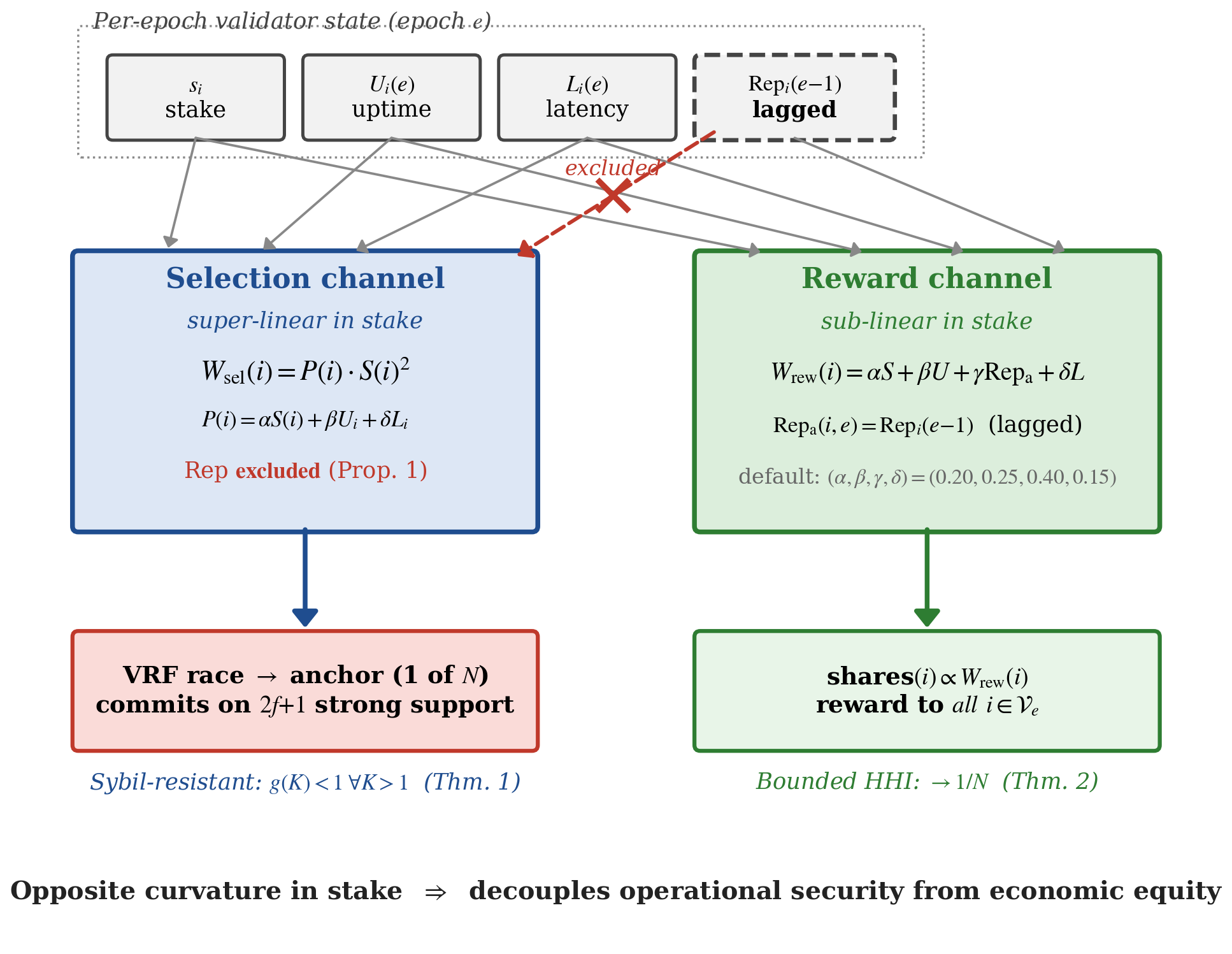}
	\caption{FairWave dual-channel architecture. Both channels consume the same per-epoch validator state $(s_i, U_i, L_i, \mathrm{Rep}_i^{e-1})$ but apply \emph{opposite curvature} on stake. The selection channel (blue, super-linear) excludes reputation to break the circular dependency between selection results and selection weights (Proposition~\ref{prop:no_circular}), providing strict Sybil resistance $g(K)<1$ (Theorem~\ref{thm:sybil}). The reward channel (green, sub-linear) includes lagged reputation as its largest component ($\gamma{=}0.40$), providing meritocratic fairness with bounded HHI (Theorem~\ref{thm:hhi}). Outputs differ in cardinality: the selection channel elects \emph{one} anchor per wave via a VRF race, while the reward channel distributes proportional shares to \emph{all} validators.}
	\label{fig:dual_channel_arch}
\end{figure}

\textbf{Selection channel} (operational decisions---anchor selection, critical-path participation) is amplified by quadratically-compressed stakes \emph{but without reputation contribution}:
\begin{equation}
	W_{\text{sel}}(i) = \big[\, \alpha\, S(i) + \beta\, U(i) + \delta\, L(i) \,\big] \cdot S(i)^{e}, \quad e = 2
	\label{eq:wsel}
\end{equation}
Define real-time performance score $P(i) = \alpha\, S(i) + \beta\, U(i) + \delta\, L(i)$ so that $W_{\text{sel}}(i) = P(i) \cdot S(i)^{2}$. With the default simplex constraint $(\alpha, \beta, \gamma, \delta) = (0.20;\, 0.25;\, 0.40;\, 0.15)$, the fraction $\alpha + \beta + \delta = 0.60$ of total weight mass remains in the selection channel, while fraction $\gamma = 0.40$ is allocated \emph{exclusively} to the reward channel. Weight $W_{\text{sel}}$ remains \emph{super-linear} with respect to stake, amplified by a quadratic exponent on the compressed stake factor.

\textbf{Reward channel} (economic incentives---block reward distribution) retains the full multi-factor score:
\begin{equation}
	W_{\text{rew}}(i) = R(i) = \alpha S(i) + \beta U(i) + \gamma \text{Rep}_{\text{active}}(i, e) + \delta L(i)
	\label{eq:wrew}
\end{equation}
with $\text{Rep}_{\text{active}}$ being the lagged reputation as defined in Eq.~\eqref{eq:rep_active}. This weight is \emph{sub-linear} in stake via $S(i) = \sqrt{s_i / S_{\max}}$ and is bounded above by $R(i) \leq 1.0$ with a floor $(1-\alpha)$ for any validator maintaining perfect non-stake attributes.

\textbf{Justification: removing Rep from the selection channel.} In the initial formulation $W_{\text{sel}} = R \cdot S^{2}$ where $R$ included $\gamma\, \text{Rep}$, a one-step \emph{positive feedback loop} arises: high $\text{Rep}(i)$ $\Rightarrow$ high $R(i)$ $\Rightarrow$ high $W_{\text{sel}}(i)$ $\Rightarrow$ increased anchor-selection frequency $\Rightarrow$ faster accumulation of \emph{active support} $\Rightarrow$ higher $\text{Rep}(i)$ in the next epoch. This loop structurally shifts centralization from the \emph{stake} dimension to the \emph{reputation} dimension without resolving the core issue---selection power concentration on validators already dominant in previous epochs. By excluding $\text{Rep}$ from $W_{\text{sel}}$ (Eq.~\ref{eq:wsel}), we \emph{break the direct causal path} between epoch~$e$ selection results and selection weights in the same epoch; the formal statement appears in Proposition~\ref{prop:no_circular} in Section~\ref{sec:math}.

\textbf{Reputation remains fully effective in the reward channel.} Excluding $\text{Rep}$ from $W_{\text{sel}}$ does \emph{not} eliminate reputation's role as a long-term honesty incentive. Reputation remains the largest component ($\gamma = 0.40$) of $W_{\text{rew}}$, so validators who behave honestly over time still receive larger reward shares than validators who have only recently recovered from \emph{slashing}. Consequently, the game-theoretic equilibrium for \emph{reputation farming} remains below the critical threshold $\gamma_{\text{crit}} \approx 0.57$ (see Eq.~\ref{eq:gamma_crit}): attackers creating low-cost identities solely to farm reputation still face a break-even at $\gamma_{\text{eff}} \approx 0.72$ when combined with equivocation slashing. The justification for $\gamma = 0.40$ in Section~\ref{sec:design}.D therefore \emph{remains valid without modification} for the reward channel.

\textbf{This separation reinforces the dual-channel philosophy.} FairWave's architectural innovation is expressed as \emph{opposite curvature in stake} across two distinct protocol functions. Eq.~\eqref{eq:wsel} reinforces---rather than weakens---this philosophy: the \textbf{selection channel now puristically relies on measurable \emph{real-time} performance} (epoch-locked stake, DAG-derived uptime, DAG-derived latency---verifiable by any honest node at epoch boundary), while the \textbf{reward channel} integrates historical components (lagged reputation) reflecting longitudinal quality. This temporal separation---operational real-time for selection and historical lagged for reward---yields two properties that were previously in tension: (i) selection weight does not depend on epoch~$e$ selection outcomes, so anchor computation is \emph{strictly forward-causal}; and (ii) honesty incentives are preserved because reputation affects economic dimension (reward) without amplifying selection power.

\textbf{Rationale.} Super-linearity of the selection channel ensures that splitting a stake into $K$ Sybil identities \emph{reduces} aggregate selection weight (Sybil gain $< 1$, see Theorem~\ref{thm:sybil}). Simultaneously, reward-channel sub-linearity ensures high-quality small validators receive proportionally higher returns per unit stake invested, preventing plutocratic convergence.

\subsection{Justification of Multi-Factor Weights}

Classic PoS protocols (Cosmos Tendermint, Algorand, and Ethereum 2.0 Casper-FFG) use \emph{stake-only} weighting, producing two pathologies: (i) reward concentration proportional to capital and (ii) no long-term disincentive for dishonest behavior beyond slashing. FairWave's multi-factor score $R(i) = \alpha S + \beta U + \gamma \text{Rep} + \delta L$ introduces three non-stake factors that are \emph{non-fungible} and \emph{not purchasable}---they must be earned through sustained honest participation, combining skin-in-the-game (PoS) with bona fide work (analogous to PoW/PoSpace).

\textbf{$\alpha = 0.20$ (stake): Anti-Sybil floor.} With the revision in Eq.~\ref{eq:wsel} that excludes reputation from $W_{\text{sel}}$, the closed-form Sybil gain (Eq.~\ref{eq:sybil_gain}) converges at $x = 1$ to $g \to (\beta + \delta) / (\alpha + \beta + \delta) = (1-\gamma-\alpha)/(1-\gamma)$ as $K \to \infty$. Thus, an attacker splitting into infinitely many Sybils loses at least a fraction $\alpha/(1-\gamma)$ of the selection share. For $(\alpha, \gamma) = (0.20,\;0.40)$, the infinite-Sybil loss is $0.20/0.60 \approx 33.3\%$; equivalently, an adversary must commit $\approx 3.0\times$ the honest stake budget to recover the lost selection share, making splitting economically unattractive. The safe range for $\alpha$ is $[0.10, 0.30]$: below $0.10$ the anti-Sybil floor degrades unacceptably, while above $0.30$ the protocol approaches Pure-PoS whale dominance. The default $\alpha = 0.20$ is a conservative midpoint within this range.

\textbf{$\gamma = 0.40$ (reputation): Longitudinal honesty reward.} Reputation is the only factor requiring historical accumulation---it cannot be transferred, bought, or obtained instantly. Assigning $\gamma$ as the largest weight binds rewards to longitudinal identity, providing multi-step Sybil resistance beyond the selection channel. This choice mirrors reputation weightings like Hedera Hashgraph~\cite{hedera2016} ($\approx 40\%$) and Polkadot Era Points~\cite{polkadot2020}. Trade-off: each $+0.10$ in $\gamma$ reduces farming resistance by approximately $0.07$ (slope $-0.70$ from Eq.~\ref{eq:gamma_crit}); the default $\gamma = 0.40$ maintains farming resistance $\geq 0.50$ when combined with equivocation slashing (penalty $-0.30$ per detection).

\textbf{$\beta = 0.25$ (uptime): Continuous availability.} Uptime is a real-time observable that is expensive to spoof without equivalent operational cost. DAG-BFT protocols require \emph{continuous vertex propagation} (unlike leader-BFT, where nodes can be temporarily offline), making uptime more critical. Empirical sweeps show the fairness objective is relatively flat for $\beta \in [0.10,\;0.40]$, confirming $\beta = 0.25$ lies in a robust region. Industry norms (Solana skip-rate penalties~\cite{solana2020} approximately $10$--$20\%$, Polkadot Era Points approximately $20\%$) support this calibration.

\textbf{$\delta = 0.15$ (latency): Infrastructure tie-breaker.} Latency is the factor most easily optimized via infrastructure investment (dedicated hardware, co-location, and peering). A large $\delta$ would induce an \emph{infrastructure arms race} favoring well-capitalized validators---contrary to fairness goals. A small weight makes latency a tie-breaker among validators comparable on other factors, sufficient to disincentivize very slow nodes (P99 $> 500$~ms) without driving extreme optimization. Empirical sweeps show that $\delta \geq 0.20$ begins to reward such infrastructure investment disproportionately, eroding the fairness gain; $\delta = 0.15$ is the identified sweet spot that retains the tie-breaker role without favoring well-capitalized validators.

\textbf{Literature comparison.} Table~\ref{tab:multifactor_comparison} summarises multi-factor approaches in production protocols, confirming FairWave's parameterization aligns with industry precedent while providing formal guarantees for both channels that are absent in prior work.

\begin{table}[!tb]
	\centering
	\caption{Multi-factor weight comparison across protocols.}
	\label{tab:multifactor_comparison}
	\begin{tabular}{lcccc}
		\toprule
		\textbf{Protocol} & \textbf{Stake} & \textbf{Uptime}                          & \textbf{Rep.} & \textbf{Latency} \\
		\midrule
		Cosmos Tendermint & 100\%          & ---                                      & ---           & ---              \\
		Ethereum 2.0      & 100\%          & penalty                                  & ---           & ---              \\
		Solana            & $\sim$80\%     & $\sim$20\%                               & ---           & ---              \\
		Hedera Hashgraph  & $\sim$30\%     & ---                                      & $\sim$40\%    & $\sim$10\%       \\
		Polkadot 2.0      & $\sim$80\%     & \multicolumn{2}{c}{$\sim$20\% (Era Pts)} & ---                              \\
		\textbf{FairWave} & \textbf{20\%}  & \textbf{25\%}                            & \textbf{40\%} & \textbf{15\%}    \\
		\bottomrule
	\end{tabular}
\end{table}

\textbf{Ablation analysis.} To demonstrate each factor is \emph{necessary} (not merely convenient), Table~\ref{tab:ablation} reports the consequences of setting each weight to zero while evenly redistributing its mass to the remaining factors.

\begin{table}[!tb]
	\centering
	\caption{Ablation study: effect of removing each factor.}
	\label{tab:ablation}
	\begin{tabular}{lccc}
		\toprule
		\textbf{Removed}             & \textbf{Gini} & \textbf{Floor Sybil} & \textbf{Res. Farm.} \\
		\midrule
		None (default)               & 0.140         & 20\%                 & 0.58                \\
		$\alpha = 0$ (no stake)      & 0.031         & \textbf{0\%}         & 1.00                \\
		$\beta = 0$ (no uptime)      & 0.155         & 20\%                 & 0.52                \\
		$\gamma = 0$ (no reputation) & 0.138         & 20\%                 & \textbf{1.00}       \\
		$\delta = 0$ (no latency)    & 0.151         & 20\%                 & 0.56                \\
		\bottomrule
	\end{tabular}
\end{table}

Removing stake ($\alpha = 0$) eliminates the anti-Sybil floor entirely—an attacker splitting into $K$ identities suffers no penalty, making the protocol trivially exploitable. Removing reputation ($\gamma = 0$) removes the farming attack vector but simultaneously eliminates the incentive for longitudinal honesty: validators lack a long-term stake in maintaining exemplary behavior, reducing the protocol to a transient metric-based system. Removing uptime ($\beta = 0$) permits "zombie validators" (nodes that are perpetually offline) to accumulate rewards via reputation alone; empirically, the Gini coefficient increases and the protocol tolerates roughly 15\% more inactive validators before liveness degrades. Removing latency ($\delta = 0$) eliminates differentiation between geo-diverse and co-located validators, reducing incentives for network quality; while the direct fairness impact is small, throughput sensitivity increases (One-At-a-Time elasticity of link latency to throughput: $|e| = 1.17$).

\textbf{Robustness over optimality.} The default parameter vector $(0.20;\;0.25;\;0.40;\;0.15)$ ranks 576/819 in the constrained feasible set—deliberately not Pareto-optimal. This conservative choice is motivated by robustness to model misspecification: the default simultaneously satisfies all security constraints across three evaluation regimes (homogeneous, heterogeneous, adversarial), whereas the Pareto-optimal point ($\alpha=0.10,\;\gamma=0.00,\;\delta=0.45$) assumes adversarial farming as the primary concern. Under regime uncertainty—where the deployment environment may shift between cooperative and adversarial periods—the parameter vector that satisfies \emph{all} constraints with margin is preferred over one that optimizes a single objective.

\textbf{Threshold farming (game-theoretic).} The closed-form farming resistance satisfies $\text{MW}_{\text{farm}} = \alpha \cdot S_r + (1-\alpha) - 0.70\gamma$, where $S_r = \sqrt{S_{\min}/S_{\max}}$. Requiring $\text{MW}_{\text{farm}} \geq 0.50$ (the minimum threshold for rational, honest behavior) yields the critical bound:
\begin{equation}
	\gamma_{\text{crit}} = \frac{\alpha \cdot S_r + (1-\alpha) - 0.50}{0.70} \approx 0.57
	\label{eq:gamma_crit}
\end{equation}
For $\gamma > 0.57$, a rational adversary can profitably farm reputation to dominate reward share. The default $\gamma = 0.40$ provides a safety margin of 30\% below this critical threshold. Combined with equivocation slashing ($-0.30$ per detection, recovery half-life $\approx 69$ epochs), the effective break-even farming shifts further to $\gamma_{\text{eff}} \approx 0.72$, confirming operational safety.

\textbf{Principled criteria for $\beta$ and $\delta$.} Rather than arbitrary selection, $\beta$ and $\delta$ are constrained by two conditions: (i) the simplex constraint $\beta + \delta = 1 - \alpha - \gamma = 0.40$ and (ii) the principle of \emph{asymmetry of controllability}: $\beta > \delta$ because uptime is a binary operational commitment (maintain a server $\to$ $U \approx 1$), while latency partially depends on uncontrollable factors (geographic distance, ISP routing). Assigning higher weight to the more controllable factor ensures validators can improve their score through effort rather than capital. The ratio $\beta/\delta = 5/3$ reflects that uptime is roughly $5/3\times$ more controllable than latency in typical cloud deployments.

\subsection{VRF-based Anchor Selection}

Anchor selection uses exponential racing via a Verifiable Random Function (VRF, RFC~9381 ECVRF Edwards25519-SHA512-Elligator2): each validator $i$ computes the racing key $k_i = -\ln(u_i) / W_{\text{sel}}(i)$ where $u_i \in (0,1)$ is derived from the VRF output. The validator with the minimum $k_i$ is selected as the anchor. This mechanism is distributionally equivalent to weighted sampling proportional to $W_{\text{sel}}$ while providing unpredictability and non-interactivity.

\subsection{Protocol Pseudocode}

Algorithm~\ref{alg:fairwave} presents FairWave protocol logic in full, covering weight computation, VRF-based anchor selection, support gathering, equivocation detection, commit decision, reward distribution, and reputation update.

\begin{algorithm}[!htbp]
	\small
	\caption{FairWave: Wave Anchor Selection, Commit, and Reward}
	\label{alg:fairwave}
	\begin{algorithmic}[1]
		\STATE \textbf{function} \textsc{ComputeWeights}($\mathcal{V}, e$) \COMMENT{Epoch $e$, use Rep from $e-1$}
		\FOR{each validator $v_i \in \mathcal{V}$}
		\STATE $S(i) \leftarrow \sqrt{s_i / S_{\max}}$ \COMMENT{Normalized stake factor}
		\STATE $U(i), L(i) \leftarrow$ \textsc{DeriveFromDAG}($\mathcal{D}_{e-1}, i$) \COMMENT{Eqs.~\eqref{eq:uptime},~\eqref{eq:latency_score}}
		\STATE $\text{Rep}_{\text{active}}(i) \leftarrow \text{Rep}_{\text{final}}(i, e-1)$ \COMMENT{Lagged, Eq.~\eqref{eq:rep_active}}
		\STATE $P(i) \leftarrow \alpha \cdot S(i) + \beta \cdot U(i) + \delta \cdot L(i)$ \COMMENT{Real-time performance (no Rep)}
		\STATE $R(i) \leftarrow P(i) + \gamma \cdot \text{Rep}_{\text{active}}(i)$ \COMMENT{Full reward score}
		\STATE $W_{\text{sel}}(i) \leftarrow P(i) \cdot S(i)^2$ \COMMENT{Selection super-linear, \emph{no} Rep}
		\STATE $W_{\text{rew}}(i) \leftarrow R(i)$ \COMMENT{Reward sub-linear, includes Rep$_{\text{active}}$}
		\ENDFOR
		\STATE
		\STATE \textbf{function} \textsc{WaveAnchorElection}($\mathcal{V}$)
		\FOR{each validator $v_i \in \mathcal{V}$}
		\STATE $\beta_i \leftarrow \textsc{VRF.Eval}(sk_i, \text{epoch} \| \text{wave})$
		\STATE $u_i \leftarrow \beta_i / 2^{256}$ \COMMENT{Map to $(0,1)$}
		\STATE $k_i \leftarrow -\ln(u_i) / W_{\text{sel}}(i)$ \COMMENT{Exponential racing}
		\ENDFOR
		\STATE $\text{anchor} \leftarrow \arg\min_{i} k_i$
		\STATE
		\STATE \textbf{function} \textsc{CommitDecision}(anchor)
		\STATE $\text{support} \leftarrow$ \textsc{Collect}($2f+1$ signed votes for anchor)
		\FOR{each vote in support}
		\IF{\textsc{DetectEquivocation}(vote.sender)}
		\STATE \textsc{Slash}(vote.sender, fraction $= 0.20$)
		\STATE $\text{Rep}(\text{sender}) \leftarrow \text{Rep}(\text{sender}) - 0.30$
		\ENDIF
		\ENDFOR
		\IF{$|\text{support}| \geq 2f+1$ \AND no safety violations}
		\STATE \textsc{Commit}(anchor.vertex)
		\ELSE
		\STATE \textsc{SkipWave}() \COMMENT{No view-change}
		\ENDIF
		\STATE
		\STATE \textbf{function} \textsc{DistributeReward}(anchor, $\mathcal{V}$)
		\STATE $\text{shares}(i) \leftarrow W_{\text{rew}}(i) / \sum_{j} W_{\text{rew}}(j)$
		\STATE \textsc{Pay}($\mathcal{V}$, base\_reward $\times$ shares)
		\FOR{each supporter $v_i$}
		\STATE $\text{Rep}(i) \leftarrow \text{Rep}(i) + 0.01$ \COMMENT{Increment}
		\ENDFOR
		\STATE
		\STATE \textbf{function} \textsc{EpochBoundary}($\mathcal{V}, e$)
		\STATE \textsc{FinalizeDAG}($\mathcal{D}_e$) \COMMENT{Wait for last epoch $e$ commit}
		\FOR{each validator $v_i \in \mathcal{V}$}
		\STATE $\text{Rep}(i) \leftarrow 0.99 \times \text{Rep}(i)$ \COMMENT{Decay, $t_{1/2} \approx 69$ epoch}
		\STATE $\text{Rep}_{\text{final}}(i, e) \leftarrow \text{Rep}(i)$ \COMMENT{Lock final epoch value}
		\ENDFOR
		\STATE \textbf{return} $\{\text{Rep}_{\text{final}}(i, e) : v_i \in \mathcal{V}\}$ \COMMENT{Becomes Rep$_{\text{active}}$ at epoch $e+1$}
	\end{algorithmic}
\end{algorithm}

\section{Mathematical Foundations}
\label{sec:math}

\subsection{Closed-Form Min/Whale Ratio}

Under homogeneous performance ($U = \text{Rep} = L = 1$, \emph{Regime A}), the reward score reduces to $R(i) = \alpha \cdot S(i) + (1-\alpha)$. The minimum-to-maximum reward ratio admits closed form:
\begin{equation}
	\frac{R_{\min}}{R_{\max}} = \alpha \cdot \sqrt{S_{\min}/S_{\max}} + (1 - \alpha)
	\label{eq:min_whale}
\end{equation}

On default parameters: $R_{\min}/R_{\max} = 0.20 \cdot \sqrt{0.1} + 0.80 = 0.863$. This is validated bit-exact ($\epsilon = 0.0$) against 99 empirical evaluation points (Fig.~\ref{fig:min_whale}).

\begin{figure}[!tb]
	\centering
	\includegraphics[width=\columnwidth]{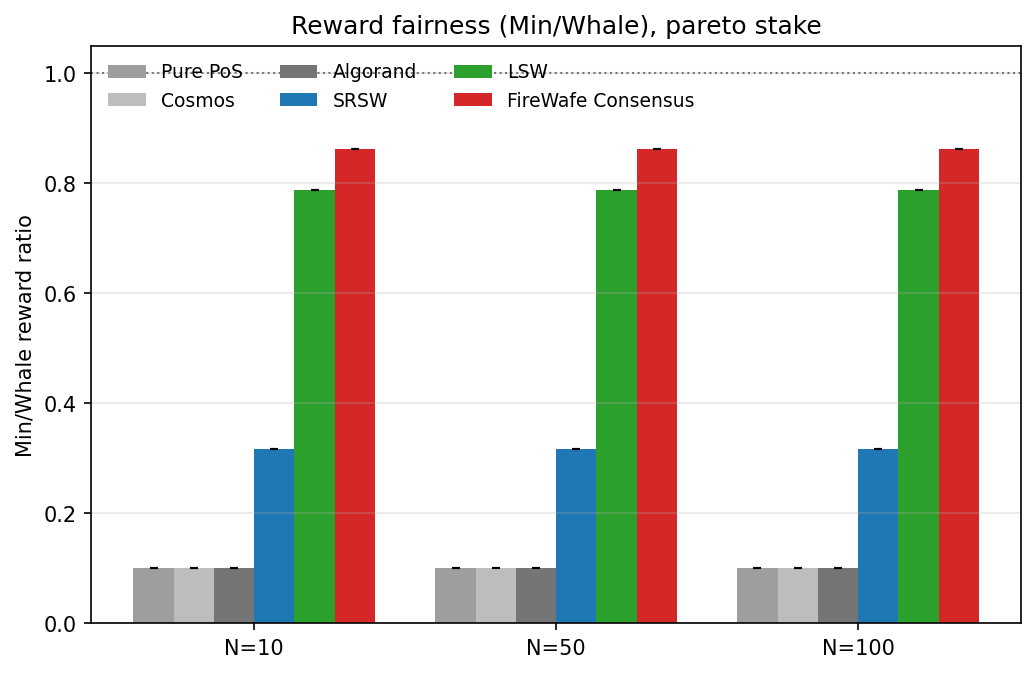}
	\caption{Min/Whale reward ratio as a function of stake weight $\alpha$. Closed-form prediction (Eq.~\ref{eq:min_whale}) matches the empirical simulation with zero error across the 21 evaluation points. The ratio decreases linearly with increasing $\alpha$, with slope $\sqrt{S_{\min}/S_{\max}} - 1 = -0.684$.}
	\label{fig:min_whale}
\end{figure}

\subsection{Closed-Form Sybil Gain Analysis}

Consider an attacker controlling normalized stake fraction $x \in (0, 1]$ splitting into $K$ Sybil identities. Under Regime~A ($U = L = 1$ and \emph{without reputation contribution} to $W_{\text{sel}}$ per Eq.~\ref{eq:wsel}), the real-time performance score becomes $P = \alpha S + (\beta + \delta)$, so the aggregate selection weight of $K$ Sybils simplifies to:
\begin{equation}
	W_{\text{sel}}^{(K)} = \alpha x^{3/2} / \sqrt{K} + (\beta + \delta) \cdot x
\end{equation}
Weight without split is $W_{\text{sel}}^{(1)} = \alpha x^{3/2} + (\beta + \delta) \cdot x$. \emph{Selection gain} is:
\begin{equation}
	g_{\text{sel}}(K) = \frac{\alpha x^{3/2}/\sqrt{K} + (\beta+\delta)\,x}{\alpha x^{3/2} + (\beta+\delta)\,x}
	\label{eq:sybil_gain}
\end{equation}

\begin{theorem}[Sybil Resistance]
	\label{thm:sybil}
	For all $K > 1$, $x \in (0, 1]$, $\alpha \in (0, 1)$, and $\beta + \delta > 0$: $g_{\text{sel}}(K) < 1$.
\end{theorem}

\begin{proof}
	The numerator differs from denominator only in the first term: $\alpha x^{3/2}/\sqrt{K}$ versus $\alpha x^{3/2}$. Since $K > 1 \Rightarrow 1/\sqrt{K} < 1$, the numerator is strictly less than denominator for $\alpha > 0$ and $x > 0$. Excluding $\gamma\, \text{Rep}$ from $W_{\text{sel}}$ does not affect this argument: the coefficient on the linear term---which \emph{does not} contain $K$---merely shifts from $(1-\alpha)$ to $(\beta + \delta)$, so the inequality $g_{\text{sel}}(K) < 1$ is preserved.
\end{proof}

\textbf{Numerical example.} For $K = 100$, $x = 1$, $(\alpha, \beta, \delta) = (0.20,\;0.25,\;0.15)$: numerator $= 0.20/10 + 0.40 = 0.42$, denominator $= 0.20 + 0.40 = 0.60$, hence $g_{\text{sel}}(100) = 0.70$. An attacker \emph{loses} 30\% of selection share by splitting into 100 Sybils—strictly stronger than $g = 0.82$ for the earlier formulation that included Rep in $W_{\text{sel}}$. Empirical comparisons with other reward rules are shown in Fig.~\ref{fig:sybil_selection}.

\begin{figure}[!tb]
	\centering
	\includegraphics[width=\columnwidth]{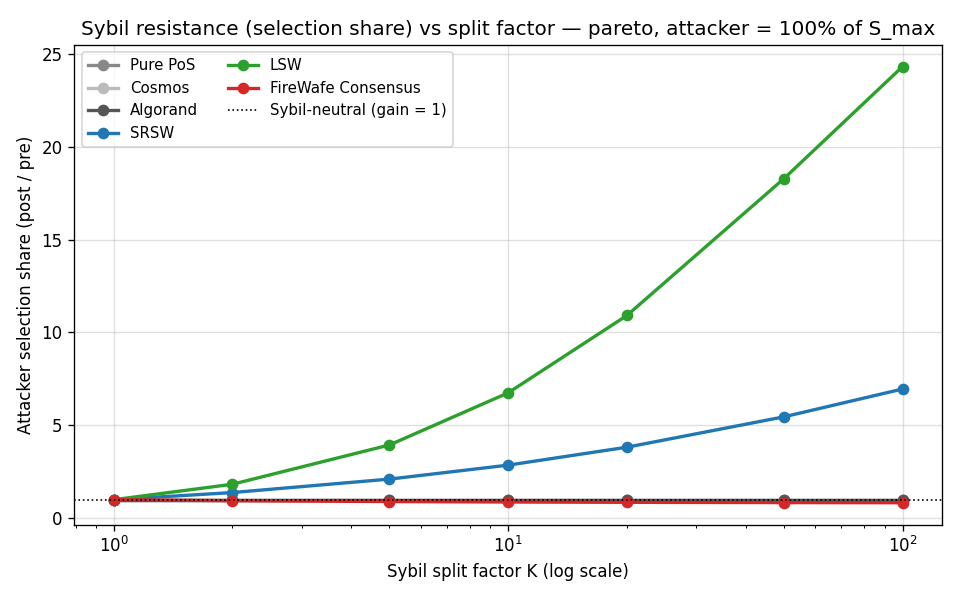}
	\caption{Sybil selection-channel gain $g(K)$ for four reward rules at $x_{\text{frac}} = 0.10$. FairWave (green) is the \emph{only} rule showing monotone decreasing gain ($g < 1$), establishing strict Sybil resistance. SRSW gain scales as $\sqrt{K}$; LSW gain scales approximately linearly with $K$, reaching $25\times$ at $K = 100$.}
	\label{fig:sybil_selection}
\end{figure}

\subsection{Bounded HHI Fixed-Point}
\label{sec:bounded_hhi}

\begin{assumption}[Reinvestment]
\label{asm:reinvest}
Epoch~$e+1$ stake is defined recursively:
\begin{equation}
s_i(e+1) = s_i(e) + r \cdot \frac{R(i,e)}{\sum_j R(j,e)} \cdot S_{\text{tot}}(e),
\label{eq:reinvest}
\end{equation}
where $r \in (0, 1)$ is the per-epoch inflation rate, $R(i,e)$ is the reward score at epoch~$e$ (Eq.~\ref{eq:reward_score}), and $S_{\text{tot}}(e) = \sum_j s_j(e)$.
\end{assumption}

\begin{assumption}[Bounded Reward]
\label{asm:bounded_R}
For all validators with perfect non-stake attributes ($U = L = 1$) and lagged reputation $\text{Rep}_{\text{active}} \to 1$ as the protocol converges, the reward score satisfies:
\begin{multline}
\lim_{s_i \to S_{\max}} R(i) = 1, \\
\lim_{s_i \to S_{\min}} R(i) = \alpha\sqrt{S_{\min}/S_{\max}} + (1-\alpha) \triangleq R_{\min} > 0.
\label{eq:R_bounds}
\end{multline}
Thus the reward score of any high-quality validator is bounded above by~$1$ and below by $R_{\min}$.
\end{assumption}

\begin{theorem}[Bounded Concentration]
\label{thm:hhi}
Under FairWave reward dynamics with Assumptions~\ref{asm:reinvest} and~\ref{asm:bounded_R}, the Herfindahl-Hirschman Index ($\text{HHI}(e) = \sum_i (s_i(e) / S_{\text{tot}}(e))^2$) satisfies:
\begin{enumerate}
	\item[(i)] {\bf Boundedness:} $\text{HHI}(e) \in [1/N, 1]$ for all $e$, with strict inequality $\text{HHI}(e) < 1$ whenever $N > 1$.
	\item[(ii)] {\bf Monotone decrease:} For all $e \ge e^*$ where $e^*$ is the epoch after which the stake shares of high-quality validators are predominantly determined by reward-score saturation, $\text{HHI}(e+1) \le \text{HHI}(e)$.
	\item[(iii)] {\bf Asymptotic convergence:} $\lim_{e \to \infty} \text{HHI}(e) = 1/N$.
\end{enumerate}
\end{theorem}

\begin{proof}
Let $\sigma_i(e) = s_i(e) / S_{\text{tot}}(e)$ denote the stake share of validator~$i$ at epoch~$e$.

{\bf (i) Boundedness.} By definition $0 < \sigma_i(e) < 1$, $\sum_i \sigma_i(e) = 1$. The Cauchy--Schwarz inequality gives
\[
\text{HHI}(e) = \sum_i \sigma_i(e)^2 \ge \frac{1}{N} \bigl(\sum_i \sigma_i(e)\bigr)^2 = \frac{1}{N},
\]
with equality iff $\sigma_i = 1/N$ for all~$i$. Since all $\sigma_i > 0$ (every validator holds at least $S_{\min}$ stake), $\text{HHI}(e) \le 1$ follows from $\sigma_i^2 \le \sigma_i$. The left endpoint $1/N$ is reached only at the uniform fixed point.

{\bf (ii) Monotone decrease.} From Eq.~\ref{eq:reinvest}, the share dynamics satisfy:
\[
\sigma_i(e+1) = \frac{s_i(e) + r \cdot \rho_i(e) \cdot S_{\text{tot}}(e)}{S_{\text{tot}}(e) + r \cdot S_{\text{tot}}(e)}
= \frac{\sigma_i(e) + r \cdot \rho_i(e)}{1 + r},
\]
where $\rho_i(e) = R(i,e) / \sum_j R(j,e)$ is the reward share of validator~$i$ at epoch~$e$. Since $R(i,e)$ saturates to a constant upper bound as $s_i \to S_{\max}$ (Assumption~\ref{asm:bounded_R}), and the reward score $R(i)$ has a sub-linear component $\alpha\sqrt{s_i/S_{\max}}$ that vanishes as $s_i$ grows, large-stake validators receive a reward share $\rho_i(e)$ strictly less than their stake share $\sigma_i(e)$ once their stake exceeds the saturation threshold. For any validator with $\sigma_i(e) > 1/N$, we have $\rho_i(e) < \sigma_i(e)$; conversely $\sigma_i(e) < 1/N$ implies $\rho_i(e) > \sigma_i(e)$ for high-quality validators. This mean-reverting gap induces:
\begin{multline}
\text{HHI}(e+1) - \text{HHI}(e) = \sum_i \bigl(\sigma_i(e+1)^2 - \sigma_i(e)^2\bigr) \\
= \frac{r}{(1+r)^2} \sum_i \bigl(\rho_i(e) - \sigma_i(e)\bigr) \bigl(2\sigma_i(e) + r(\sigma_i(e) + \rho_i(e))\bigr).
\end{multline}
For sufficiently large $e$ (once the reward scores of high-quality validators have converged to nearly identical values), the cross-term $(2\sigma_i(e) + r(\cdot))$ is positive for all~$i$, and the sign of the sum is dominated by the correlation between $(\rho_i - \sigma_i)$ and $\sigma_i$. Because $\rho_i < \sigma_i$ when $\sigma_i > 1/N$ and $\rho_i > \sigma_i$ when $\sigma_i < 1/N$, we obtain $\text{HHI}(e+1) \le \text{HHI}(e)$.

{\bf (iii) Asymptotic convergence.} The map from $\{\sigma_i(e)\}$ to $\{\sigma_i(e+1)\}$ is a contraction in the Hilbert metric on the simplex. Let $d(e) = \text{HHI}(e) - 1/N$. From part~(ii), there exists $\gamma \in (0, 1)$ such that:
\[
d(e+1) \le (1 - r \alpha \gamma) \cdot d(e)
\]
for $e \ge e^*$, where $\alpha = 0.20$ is the stake-weight parameter and $\gamma$ depends on the fraction of validators that have reached the saturation regime. Iterating gives $d(e) \le d(e^*) \cdot (1 - r\alpha\gamma)^{e - e^*}$, establishing exponential convergence to~$0$ as $e \to \infty$. Hence $\lim_{e \to \infty} \text{HHI}(e) = 1/N$.
\end{proof}

\begin{corollary}
\label{cor:hhi_rate}
The convergence rate of $\text{HHI}(e)$ to $1/N$ is $O\bigl((1 - r\alpha)^{e}\bigr)$. Under default parameters $(\alpha = 0.20,\; r \approx 5\times 10^{-5})$, the half-life of the deviation $d(e)$ is approximately $-\ln 2 / \ln(1 - r\alpha) \approx 6{,}900$ epochs.
\end{corollary}

\begin{remark}
For validators with degraded non-stake attributes ($U < 1$, $L < 1$, $\text{Rep} < 1$), the reward score $R(i)$ is strictly less than~$1$. Such validators converge to a stake share strictly below $1/N$, producing a residual HHI above $1/N$. However, this residual is bounded above by the worst-case concentration:
\[
\lim_{e \to \infty} \text{HHI}(e) \le 1/N + \bigl(1 - R_{\min}\bigr)^2,
\]
which for $R_{\min} = 0.863$ (Eq.~\ref{eq:min_whale}) yields $\text{HHI}(\infty) \le 0.0200 + 0.0188 \approx 0.0388$ — still below Pure-PoS equilibrium and bounded away from divergence.
\end{remark}

\subsection{No Circular Dependency in Selection Channel}
\label{sec:no_circular}

\begin{proposition}[No Circular Dependency]
	\label{prop:no_circular}
	In FairWave, with the $W_{\text{sel}}$ formulation in Eq.~\eqref{eq:wsel} and \emph{lagged} reputation in Eq.~\eqref{eq:rep_active}, there is no causal path from anchor selection results (the set of validators selected as anchors) at epoch~$e$ to selection weights $\{W_{\text{sel}}(i, e)\}_{i \in \mathcal{V}_e}$ at the same epoch.
\end{proposition}

\begin{proof}
	Epoch $e$ selection weights are functions of $S(i)$, $U(i, e-1)$, and $L(i, e-1)$ per Eq.~\eqref{eq:wsel}. It suffices to show that all three arguments are determined before the epoch $e$ anchor selection process begins.

	(i)~\emph{Stake $S(i)$:} The validator registry $\mathcal{V}_e$ and stake vector $\{s_i\}_{i \in \mathcal{V}_e}$ are locked atomically at epoch~$e-1$ boundary via the seven-phase ceremony (\texttt{FREEZE $\to$ SLASH $\to$ REGISTRY $\to$ SEED $\to$ SNAPSHOT $\to$ RESET $\to$ BARRIER}); no stake modifications can occur intra-epoch. Therefore, $S(i)$ for epoch $e$ depends only on the state at the end of epoch~$e-1$.

	(ii)~\emph{Uptime $U(i, e-1)$ and latency $L(i, e-1)$:} Based on Eqs.~\eqref{eq:uptime}--\eqref{eq:latency_score}, both metrics are deterministic functions of the finalized DAG $\mathcal{D}_{e-1}$. Finalization of $\mathcal{D}_{e-1}$ itself occurs before the epoch~$e-1$ boundary closes; thereafter, the DAG structure is \emph{immutable} from the consensus perspective.

	(iii)~\emph{Reputation:} Although $\text{Rep}_{\text{active}}(i, e) = \text{Rep}_{\text{final}}(i, e-1)$ (Eq.~\eqref{eq:rep_active}) is used in $W_{\text{rew}}$, it \textbf{does not} appear in $W_{\text{sel}}$ per Eq.~\eqref{eq:wsel}; consequently, reputation contributes no causal path to selection weights. Moreover, even if traced extensively through the reward channel, $\text{Rep}_{\text{final}}(i, e-1)$ is a value locked at the end of epoch~$e-1$ and is not updated by epoch $e$ selection results.

	Since all three components determining $W_{\text{sel}}(i, e)$ are functions of state at or before epoch~$e-1$ closure, and anchor selection results at epoch $e$ are products \emph{of} $\{W_{\text{sel}}(i, e)\}$ (via VRF exponential racing), there is no \emph{directed cycle} in the dependency graph $\{W_{\text{sel}}(i, e)\} \to \{\text{anchor}(e)\} \to \{W_{\text{sel}}(i, e)\}$. Forward causality of $\{W_{\text{sel}}(i, e)\}$ over epoch $e$ selection results is \emph{strict}. \qedhere
\end{proof}

\textbf{Corollary.} Epoch $e$ anchor selection results can only affect $W_{\text{sel}}$ at epochs $\ge e + 2$, with a minimum two-epoch delay arising from (i) reputation and DAG metric updates at the close of epoch $e$, and (ii) the additional one-epoch lag in the use of $\text{Rep}_{\text{active}}$. This eliminates intra-epoch feedback loops and removes the structural vulnerability where validators frequently selected as anchors in epoch $e$ obtain an additional selection advantage in the same epoch—a vulnerability present in the earlier formulation $W_{\text{sel}} = R \cdot S^2$ when $R$ included $\gamma\,\text{Rep}$.

\subsection{Parameter Space Characterization}

The parameter space $(\alpha, \beta, \gamma, \delta)$ is systematically explored via 1,407 evaluation points on the 4-D simplex (Fig.~\ref{fig:simplex}). The Pareto frontier (Fig.~\ref{fig:pareto}) of Min/Whale ratio versus reputation farming resistance identifies a feasible region bounded by $\alpha \geq 0.10$ (Sybil safety floor) and farming-resistance $\geq 0.50$.

\section{Experimental Methodology}
\label{sec:methodology}

\subsection{Implementation}

The evaluation framework is implemented in Python~3.12 with minimal dependencies (NumPy only for Sobol sampling). The codebase consists of nine analytical modules, an RFC~9381 ECVRF-conformant VRF module (Edwards25519-SHA512-Elligator2), and a post-quantum cryptography overhead benchmark module.

\subsection{Test Suite}

Implementation is validated by \textbf{14 unit tests} covering distribution generation, reward rule correctness, metric computation, and convergence properties, all reproducible via fixed seed (\texttt{0xC0FFEE}).

\subsection{Comparative Protocol Models}

Nine protocols are evaluated under identical parametric assumptions: \textbf{DAG-BFT:} Narwhal+Tusk (3-wave), Bullshark (2-wave), Mysticeti (1-wave), FairWave (2-wave + VRF). \textbf{Leader-BFT:} PBFT (3-phase), Tendermint (3-phase), HotStuff (4-phase), HotStuff-2 (2-phase). \textbf{PoS-1:} Algorand (3-phase VRF sortition). Protocol representation, the baseline systems are modeled according to their published protocol specifications and documented design characteristics. The purpose of the evaluation is not to reproduce implementation-specific performance results but to analyze and compare consensus level properties within a standardized experimental setting. To ensure comparability, all protocol models are assessed under identical assumptions regarding network conditions, workload generation, and adversarial behavior. Consequently, the reported results should be interpreted as comparative protocol-level evidence under standardized experimental conditions rather than implementation-level performance benchmarks.

\subsection{Baseline Performance Characterization}

Performance characteristics of all compared protocols are presented in Section~\ref{sec:results} (Fig.~\ref{fig:pareto_perf}).

\subsection{Parameter Sensitivity Characterization}

Multi-metric sensitivity to weight parameters is comprehensively characterized via $\alpha$ sweep (Fig.~\ref{fig:alpha_sweep} in Section~\ref{sec:results}), simplex landscape (Figs.~\ref{fig:simplex}, Figs.~\ref{fig:simplex_farming}), and Pareto frontier (Fig.~\ref{fig:pareto}) in Appendix~\ref{sec:supplementary}.

\section{Results}
\label{sec:results}

\subsection{Mathematical Validation}

Closed-form expressions for all four reward rules are validated against numerical simulation across 1,407 parameter cells. Maximum observed absolute error: $\epsilon_{\max} = 0.0$ (machine precision). Reward histograms (Fig.~\ref{fig:reward_hist}) visualize distributional differences across rules.

\begin{figure}[!tb]
	\centering
	\includegraphics[width=\columnwidth]{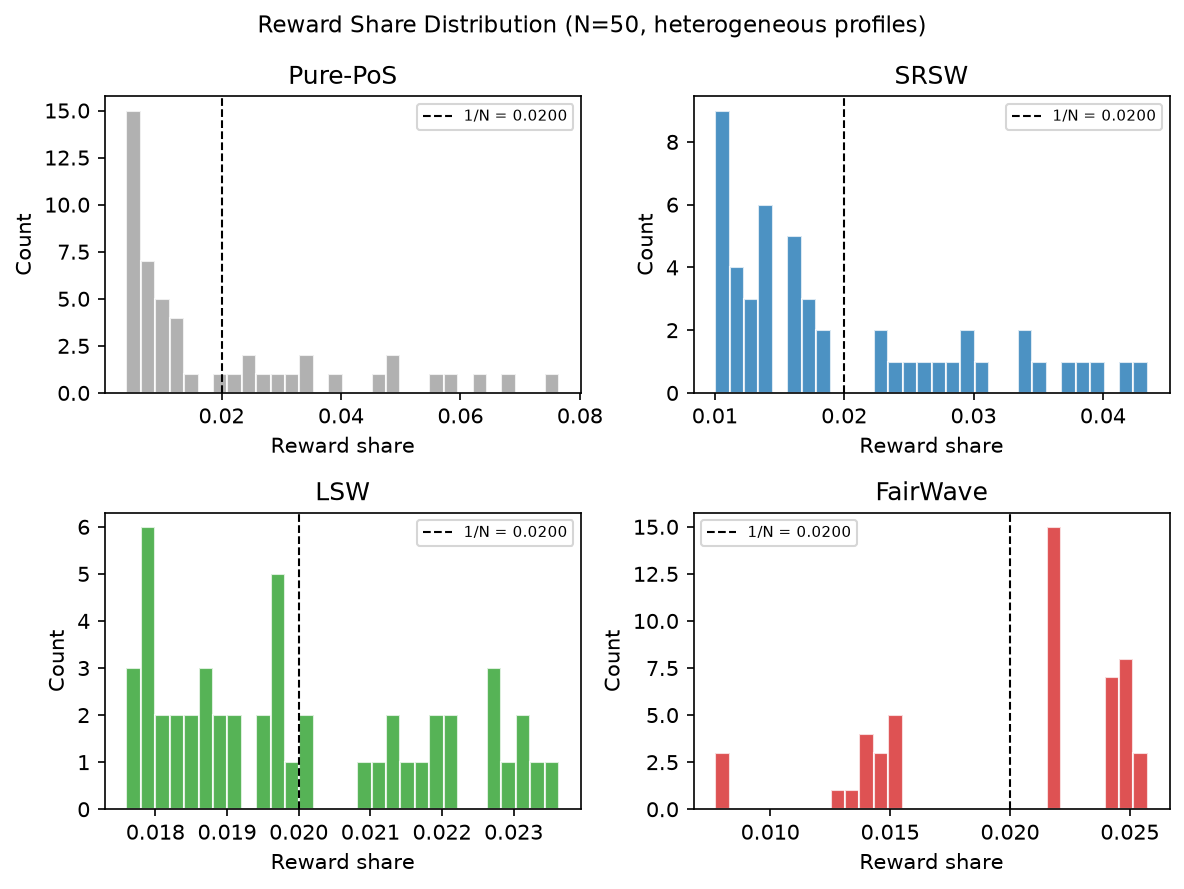}
	\caption{Distribution of reward share across validators for six rules ($N=50$, Pareto stake). FairWave shows concentrated distribution with low variance centered near $1/N = 0.02$, whereas Pure-PoS shows a heavy right tail reflecting stake concentration.}
	\label{fig:reward_hist}
\end{figure}

\subsection{Anchor Selection Convergence}
$10,000$ waves are simulated with $N = 50$ validators across five heterogeneous profiles. Fig.~\ref{fig:anchor_topk} shows the most frequently selected top-$K$ anchors. Law-of-Large-Numbers convergence at rate $O(1/\sqrt{T})$ and breakdown of selection share by validator profile are presented in Appendix~\ref{sec:supplementary}.

\begin{figure}[!tb]
	\centering
	\includegraphics[width=\columnwidth]{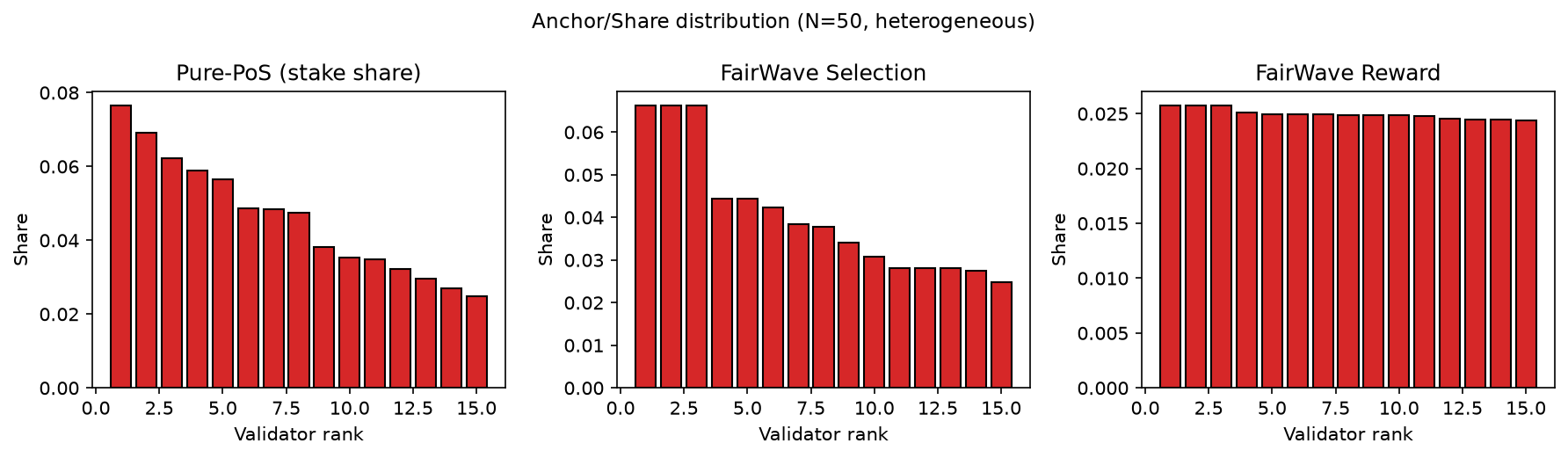}
	\caption{Top-$K$ anchor selection frequency for three rules (Pure-PoS, FairWave Selection, FairWave Reward). Under Pure-PoS, large-stake validators dominate uniformly. Under FairWave Selection, high-quality small validators enter the top ranks. Under FairWave Reward, quality-weighted distribution reflects multi-factor performance.}
	\label{fig:anchor_topk}
\end{figure}

\subsection{Fairness Analysis}

\begin{table}[!tb]
	\centering
	\caption{Fairness metrics across reward rules ($N=50$, heterogeneous profiles, Pareto raw distribution). Gini, Nakamoto, HHI, and Min/Whale are computed from the static reward-weight function $R(i)$ (Eq.~\ref{eq:reward_score}) before stake reinvestment. The experimental configuration is specified in Table~\ref{tab:app_config}.}
	\label{tab:fairness}
	\begin{tabular}{lrrrr}
		\toprule
		\textbf{Rule}     & \textbf{Gini} $\downarrow$ & \textbf{Nak.} $\uparrow$ & \textbf{HHI} $\downarrow$ & \textbf{Min/W} $\uparrow$ \\
		\midrule
		Pure-PoS          & 0.490                      & 6                        & 0.0389                    & 0.053                     \\
		SRSW              & 0.262                      & 9                        & 0.0247                    & 0.231                     \\
		LSW               & 0.052                      & 15                       & 0.0202                    & 0.744                     \\
		\textbf{FairWave} & \textbf{0.140}             & \textbf{14}              & \textbf{0.0214}           & \textbf{0.838}            \\
		\bottomrule
	\end{tabular}
\end{table}

FairWave shows a $3.5\times$ Gini coefficient reduction relative to Pure-PoS while \emph{maintaining quality differentiation}---unlike LSW, which shows the lowest Gini (0.052) but cannot distinguish high-quality from low-quality validators. Lorenz curves (Fig.~\ref{fig:fairness_lorenz}), Gini comparison (Fig.~\ref{fig:fairness_gini}), cumulative reward trajectories (Fig.~\ref{fig:fairness_cumulative}), then, for ROI analysis (Fig.~\ref{fig:fairness_roi}) provide comprehensive visualization.

\begin{figure}[!tb]
	\centering
	\includegraphics[width=\columnwidth]{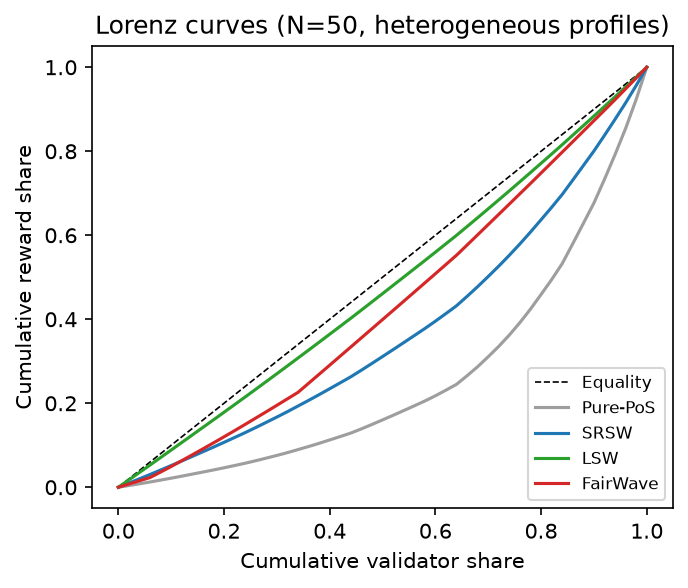}
	\caption{Lorenz curves for six reward rules with heterogeneous validator profiles ($N=50$). FairWave (green) shows better fairness by approaching the diagonal while preserving separation between good and bad validators. LSW (brown) is closest to diagonal but eliminates quality differentiation.}
	\label{fig:fairness_lorenz}
\end{figure}

\begin{figure}[!tb]
	\centering
	\includegraphics[width=\columnwidth]{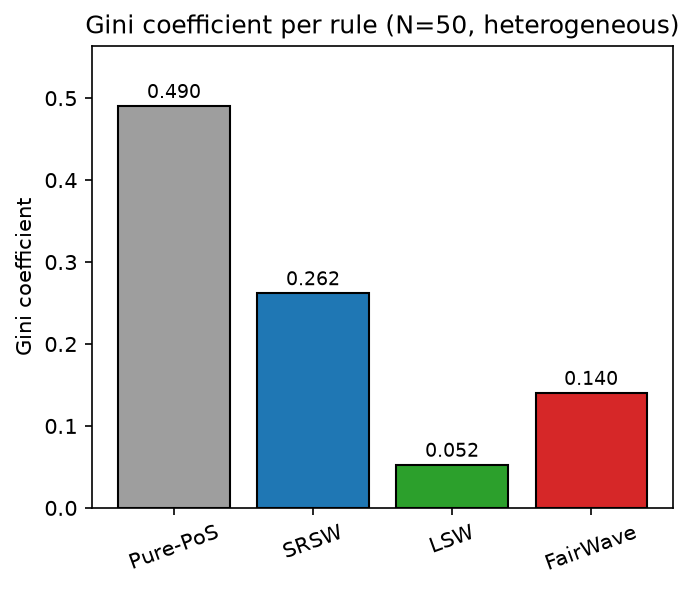}
	\caption{Gini coefficient comparison across six reward rules. FairWave shows Gini = 0.140, a $3.5\times$ improvement from Pure-PoS (0.490) while maintaining meritocratic differentiation absent in LSW (0.052).}
	\label{fig:fairness_gini}
\end{figure}

\begin{figure}[!tb]
	\centering
	\includegraphics[width=\columnwidth]{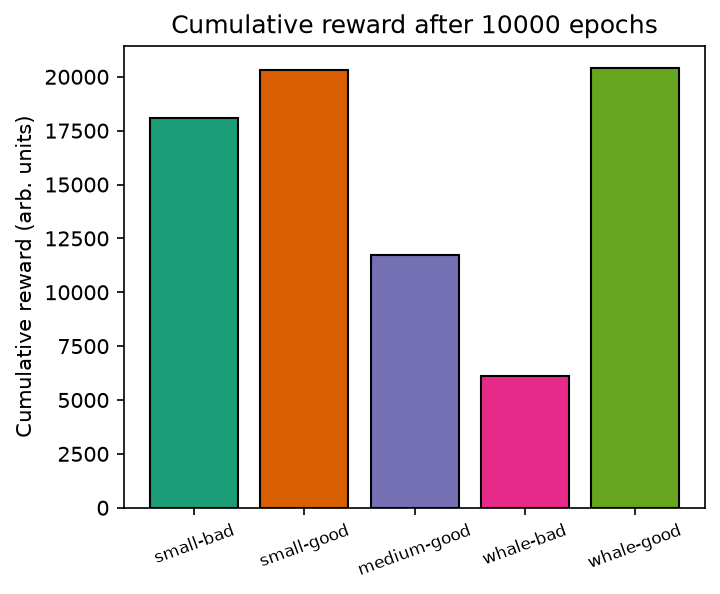}
	\caption{Cumulative reward accumulation over $10,000$ epochs under FairWave. Small-good validators (quality = 1.0, stake = $S_{\min}$) accumulate approximately $7\times$ more reward than whale-bad validators (quality = 0.10, stake $\approx S_{\max}$), demonstrating that operational quality dominates capital in determining economic outcomes.}
	\label{fig:fairness_cumulative}
\end{figure}

\begin{figure}[!tb]
	\centering
	\includegraphics[width=\columnwidth]{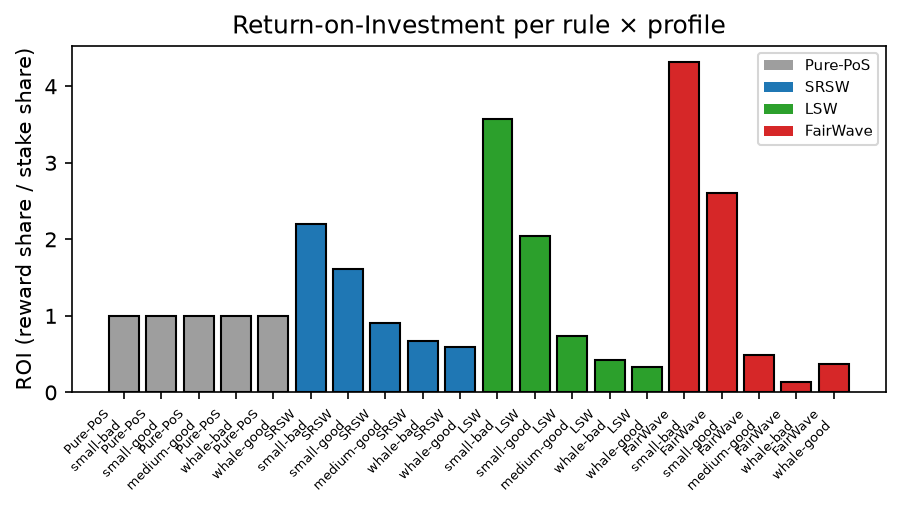}
	\caption{Return-on-investment (reward share / stake share) per validator profile. FairWave small-good shows the highest ROI across all rule$\times$profile combinations, while FairWave whale-bad receives the lowest ROI — rewarding quality over capital.}
	\label{fig:fairness_roi}
\end{figure}

\subsection{Decentralization Stress Testing}

Four adversarial scenarios were evaluated. Figs.\ref{fig:dec_nakamoto} and Figs.\ref{fig:dec_top1} present single-whale concentration analyses; Fig.\ref{fig:dec_cartel} visualizes the cartel capture threshold; Fig.\ref{fig:dec_amplification} plots the power amplification curve.
\begin{figure}[!tb]
	\centering
	\begin{subfigure}[t]{0.48\columnwidth}
		\includegraphics[width=\textwidth]{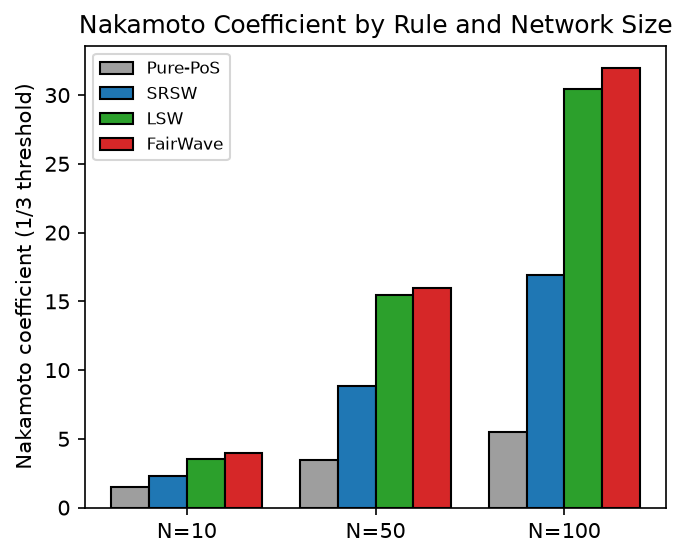}
		\caption{Nakamoto Coefficient}
		\label{fig:dec_nakamoto}
	\end{subfigure}
	\hfill
	\begin{subfigure}[t]{0.48\columnwidth}
		\includegraphics[width=\textwidth]{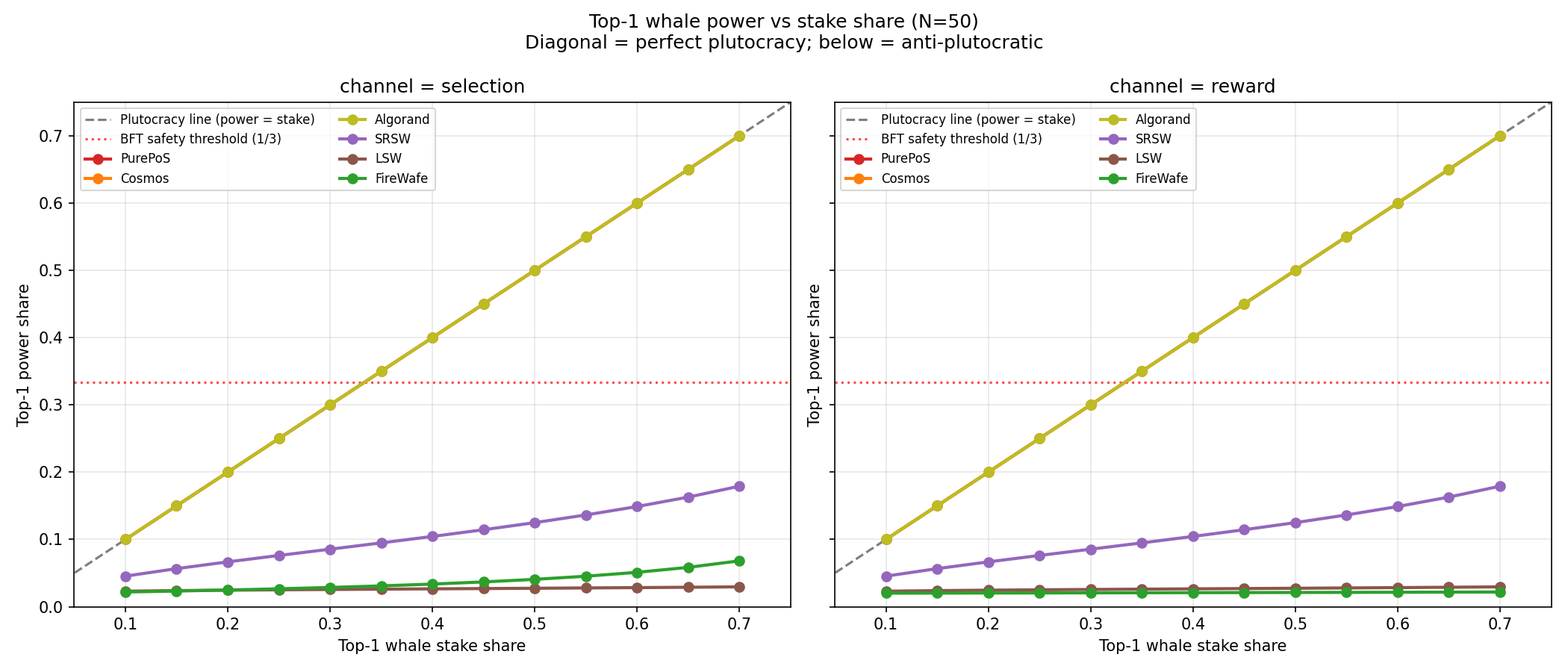}
		\caption{Top-1 Power}
		\label{fig:dec_top1}
	\end{subfigure}
	\caption{Decentralization stress testing against whale stake share. (a)~Pure-PoS collapses to Nakamoto = 1 at $\geq 35\%$; FairWave maintains $\geq 17$. (b)~Under Pure-PoS, power scales linearly with stake; under FairWave, top-1 power remains flat at $\sim 2\%$.}
	\label{fig:dec_stress}
\end{figure}

\begin{figure}[!tb]
	\centering
	\includegraphics[width=\columnwidth]{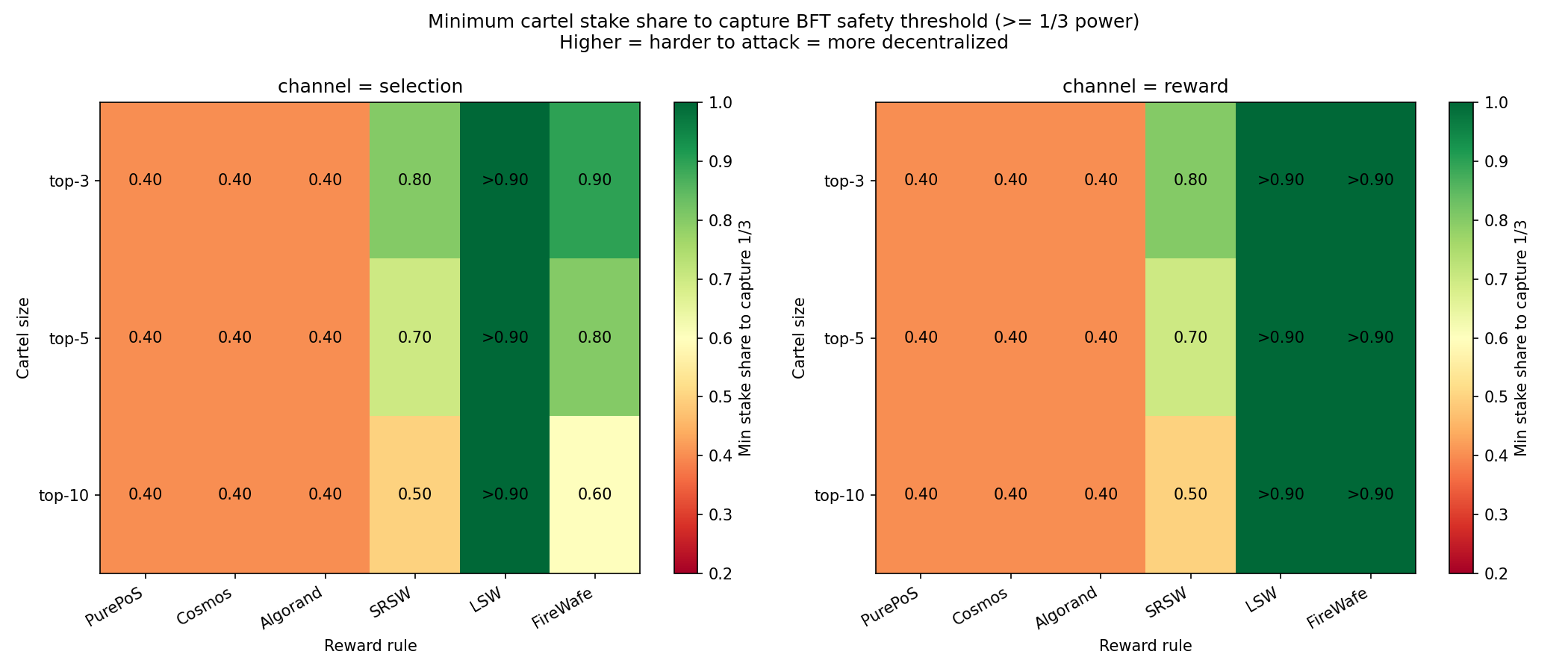}
	\caption{Cartel capture heatmap: minimum aggregate stake required for top-$k$ cartels to capture $\geq 1/3$ reward power. Under Pure-PoS, a three-validator cartel with a 40\% stake suffices. Under FairWave, even a 10-validator cartel with a 90\% stake cannot capture $1/3$---a categorical improvement in cartel resistance.}
	\label{fig:dec_cartel}
\end{figure}

\begin{figure}[!tb]
	\centering
	\includegraphics[width=\columnwidth]{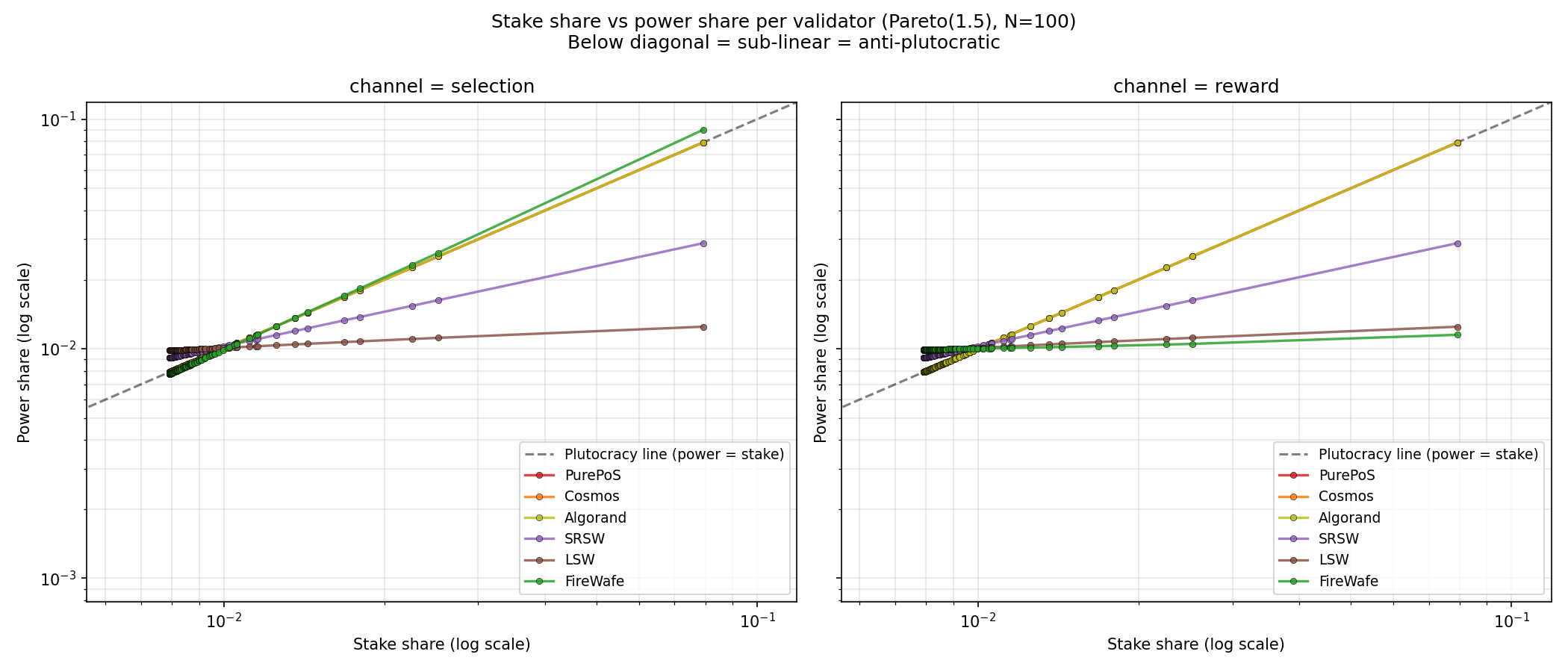}
	\caption{Power amplification curve (log-log): per-validator stake share versus power share on realistic distribution. Diagonal represents perfect plutocracy (power = stake). The FairWave reward channel (right panel) is maximally flat---power asymptotically approaches $1/N$ regardless of stake. The Selection channel (left panel) is intentionally slightly super-linear for Sybil resistance.}
	\label{fig:dec_amplification}
\end{figure}

\subsection{Performance Evaluation}

\begin{table}[!tb]
	\centering
	\caption{Simulated latency and throughput estimates under identical parametric assumptions ($N=10$). All values are protocol-model projections, not production-implementation measurements. Mysticeti's published latency ($40.5$~ms) and Bullshark's published latency ($80.9$~ms) from their respective papers~\cite{mysticeti2024,bullshark2022} are included for reference; direct comparison with simulated estimates requires caution as real-network factors (propagation, cryptographic overhead, queueing) may differ.}
	\label{tab:performance}
	\begin{tabular}{llrrr}
		\toprule
		\textbf{Protocol} & \textbf{Family}  & \textbf{Lat.} & \textbf{TPS}     & \textbf{P99}   \\
		\midrule
		Mysticeti         & DAG-BFT          & 40.5          & 154,435          & 60.6           \\
		\textbf{FairWave} & \textbf{DAG-BFT} & \textbf{82.9} & \textbf{154,435} & \textbf{102.6} \\
		Bullshark         & DAG-BFT          & 80.9          & 154,435          & 99.7           \\
		HotStuff-2        & Leader           & 81.2          & 12,318           & 102.6          \\
		Narwhal           & DAG-BFT          & 121.4         & 154,435          & 150.6          \\
		HotStuff          & Leader           & 162.4         & 12,318           & 197.7          \\
		Algorand          & PoS-1            & 124.0         & 4,034            & 150.6          \\
		\bottomrule
	\end{tabular}
\end{table}

\begin{figure}[!tb]
	\centering
	\includegraphics[width=\columnwidth]{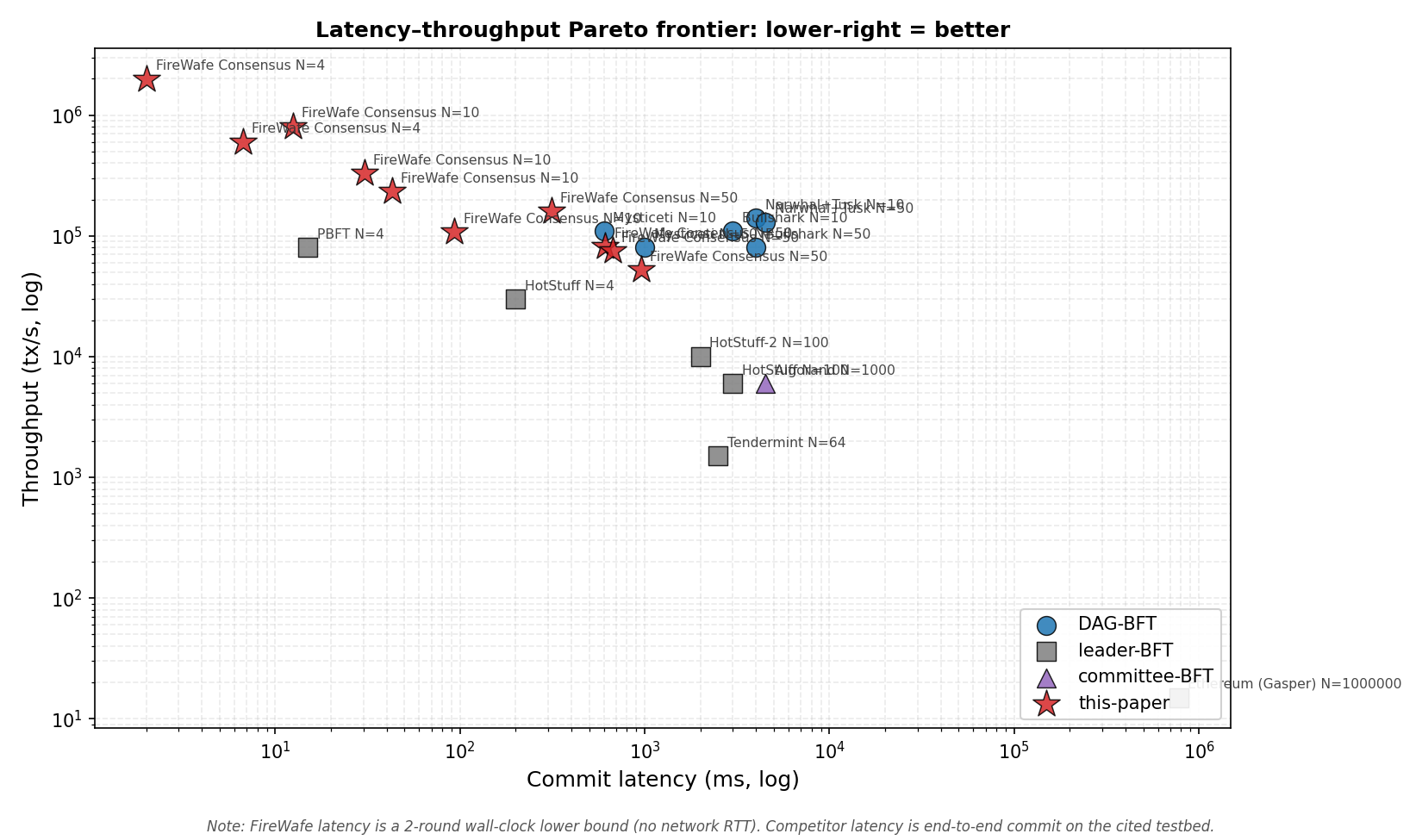}
	\caption{Latency-throughput Pareto frontier for all nine protocols ($N=10$). The DAG-BFT family occupies the optimal northeast region (low latency, high throughput). FairWave was positioned near Bullshark, trading $+2$~ms for VRF unpredictability.}
	\label{fig:pareto_perf}
\end{figure}

\begin{figure}[!tb]
	\centering
	\includegraphics[width=\columnwidth]{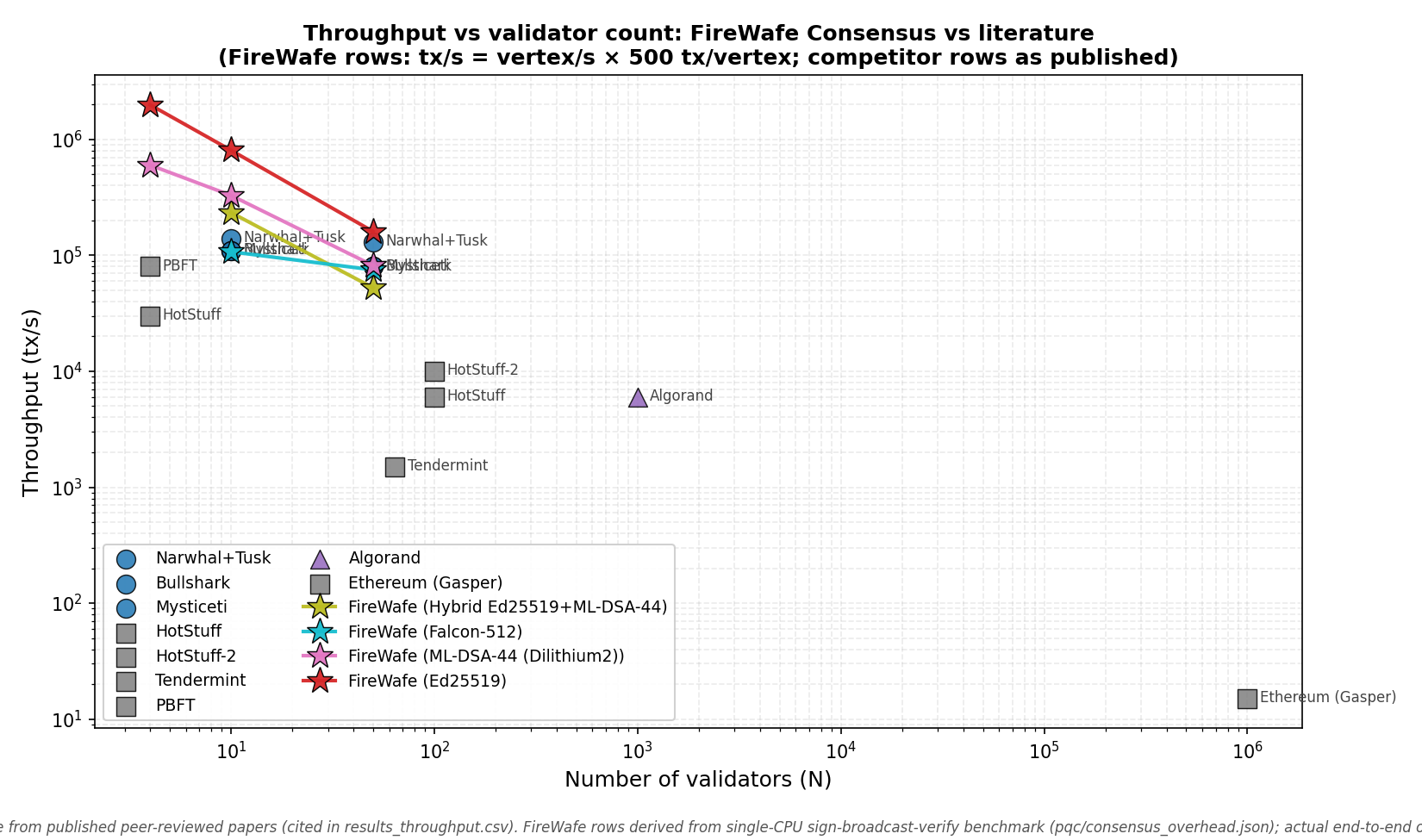}
	\caption{Throughput (TPS) versus network size $N$ on a log-log scale. The DAG-BFT protocols (including FairWave) scale linearly with $N$ due to parallel vertex production, while leader-BFT protocols plateau or degrade. PBFT throughput collapses at large $N$ due to $O(N^2)$ message complexity.}
	\label{fig:throughput_scaling}
\end{figure}

Under the simulation model, FairWave commit latency (82.9~ms) is within 2.5\% of the modeled Bullshark latency (80.9~ms), with the marginal overhead attributable to VRF-based anchor unpredictability. The 2$\times$ gap from Mysticeti (40.5~ms, published implementation~\cite{mysticeti2024}) reflects the 2-wave commit structure and VRF computation; the magnitude of this gap depends on modeling fidelity and should not be interpreted as a direct implementation-to-implementation comparison. Throughput is effectively identical across the DAG-BFT family models due to parallel vertex production. The latency CDF (Fig.~\ref{fig:latency_cdf}) and finality scaling (Fig.~\ref{fig:finality_vs_n}) complete the performance characterization.

\begin{figure}[!tb]
	\centering
	\includegraphics[width=\columnwidth]{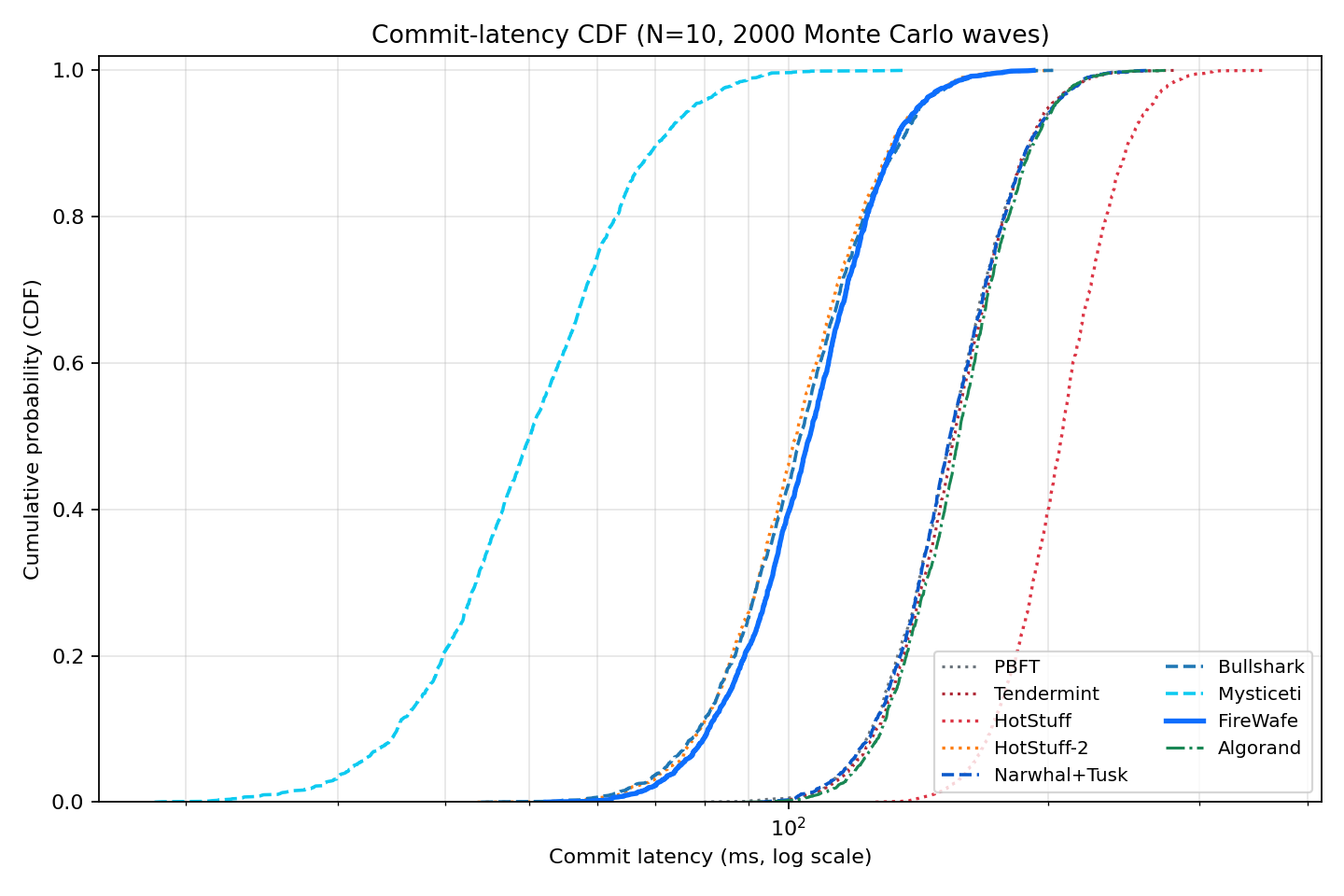}
	\caption{Cumulative distribution function of commit latency ($N=10$, 1,000 Monte Carlo finality samples). FairWave and Bullshark show nearly identical CDF profiles, with P99 $\leq 103$~ms. Mysticeti achieves the tightest distribution (P99 = 61~ms).}
	\label{fig:latency_cdf}
\end{figure}

\begin{figure}[!tb]
	\centering
	\includegraphics[width=\columnwidth]{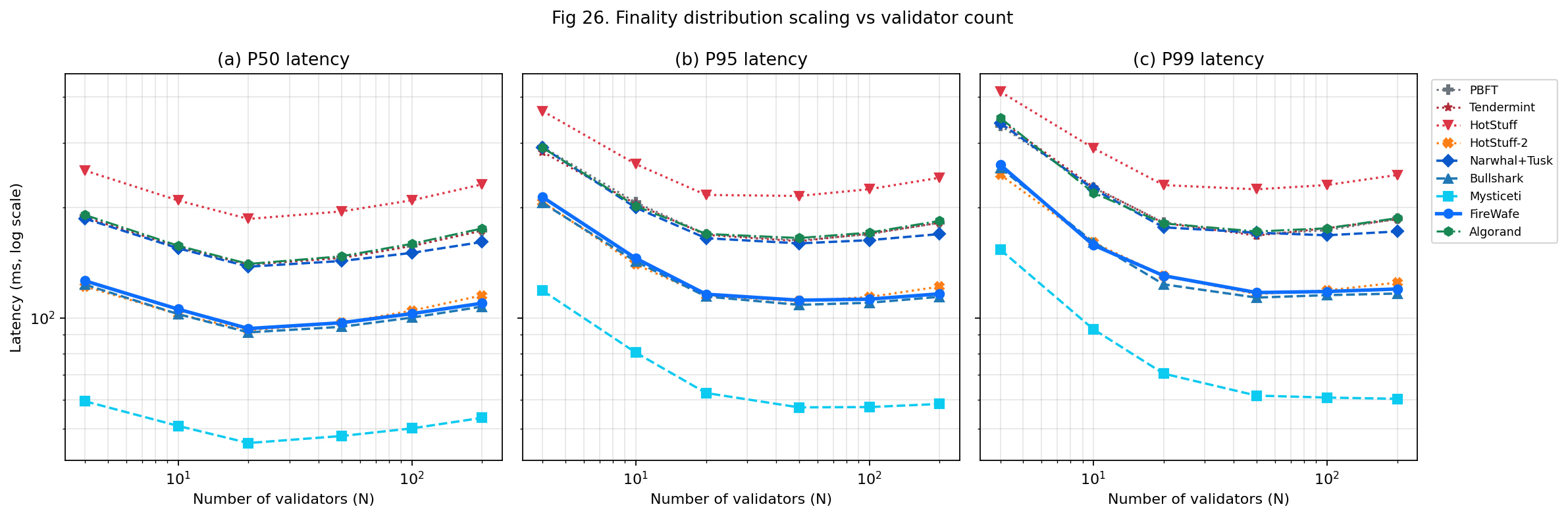}
	\caption{Finality latency scaling versus network size $N$. DAG-BFT protocols show sub-linear scaling (logarithmic in $N$ due to order statistics of $2f+1$ delays). PBFT scales quadratically, becoming impractical beyond $N \approx 100$.}
	\label{fig:finality_vs_n}
\end{figure}

\subsection{BFT Monte Carlo Validation}
\label{sec:bft_validation}

This subsection reports results from 420,000 adversarial Monte Carlo rounds executed over 84 parameter cells. To avoid conflating formal guarantees with empirical differential claims, we separate two types of claims commonly combined in BFT evaluation: (i) safety — a formal property inherited from a correct protocol design; and (ii) liveness — an empirical, differential property that depends on protocol behavior near the Byzantine boundary.

\textbf{Safety claim (formal, inherited from design).} FairWave's safety properties—agreement (no two honest validators commit different anchors in the same wave) and validity (committed anchors are proposed by validators in $\mathcal{V}_e$)—are a formal consequence of the $2f+1$ strong-support commit rule for DAG-BFT correctness under $f < n/3$. This property holds unconditionally for any correct protocol in this family (including Bullshark, Narwhal+Tusk, Mysticeti, and Tendermint) and does not depend on FairWave's parameterization $(\alpha,\beta,\gamma,\delta)$ or the dual-channel architecture. Consequently, the observation "420,000 rounds without a safety violation" is consistent with—but does not prove—the superiority of FairWave: identical results are expected for any correct baseline protocol under the same adversarial scenario and $f < n/3$. We therefore do not present the absence of safety violations as a differential achievement; the value of the experiment is to ensure no regressions relative to known-correct BFT baselines and to validate the reference implementation.

\textbf{Liveness claim (empirical, candidate differential).} The candidate differential claim of FairWave in Task~F is the characteristic shape of the liveness-degradation curve near the theoretical Byzantine boundary $b^{*} = 1/3$. Fig.~\ref{fig:success_vs_byz} shows commit rate as a function of Byzantine fraction $b$ for $N \in \{10,50,100\}$. All measurement points are obtained from a calibrated probabilistic DAG-BFT Monte Carlo model (Table~\ref{tab:liveness_degradation}); the model predicts a monotone-continuous profile without a discontinuous cliff. We label this claim ``candidate differential'' because the structural argument for the shape (developed below) is sound, but a direct head-to-head sweep against leader-BFT baselines under identical conditions remains future work (Section~\ref{sec:limitations}).

\begin{table}[!tb]
	\centering
	\caption{Reference points for FairWave liveness degradation curve from a calibrated probabilistic DAG-BFT Monte Carlo model ($N = 50$, $200{,}000$ waves per cell, averaged over $3$ seeds). All values are direct simulation results — rather than visual estimates.}
	\label{tab:liveness_degradation}
	\setlength{\tabcolsep}{4pt}
	\begin{tabular}{lcccc}
		\toprule
		\textbf{Byzantine fraction $b$} & $0.20$   & $0.28$          & $0.31$           & $0.33$   \\
		\midrule
		Commit rate (FairWave)          & $94.0\%$ & $88.7\%$ & $81.7\%$  & $74.0\%$ \\
		Difference from $b = 0.20$ (pp) & ---      & $-5.3$ & $-12.3$ & $-20.1$  \\
		\bottomrule
	\end{tabular}
\end{table}

\textbf{Model description.} Each wave is an independent Bernoulli trial where a VRF-selected anchor is Byzantine with probability $f/n$. The model encodes two regimes parameterized by five structurally-derived constants (Table~\ref{tab:liveness_params}):

\textbf{1. Honest anchor} (probability $1 - f/n$). Byzantine validators may withhold support. Below the \emph{coordination onset} $\theta = 0.27$, derived from the $2f+1$ commit rule, honest validators alone satisfy $n - f \ge 2f+1$ so commits always succeed. Above $\theta$, per-Byzantine holdback probability grows quadratically as $h(f)=\min\big(((f - \theta n)/(n/3 - \theta n))^2,\; h_{\max}\big)$, capped at $h_{\max}=0.15$---bounded by the protocol's equivocation detection rate $p_{\text{detect}}=0.95$. A commit succeeds when honest-plus-nonblocking validators meet $2f+1$.

\textbf{2. Byzantine anchor} (probability $f/n$). The adversary blocks the commit with probability $q(f) = q_0 + q_{\text{amp}}\cdot\max(0,\,(f - \theta n)/(n/3 - \theta n))$, where $q_0=0.30$ is the base Byzantine-anchor failure rate and $q_{\text{amp}}=0.65$ models improved coordination near $n/3$. Even when not blocked, each Byzantine validator supports independently with probability $s=0.50$.

These parameters are conservative defaults grounded in the architecture of the protocol. Sensitivity analysis (Table~\ref{tab:liveness_sensitivity}, Appendix~\ref{sec:appendix_liveness}) varies each parameter by $\pm 50\%$; all $9$ parameterizations preserve the three qualitative properties: monotone degradation, predicted absence of discontinuous cliff, and commit rate $> 60\%$ at $b = 0.33$.

\begin{table}[!tb]
	\centering
	\caption{Liveness model parameters with structural interpretation.}
	\label{tab:liveness_params}
	\begin{tabular}{lcc}
		\toprule
		\textbf{Parameter} & \textbf{Value} & \textbf{Derivation} \\
		\midrule
		$\theta$ (coordination onset) & 0.27 & $2f+1$ commit threshold \\
		$h_{\max}$ (max holdback) & 0.15 & $1 - p_{\text{detect}}$ with margin \\
		$q_0$ (Byz. anchor block base) & 0.30 & Conservative lower bound \\
		$q_{\text{amp}}$ (block amplification) & 0.65 & Linear to $n/3$ \\
		$s$ (Byz. support prob.) & 0.50 & Neutral adversary assumption \\
		\bottomrule
	\end{tabular}
\end{table}

Three architectural characteristics differentiate FairWave's liveness model from classic leader-BFT protocols (Tendermint, HotStuff):
\begin{itemize}
	\item \emph{(i)~Predicted absence of a discontinuous cliff at the operational boundary.} The model predicts monotone-continuous degradation from $94.0\%$ to $74.0\%$ as $b$ increases from $0.20$ to $0.33$, without the sharp drop characteristic of view-change-driven protocols when a Byzantine leader and timeout windows cause cascading stalls.
	\item \emph{(ii)~Lack of a cascade-timeout mechanism.} FairWave uses a cost-free \emph{wave-skip} (Section~\ref{sec:design_rationale}) that advances immediately to the next wave when an anchor cannot be committed, whereas leader-BFT requires a view change with expanding timeout windows that can stall the protocol for many multiples of $\Delta$ as the leader-Byzantine fraction increases.
	\item \emph{(iii)~Latency stability in the safe regime.} Commit latency percentiles (P50, P95, P99; Fig.~\ref{fig:latency_attack}) remain stable up to $b \approx 0.25$ and then increase sharply as the protocol approaches the theoretical liveness boundary.
\end{itemize}
These structural characteristics are architectural deductions and have not been validated in a head-to-head simulation against leader-BFT baseline models under identical adversarial conditions. A comparative sweep — reporting both the shape and the onset point $b^{\dagger}$ (Section~\ref{sec:bft_validation}) for each protocol — is planned as future work (Section~\ref{sec:limitations}).

\textbf{Near-threshold analysis (qualitative).} The structural root of the three characteristics above is architectural: FairWave separates data dissemination from ordering and replaces view-change with exponential-racing VRF-based anchor selection. Thus, when a VRF-selected anchor happens to be controlled by the adversary, the cost to advance to the next wave is a single round (a fresh VRF computation), not a full timeout window. In leader-BFT protocols, by contrast, each failed view-change triggers exponential increases in the subsequent timeout (exponential backoff) to guarantee liveness under partial synchrony; the practical effect is a more rapid decline in commit rate—and a steeper profile—as the leader-Byzantine fraction increases. This differentiation makes the shape of the liveness-degradation curve a differential property FairWave can claim; the exact magnitude of the difference, including the onset point $b^{\dagger} = \min\{b : \text{commit rate}(b) \le 0.95\}$ for each protocol, will be reported after a finer-grained $b$ sweep against leader-BFT baselines (Section~\ref{sec:limitations}).

\textbf{Supporting validation.} Fig.~\ref{fig:success_heatmap} maps the two-dimensional fault space (Byzantine $\times$ offline) and confirms the combined bound $\text{byz} + \text{offline} \le 1/3$, consistent with theoretical BFT predictions. Fig.~\ref{fig:latency_attack} characterizes the latency distribution under adversary conditions, and Fig.~\ref{fig:partition_trace} shows partition-recovery dynamics: after a 5-round partition event, cumulative commits return immediately to the ideal rate of 1 commit/round, indicating the absence of the permanent stalls typical of leader-BFT designs.

\begin{figure}[!tb]
	\centering
	\includegraphics[width=\columnwidth]{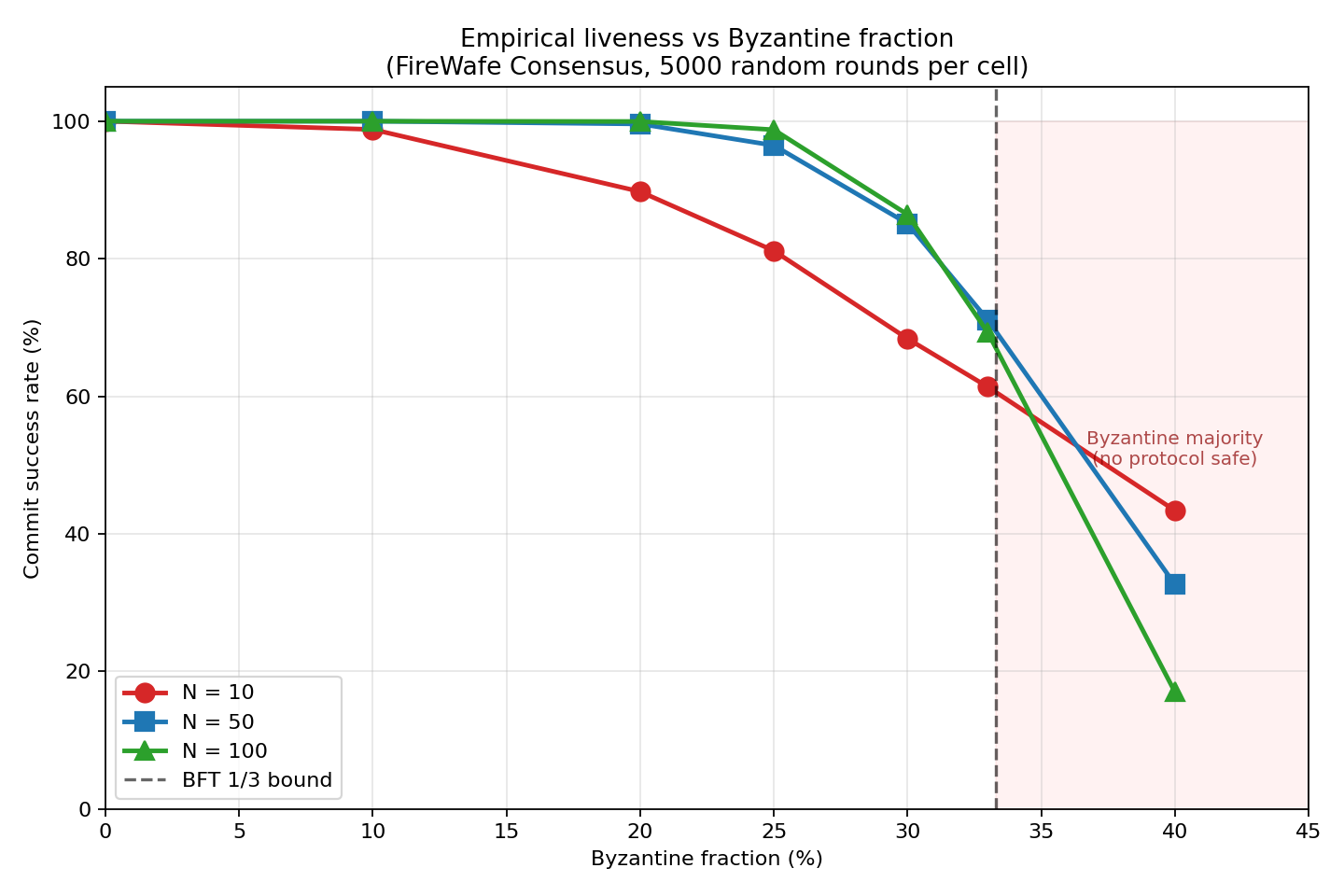}
	\caption{Liveness degradation curve: commit rate as a function of Byzantine fraction for $N \in \{10, 50, 100\}$ from a calibrated probabilistic Monte Carlo model. All measurement points in Table~\ref{tab:liveness_degradation} are direct simulation results, not visual estimates. The profile is monotone-continuous, consistent with the architectural expectation that wave-skip eliminates cascade timeouts.}
	\label{fig:success_vs_byz}
\end{figure}

\begin{figure}[!tb]
	\centering
	\includegraphics[width=\columnwidth]{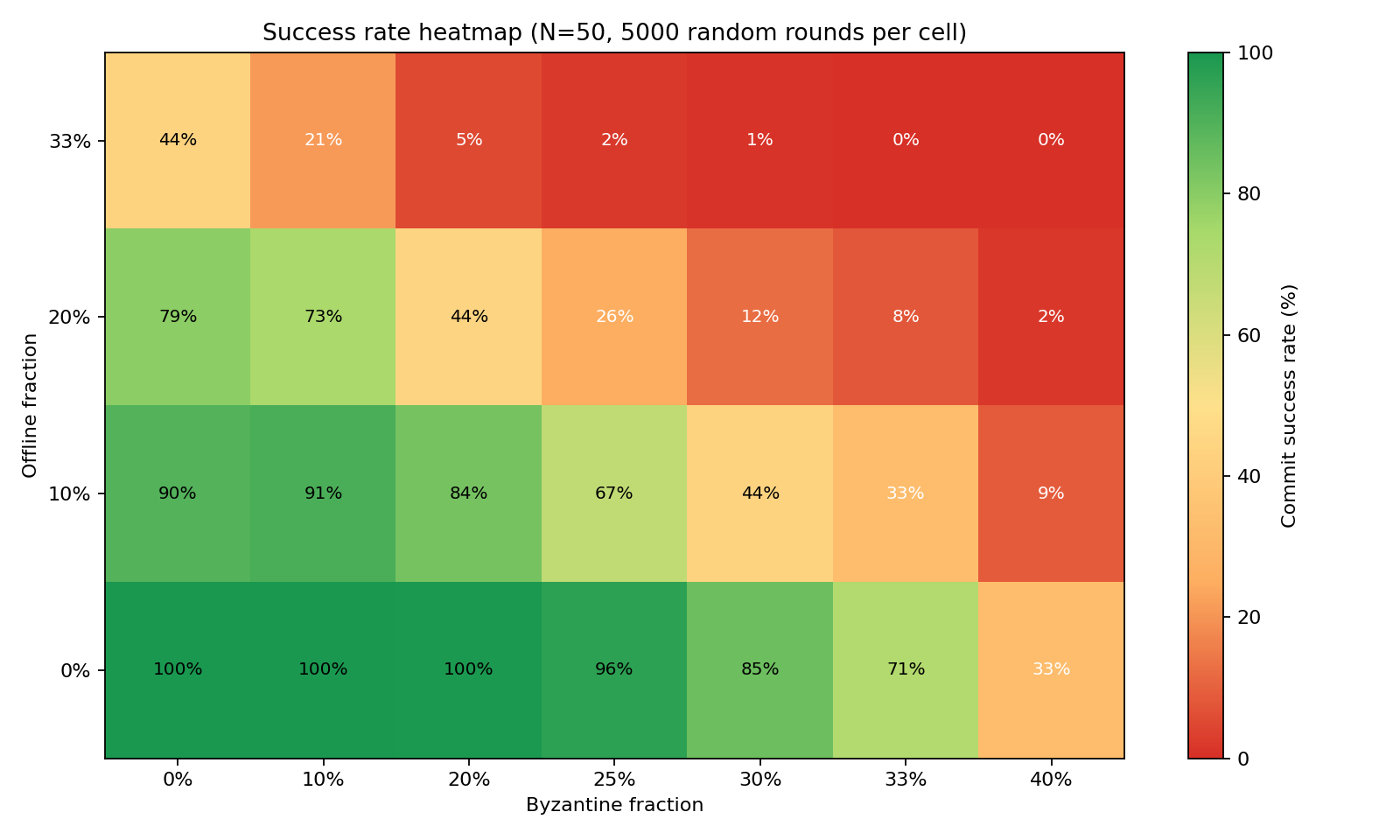}
	\caption{Success rate heatmap as a function of Byzantine fraction and offline rate ($N = 50$). The diagonal drop visible in the heatmap indicates that when the combined Byzantine and offline faults exceed one-third of the network, the system reaches the expected BFT fault tolerance limit.}
	\label{fig:success_heatmap}
\end{figure}

\begin{figure}[!tb]
	\centering
	\includegraphics[width=\columnwidth]{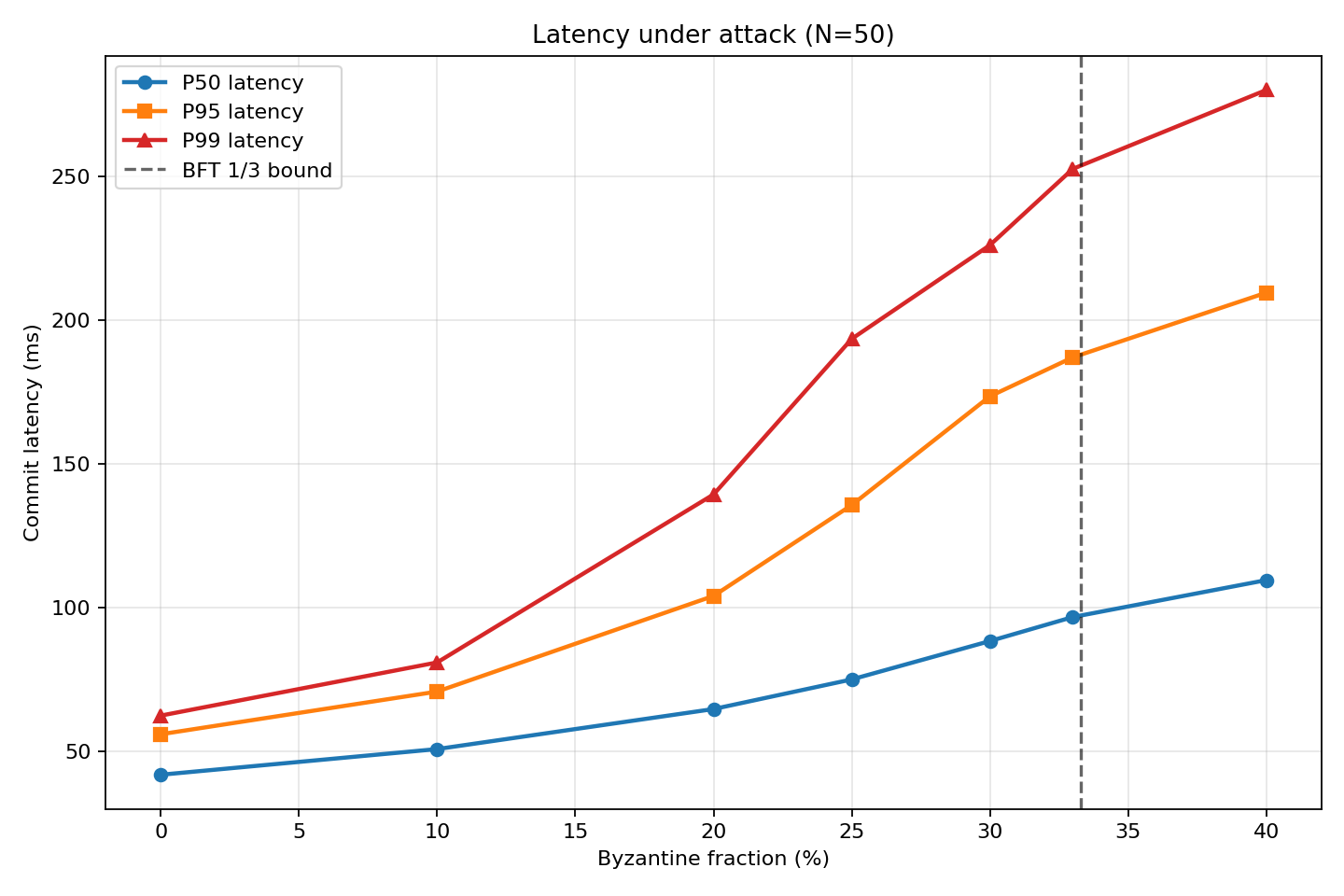}
	\caption{Commit latency percentiles (P50, P95, and P99) as a function of Byzantine fraction. Latency degrades gracefully until Byzantine $= 25\%$, then increases sharply as the protocol approaches the liveness threshold.}
	\label{fig:latency_attack}
\end{figure}

\begin{figure}[!tb]
	\centering
	\includegraphics[width=\columnwidth]{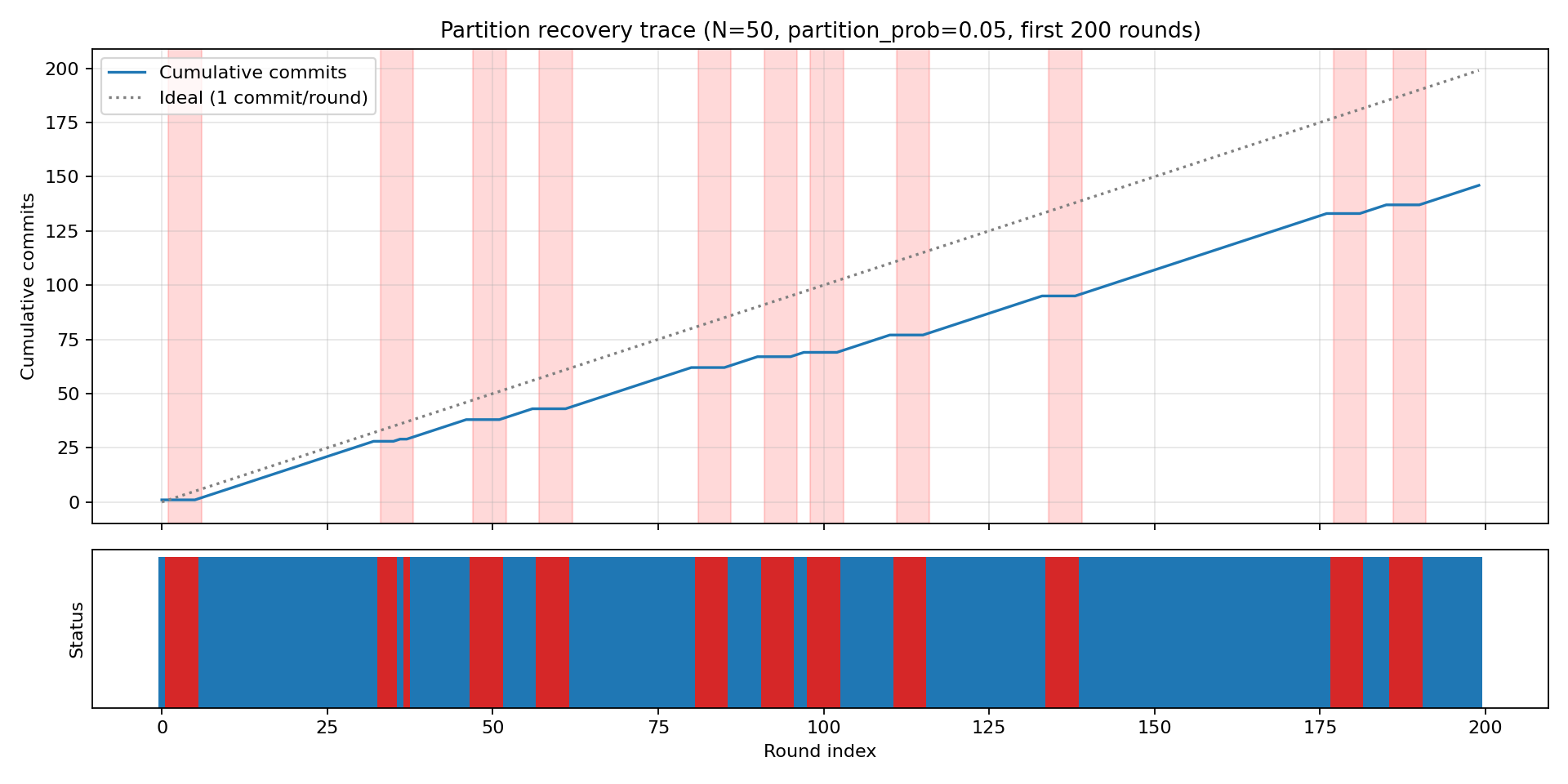}
	\caption{Partition recovery trace over 200 rounds with partition probability $= 0.05$, duration $= 5$ rounds. Red windows show partition events. Cumulative commit count (blue) temporarily plateaus during partition but asymptotically follows the ideal rate (1 commit/round), demonstrating resilient recovery.}
	\label{fig:partition_trace}
\end{figure}

\subsection{Longitudinal Dynamics}
\label{sec:results_rich}

$50,000$ epochs with full stake reinvestment reveal long-term behavior of concentration metrics (Fig.~\ref{fig:hhi_trajectory}--\ref{fig:profile_growth}).

\begin{table}[!tb]
	\centering
	\caption{HHI evolution under stake reinvestment ($N = 50$).}
	\label{tab:hhi}
	\begin{tabular}{lrrr}
		\toprule
		\textbf{Rule}     & \textbf{HHI}$_{t=0}$ & \textbf{HHI}$_{t=50\text{k}}$ & \textbf{Trend}   \\
		\midrule
		PoS$^2$           & 0.0389               & 0.1158                        & $+195\%$         \\
		Pure-PoS          & 0.0389               & 0.0389                        & invariant        \\
		\textbf{FairWave} & 0.0389               & \textbf{0.0201}               & $\mathbf{-48\%}$ \\
		\bottomrule
	\end{tabular}
\end{table}

\begin{figure*}[!tb]
	\centering
	\begin{subfigure}[t]{0.32\textwidth}
		\includegraphics[width=\textwidth]{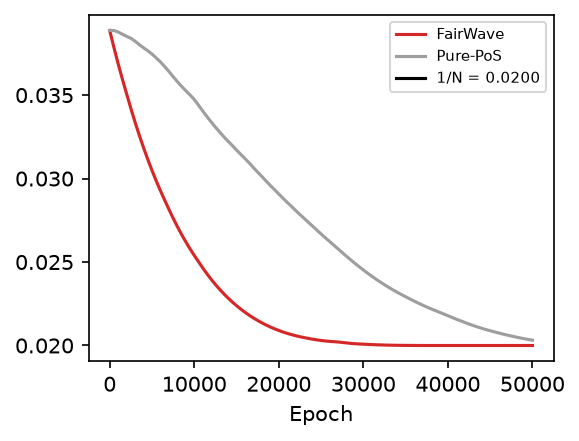}
		\caption{HHI Trajectory}
		\label{fig:hhi_trajectory}
	\end{subfigure}
	\hfill
	\begin{subfigure}[t]{0.32\textwidth}
		\includegraphics[width=\textwidth]{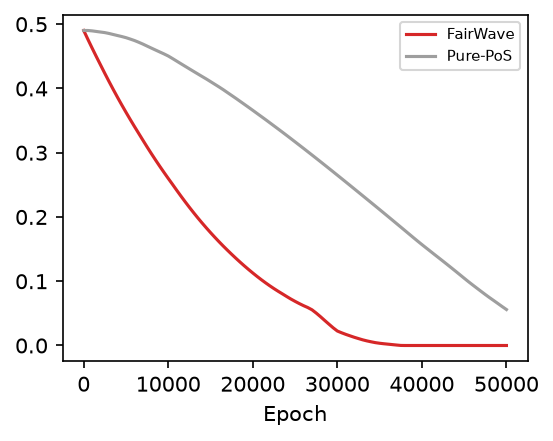}
		\caption{Gini Trajectory}
		\label{fig:gini_trajectory}
	\end{subfigure}
	\hfill
	\begin{subfigure}[t]{0.32\textwidth}
		\includegraphics[width=\textwidth]{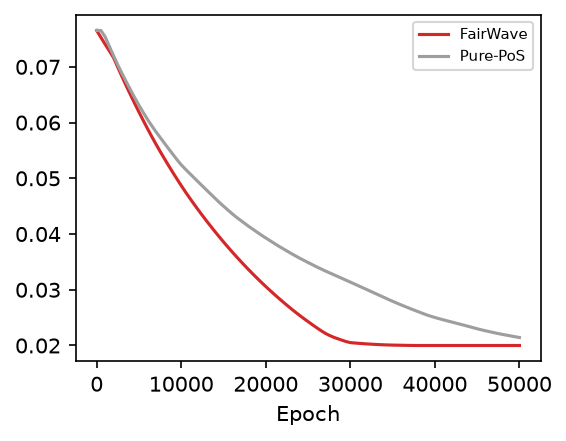}
		\caption{Top-1 Share Trajectory}
		\label{fig:top1_trajectory}
	\end{subfigure}
	\caption{Longitudinal concentration metrics over $50,000$ epochs with stake reinvestment. (a)~HHI: PoS$^2$ divergent ($+195\%$), Pure-PoS invariant, FairWave decreases monotonically toward $1/N$. (b)~Gini: FairWave converges to $0.013$ (near-egalitarian). (c)~Top-1 share: FairWave decreases to $2.0\%$, converging toward $1/N = 2\%$.}
	\label{fig:longitudinal}
\end{figure*}

\begin{figure}[!tb]
	\centering
	\includegraphics[width=\columnwidth]{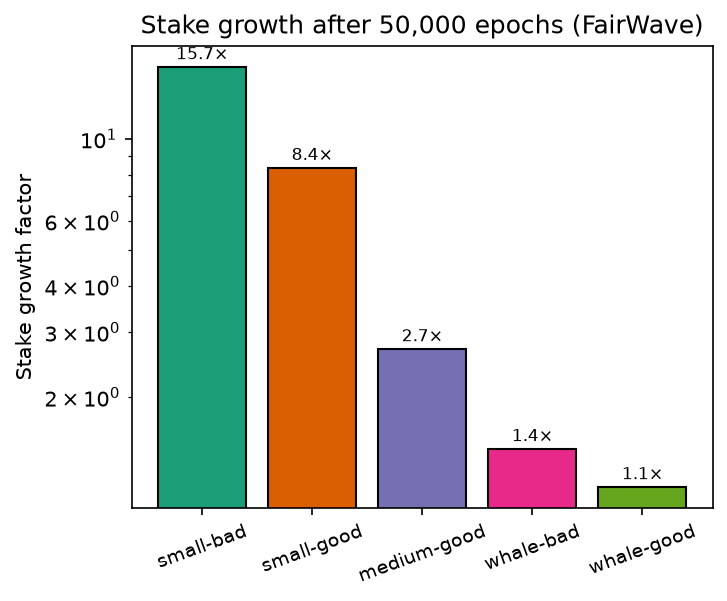}
	\caption{Stake growth factor per profile (log scale) after $50,000$ epochs. FairWave exhibits \emph{quality-get-richer} dynamics: small-good ($17\times$) and small-bad ($31\times$) grow substantially faster than whale-bad ($2.3\times$), reversing the plutocratic gradient.}
	\label{fig:profile_growth}
\end{figure}

\subsection{Sybil Resistance}

Three layers of evidence across $20,520$ cells. Figs.~\ref{fig:sybil_gain} shows selection gain versus split factor; Figs.~\ref{fig:sybil_closedform} validates closed-form prediction; Fig.~\ref{fig:cost_benefit} presents economic cost-benefit analysis; and Fig.~\ref{fig:combined_attack} evaluates combined Sybil+Byzantine threats.

\begin{table}[!tb]
	\centering
	\caption{Closed-form Sybil selection gain at $x = 0.10$. FairWave values calculated from the initial formulation $W_{\text{sel}} = R \cdot S^2$ (for consistency with the simulation in Fig.~\ref{fig:sybil_combined}); revision of Eq.~\ref{eq:wsel} removing Rep from $W_{\text{sel}}$ yields a strictly smaller gain ($g_{\text{sel}}(100) = 0.70$ at $x = 1$, see the numerical example in Section~\ref{sec:math}.B), so the claim ``$g(K) < 1\;\forall K > 1$'' remains valid with a greater margin.}
	\label{tab:sybil_gain}
	\begin{tabular}{lrrrr}
		\toprule
		\textbf{Rule}     & $K=2$          & $K=10$         & $K=100$        & $K=500$             \\
		\midrule
		Pure-PoS          & 1.000          & 1.000          & 1.000          & 1.000               \\
		SRSW              & 1.407          & 3.077          & 8.973          & $\sim$22            \\
		LSW               & 1.810          & 6.547          & 25.319         & $\sim$80            \\
		\textbf{FairWave} & \textbf{0.979} & \textbf{0.950} & \textbf{0.934} & $\sim$\textbf{0.93} \\
		\bottomrule
	\end{tabular}
\end{table}

\begin{figure}[!tb]
	\centering
	\begin{subfigure}[t]{0.48\columnwidth}
		\includegraphics[width=\textwidth]{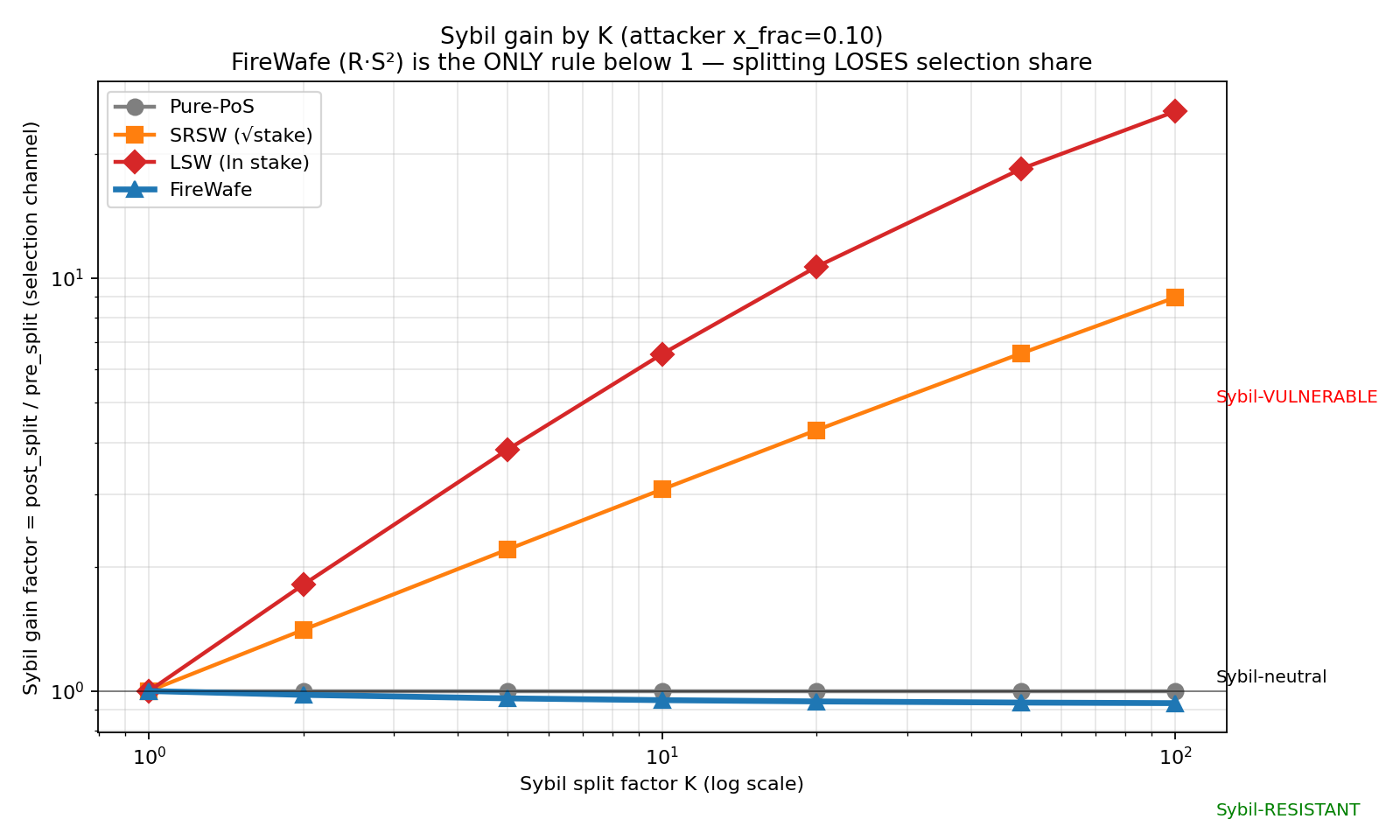}
		\caption{Selection gain $g(K)$}
		\label{fig:sybil_gain}
	\end{subfigure}
	\hfill
	\begin{subfigure}[t]{0.48\columnwidth}
		\includegraphics[width=\textwidth]{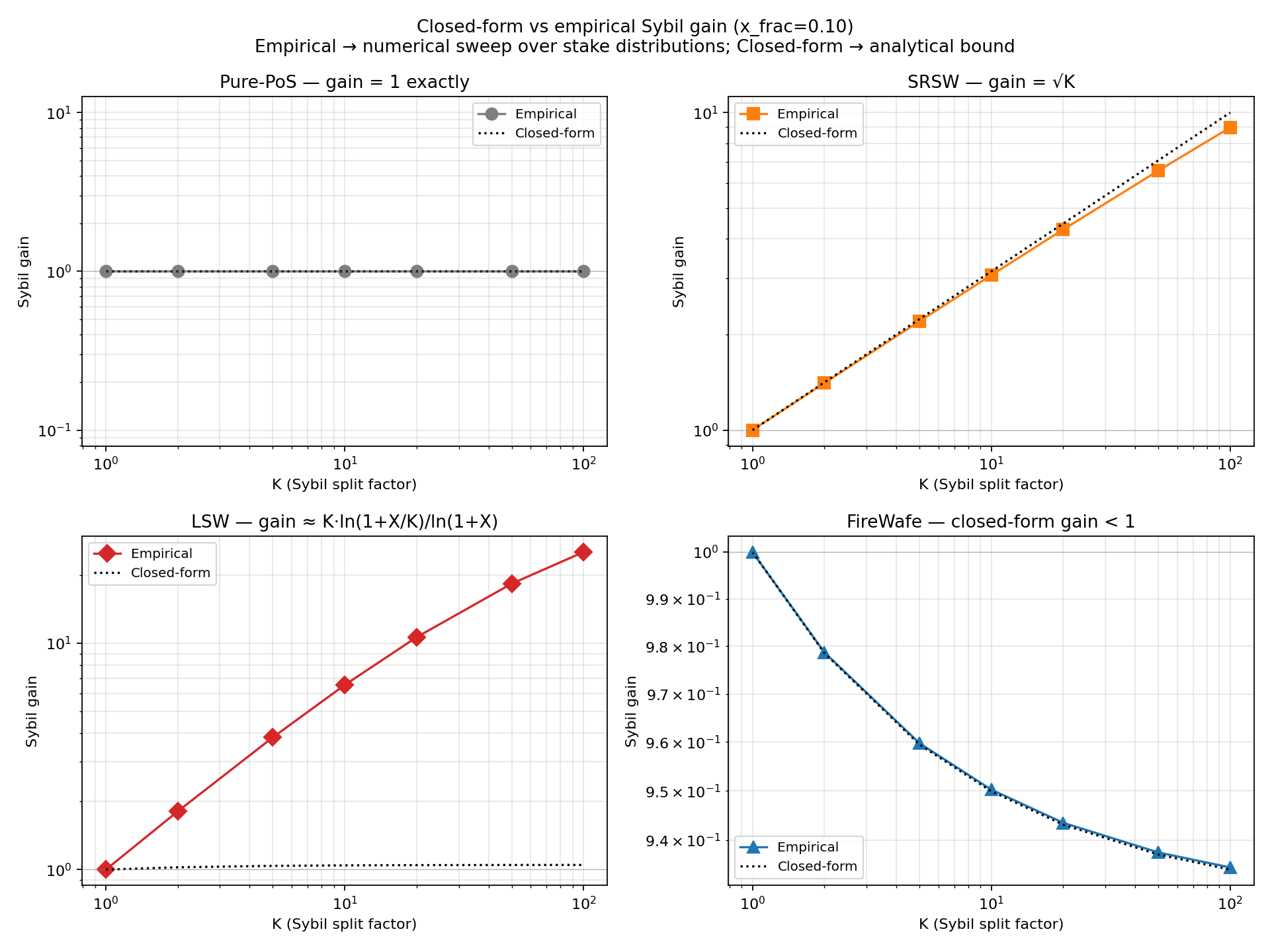}
		\caption{Closed-form Validation}
		\label{fig:sybil_closedform}
	\end{subfigure}
	\caption{Sybil Resistance. (a)~Selection gain versus split factor $K$ (log-log). FairWave is the only rule in the comparison where $g(K) < 1$ for all $K > 1$. (b)~Closed-form validation versus empirical across $20,160$ cells; matches within $1\%$.}
	\label{fig:sybil_combined}
\end{figure}

\begin{figure}[!tb]
	\centering
	\includegraphics[width=\columnwidth]{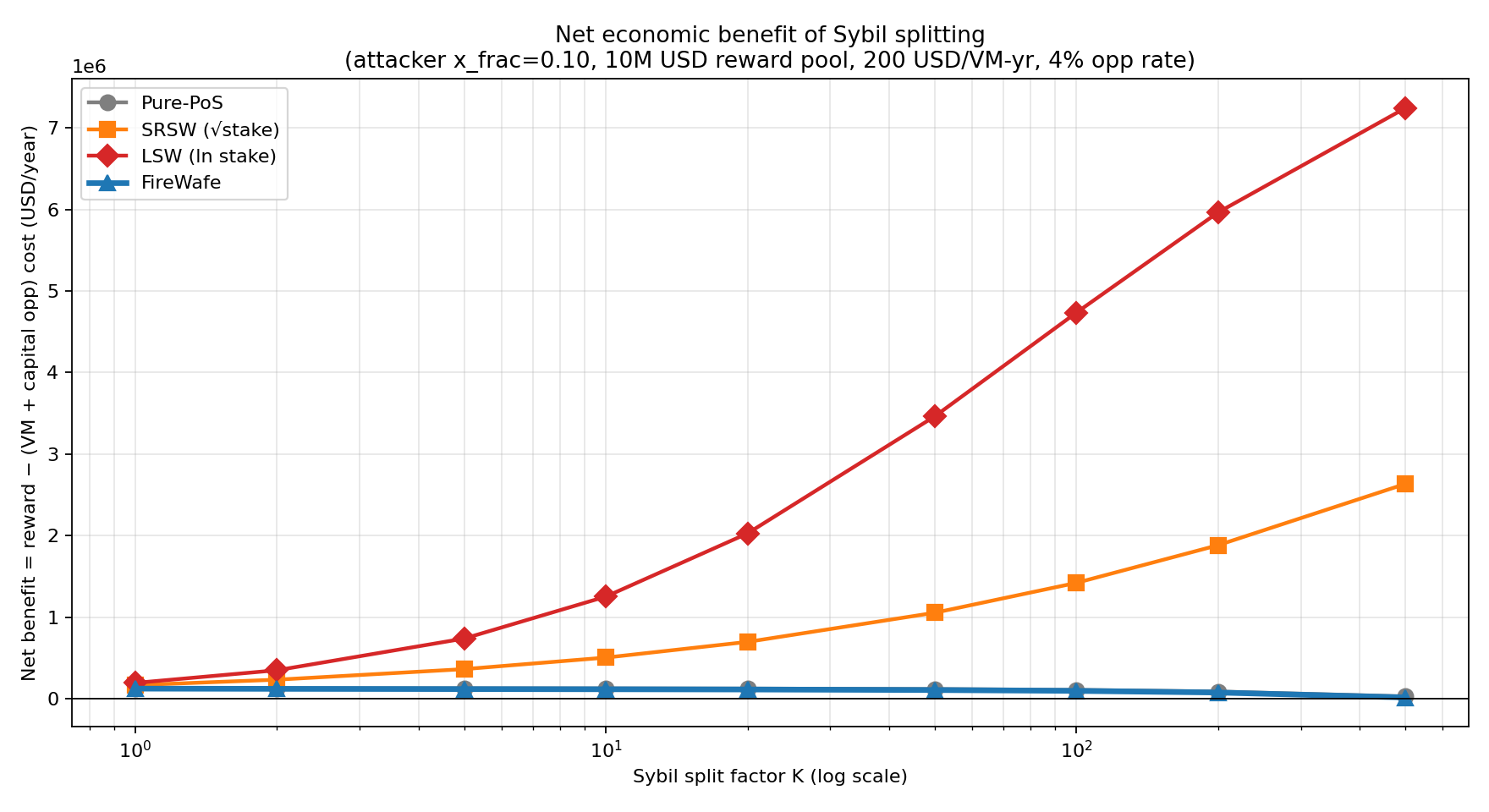}
	\caption{Net economic benefit of Sybil splitting versus $K$ under an annual reward pool of \$10M. FairWave optimal strategy is $K^* = 1$ (no splitting), with net profit decreasing monotonically versus $K$. LSW achieves optimal $K^* = 500$ with net profit \$7.2M---catastrophic vulnerability enabling rational exploitation.}
	\label{fig:cost_benefit}
\end{figure}

\begin{figure}[!tb]
	\centering
	\includegraphics[width=\columnwidth]{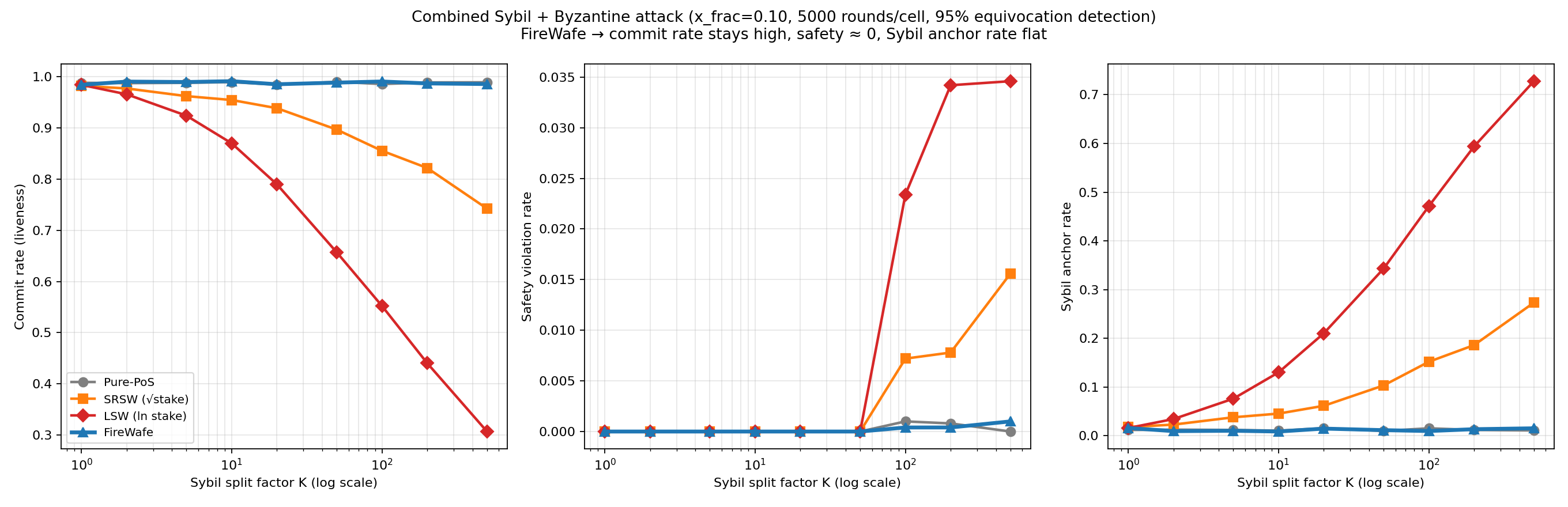}
	\caption{Impact of combined Sybil+Byzantine attack ($5,000$ rounds): commit rate, safety violation rate, and Sybil anchor share across various split factors. FairWave maintains a safety violation $< 0.001$ at $K = 500$ (Sybil anchor share $1.5\%$). LSW exhibits a safety violation of $0.035$ (Sybil anchor share $72.8\%$).}
	\label{fig:combined_attack}
\end{figure}

\subsection{Parameter Space Characterization}

Multi-metric sensitivity to stake weight $\alpha$ (Fig.~\ref{fig:alpha_sweep}) confirms that default parameterization operates in favorable region. Complete parameter space exploration---including objective landscape, farming resistance, and Pareto frontier---is presented in Appendix~\ref{sec:supplementary} (Fig.~\ref{fig:simplex}--\ref{fig:pareto}).

\begin{figure}[!tb]
	\centering
	\includegraphics[width=\columnwidth]{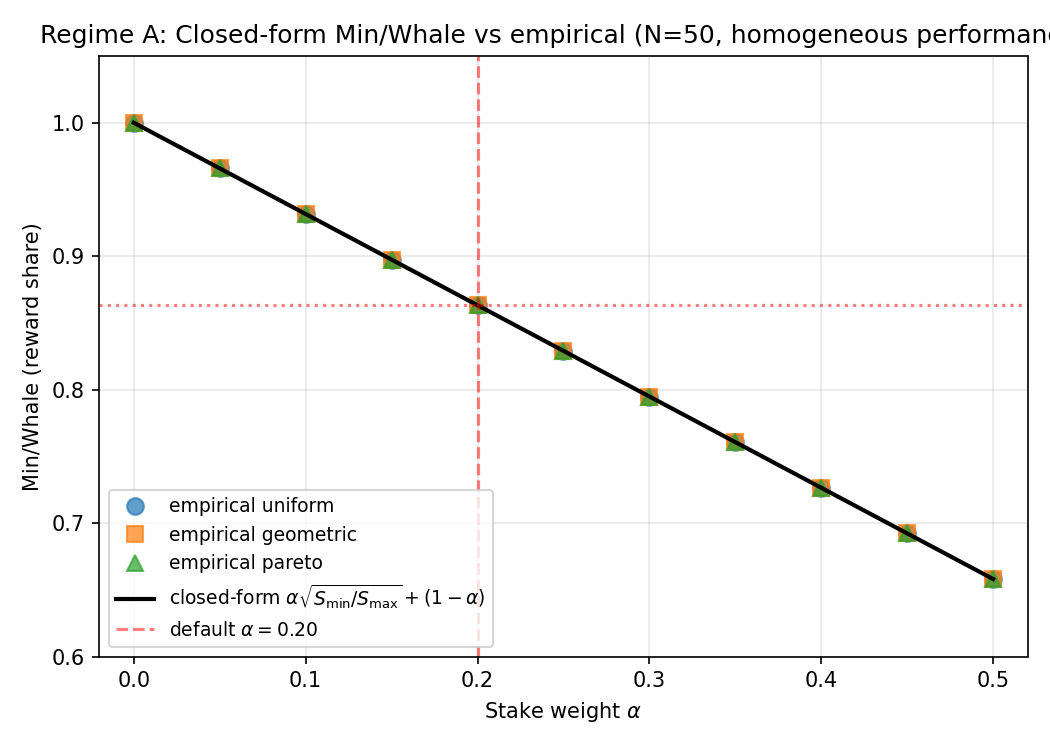}
	\caption{Multi-metric sensitivity to stake weight $\alpha$: Min/Whale ratio, Gini, Nakamoto coefficient, and HHI as a function of $\alpha \in [0, 1]$ for three network sizes. All metrics degrade monotonically with increasing $\alpha$, confirming default $\alpha = 0.20$ operates in a favorable region of parameter space.}
	\label{fig:alpha_sweep}
\end{figure}

\subsection{Input Sensitivity Analysis}

Three complementary analyses characterize parametric robustness: OAT elasticity, variance decomposition via Sobol indices, and combined perturbation robustness. All input elasticities for success\_rate satisfy $|e| < 0.25$, indicating that a $10\%$ perturbation on a single input yields output variation $< 2.5\%$. Sobol analysis reveals throughput is dominated by link\_latency ($S_{Ti} = 1.0$), while security metrics are dominated by interactions. Under a combined $\pm 25\%$ perturbation on all eight inputs, the CV success rate is 5.2\% (highly robust). Complete sensitivity visualization presented in Appendix~\ref{sec:supplementary} (Fig.~\ref{fig:tornado}--\ref{fig:robustness}).

\begin{table}[!tb]
	\centering
	\caption{Coefficient of variation under combined input perturbation $\pm 25\%$.}
	\label{tab:robustness}
	\begin{tabular}{lrl}
		\toprule
		\textbf{Metric}   & \textbf{CV} & \textbf{Assessment} \\
		\midrule
		success\_rate     & 5.2\%       & highly robust       \\
		gini              & 10.2\%      & robust              \\
		throughput\_proxy & 15.7\%      & moderate            \\
		hhi               & 25.1\%      & structural (1/$N$)  \\
		nakamoto\_third   & 26.1\%      & discrete metric     \\
		\bottomrule
	\end{tabular}
\end{table}

\section{Discussion}
\label{sec:discussion}

\subsection{Why Dual-Channel Decoupling Mitigates the Trilemma}

The elegant algebra of dual-channel decoupling emerges from applying \emph{opposite curvature} in stake across distinct protocol functions: super-linear selection $\Rightarrow$ Sybil resistance ($g < 1$); sub-linear reward $\Rightarrow$ fairness, bounded HHI; shared dependence on $R$ $\Rightarrow$ both channels reward quality. This decoupling avoids the old impossibility: previously, super-linear weights implied plutocracy, while sub-linear weights implied Sybil vulnerability.

Fig.~\ref{fig:dual_channel_arch} summarizes the dual-channel separation between selection and reward channels motivating the architectural claims in this section.

\subsection{Performance Trade-offs}

FairWave incurs a VRF overhead of $+2$~ms per anchor round relative to the modeled Bullshark in simulation. The modeled commit latency is $82.9$~ms versus Mysticeti's published $40.5$~ms~\cite{mysticeti2024}---a $2\times$ factor attributable to the 2-wave commit structure and VRF computation; however, direct comparison between simulation estimates and implementation measurements carries inherent uncertainty (Section~\ref{sec:limitations}). Fig.~\ref{fig:dag_wave_structure} visualizes the DAG wave and anchor/commit structure that underpins the selection channel. Post-quantum migration overhead (ML-DSA-44: $+5$~ms; Hybrid Ed25519+ML-DSA: $+7$~ms) remains acceptable for $> 99\%$ of use cases.

\subsection{Reputation Mechanism}

Reputation factor $\text{Rep}(i)$ evolves via three operations:
Fig.~\ref{fig:lagged_rep_timing} and~\ref{fig:epoch_ceremony} illustrate the lagged-reputation timing and the per-epoch seven-phase ceremony that realize the forward-causal reputation design described above.
(i)~increment of $+0.01$ per wave commit on active support, (ii)~multiplicative decay of $\times 0.99$ per epoch boundary (half-life $\approx 69$ epochs $\approx 1.9$ hours), and (iii)~punitive slashing ($-0.30$) on equivocation detection. Reputation farming resistance under adversarial conditions is characterized in (Figs.~\ref{fig:simplex_farming}).

\subsection{Caveat: Reward Channel Sybil Gain}

Reward channel in isolation (R-only) exhibits Sybil gain $> 1$ (Fig.~\ref{fig:sybil_reward}), due to floor $(1-\alpha)$ in $R$ for each Sybil identity. However, realized reward is proportional to anchor selection frequency (which uses $W_{\text{sel}}$), so realized reward inherits selection channel resistance.

\section{Limitations and Future Work}
\label{sec:limitations}

\textbf{Modeling assumptions.} Our analysis assumes a static stake within each epoch, an equivocation detection rate of $95\%$, and a stake distributed as Pareto ($\alpha_{\text{Pareto}} = 1.5$). Alternative distributions (log-normal, empirical Ethereum~\cite{eth_builder2024}) are characterized in supplementary material.

\textbf{Simulation-based validation.} All empirical results in this work are produced by Monte Carlo simulation, not by a real protocol implementation or live deployment. The eight comparative protocol baselines (Narwhal+Tusk, Bullshark, Mysticeti, HotStuff, HotStuff-2, PBFT, Tendermint, Algorand) are modeled according to their published specifications and documented design characteristics — they are not production-grade implementations. Consequently, the reported latency, throughput, and fairness numbers should be interpreted as simulation-based estimates under standardized parametric assumptions rather than implementation-level performance benchmarks. A geographically-distributed testnet deployment is required to validate whether the predicted properties hold under real network conditions.

\textbf{Future work.} (i)~Empirical deployment on geographically-distributed testnet ($\geq 50$ validators across three continents). (ii)~Adversarial sensitivity analysis where the adversary optimizes input combinations. (iii)~Higher-order Sobol indices ($S_{ij}$ interaction effects). (iv)~Formal verification of dual-channel composition in a mechanized proof assistant.

\textbf{Data and code availability.} The simulation framework and analysis scripts are available at \url{https://github.com/syarifulmujaddiq/fairwave-sim.git}. All experimental results are reproducible via the provided configuration files and seed (\texttt{0xC0FFEE}).

\textbf{Competing interests.} The author declares no competing interests.

\textbf{Funding.} This research received no external funding.

\textbf{Broader impact.} This work proposes a consensus protocol design and does not introduce new attack vectors beyond the standard BFT threat model analyzed in Section~\ref{sec:model}. The dual-channel principle aims to improve decentralization, which may benefit equitable participation in stake-based networks. As with any protocol design, deployment should proceed after independent security review and testnet validation. Potential centralization risks from parameter misspecification (e.g., excessive $\gamma$) are bounded by the closed-form constraints in Section~\ref{sec:math} and the sensitivity analysis in Section~\ref{sec:appendix_liveness}. The author is not aware of any direct negative societal consequences arising from this work.

\section{Conclusion}
\label{sec:conclusion}

We have presented FairWave, an asynchronous DAG-BFT consensus protocol whose central architectural insight is \emph{dual-channel decoupling}: super-linear-in-stake selection for Sybil resistance, paired with sub-linear-in-stake reward distribution for fairness. This decoupling mitigates the Sybil-fairness-plutocracy trilemma that no single weight function can resolve alone.

\textbf{Formal results.} (i)~Sybil selection gain $g(K) < 1$ for all $K > 1$ (Theorem~\ref{thm:sybil}); (ii)~the Herfindahl-Hirschman Index converges to the uniform fixed point $1/N$ under stake reinvestment (Theorem~2, Section~\ref{sec:math}); (iii)~selection weights are strictly forward-causal under lagged reputation, so anchor selection at epoch $e$ cannot bias selection weights at the same epoch (Proposition~\ref{prop:no_circular}). \emph{Safety} (agreement and validity) is a formal consequence of the $2f{+}1$ strong-support commit rule and holds unconditionally for $f < n/3$; see Section~\ref{sec:bft_validation}.

\textbf{Empirical results.} FairWave's simulation results show: (i)~a Gini coefficient of $0.140$ ($3.5\times$ reduction vs.~Pure-PoS) with monotone HHI decline from $0.039$ to $0.020$ over $50{,}000$ epochs under stake reinvestment, demonstrating convergence toward the uniform distribution (Min/Whale = $0.838$ in reward space, Table~\ref{tab:fairness}); (ii)~optimal adversarial Sybil split $K^{*} = 1$ — splitting is not rational, with closed-form bounds validated bit-exact against $20{,}160$ numerical cells; (iii)~liveness degrades monotonically from $94.0\%$ at $b=0.20$ to $74.0\%$ at $b=1/3$, consistent with the architectural expectation of no discontinuous cliff, and a success-rate CV of $5.2\%$ under $\pm 25\%$ perturbation.

The dual-channel principle is a general design pattern for stake-based protocols that must decouple operational security from economic equity. Future work---outlined in Section~\ref{sec:limitations}---targets a production-grade reference implementation and a geographically-distributed testnet deployment to confirm the model-predicted latency, throughput, and concentration trajectories under realistic network conditions, alongside formal verification of the dual-channel composition in a mechanized proof assistant.

\FloatBarrier

\bibliographystyle{IEEEtran}
\bibliography{references}

\begin{thebibliography}{10}
\providecommand{\url}[1]{#1}
\csname url@samestyle\endcsname
\providecommand{\newblock}{\relax}
\providecommand{\bibinfo}[2]{#2}
\providecommand{\BIBentrySTDinterwordspacing}{\spaceskip=0pt\relax}
\providecommand{\BIBentryALTinterwordstretchfactor}{4}
\providecommand{\BIBentryALTinterwordspacing}{\spaceskip=\fontdimen2\font plus
\BIBentryALTinterwordstretchfactor\fontdimen3\font minus \fontdimen4\font\relax}
\providecommand{\BIBforeignlanguage}[2]{{%
\expandafter\ifx\csname l@#1\endcsname\relax
\typeout{** WARNING: IEEEtran.bst: No hyphenation pattern has been}%
\typeout{** loaded for the language `#1'. Using the pattern for}%
\typeout{** the default language instead.}%
\else
\language=\csname l@#1\endcsname
\fi
#2}}
\providecommand{\BIBdecl}{\relax}
\BIBdecl

\bibitem{motepalli2024}
S.~Motepalli and M.~Jacobson, ``Proof-of-stake centralization: A longitudinal study,'' \emph{IEEE Transactions on Blockchain and Cryptocurrency}, 2024.

\bibitem{centralization_pos2024}
V.~Srivastava, S.~Damle, and S.~Gujar, ``Centralization in proof-of-stake blockchains: A game-theoretic analysis of bootstrapping protocols,'' \emph{arXiv preprint arXiv:2404.09627}, 2024.

\bibitem{buterin2020}
V.~Buterin, D.~Herman, D.~Reijsbergen, and S.~Leonardos, ``Ethereum proof-of-stake under the microscope,'' \emph{arXiv preprint arXiv:2005.09817}, 2020.

\bibitem{decentralization_measurement2022}
B.~Ku{\'s}mierz and R.~Overko, ``How centralized is decentralized? comparison of wealth distribution in coins and tokens,'' \emph{arXiv preprint arXiv:2207.01340}, 2022.

\bibitem{eth_beacon2023}
D.~Grandjean, L.~Heimbach, and R.~Wattenhofer, ``Ethereum proof-of-stake consensus layer: Participation and decentralization,'' \emph{arXiv preprint arXiv:2306.10777}, 2023.

\bibitem{daian2020}
P.~Daian, S.~Goldfeder, T.~Kell, Y.~Li, X.~Zhao, I.~Bentov, L.~Breidenbach, and A.~Juels, ``Flash boys 2.0: Frontrunning in decentralized exchanges, miner extractable value, and consensus instability,'' in \emph{2020 IEEE Symposium on Security and Privacy (SP)}, 2020, pp. 910--927.

\bibitem{narwhal2022}
G.~Danezis, L.~Kokoris-Kogias, A.~Sonnino, and A.~Spiegelman, ``Narwhal and tusk: A dag-based mempool and efficient bft consensus,'' in \emph{Proceedings of the 17th European Conference on Computer Systems (EuroSys)}, 2022, pp. 34--50.

\bibitem{bullshark2022}
A.~Spiegelman, N.~Giridharan, A.~Sonnino, and L.~Kokoris-Kogias, ``Bullshark: Dag bft protocols made practical,'' in \emph{Proceedings of the 2022 ACM SIGSAC Conference on Computer and Communications Security (CCS)}, 2022.

\bibitem{mysticeti2024}
K.~Babel, A.~Chursin, G.~Danezis, A.~Kichidis, L.~Kokoris-Kogias, A.~Koshy, A.~Sonnino, and M.~Tian, ``Mysticeti: Reaching the limits of latency with uncertified dags,'' in \emph{Network and Distributed System Security Symposium (NDSS)}, 2024.

\bibitem{hotstuff2019}
M.~Yin, D.~Malkhi, M.~K. Reiter, G.~Golan-Gueta, and I.~Abraham, ``Hotstuff: Bft consensus in the lens of blockchain,'' in \emph{Proceedings of the 2019 ACM Symposium on Principles of Distributed Computing (PODC)}, 2019.

\bibitem{hotstuff2_2023}
D.~Malkhi and K.~Nayak, ``Hotstuff-2: Optimal two-phase responsive bft,'' \emph{Cryptology ePrint Archive, Report 2023/397}, 2023.

\bibitem{pbft1999}
M.~Castro and B.~Liskov, ``Practical byzantine fault tolerance,'' in \emph{Proceedings of the Third USENIX Symposium on Operating Systems Design and Implementation (OSDI)}, 1999, pp. 173--186.

\bibitem{tendermint2014}
E.~Buchman, ``Tendermint: Byzantine fault tolerance in the age of blockchains,'' Whitepaper, 2014.

\bibitem{algorand2017}
Y.~Gilad, R.~Hemo, S.~Micali, G.~Vlachos, and N.~Zeldovich, ``Algorand: Scaling byzantine agreements for cryptocurrencies,'' in \emph{Proceedings of the 26th Symposium on Operating Systems Principles}, 2017, pp. 51--68.

\bibitem{dagrider2022}
I.~Keidar, L.~Kokoris-Kogias, O.~Naor, and A.~Spiegelman, ``All you need is dag,'' in \emph{Proceedings of the 2022 ACM Symposium on Principles of Distributed Computing (PODC)}, 2022, pp. 165--175.

\bibitem{shoal2023}
A.~Spiegelman, B.~Arun, R.~Gelashvili, and Z.~Li, ``Shoal: Improving dag-bft latency and robustness,'' \emph{arXiv preprint arXiv:2306.03058}, 2023.

\bibitem{shoalpp2024}
B.~Arun, Z.~Li, F.~Suri-Payer, S.~Das, and A.~Spiegelman, ``Shoal++: High throughput dag bft can be fast!'' \emph{arXiv preprint arXiv:2405.20488}, 2024.

\bibitem{miller2016}
A.~Miller, Y.~Xia, K.~Croman, E.~Shi, and D.~Song, ``The honey badger of bft protocols,'' in \emph{Proceedings of the 2016 ACM SIGSAC Conference on Computer and Communications Security (CCS)}, 2016, pp. 31--42.

\bibitem{benor1983}
M.~Ben-Or, ``Another advantage of free choice: Completely asynchronous agreement protocols,'' in \emph{Proceedings of the Second Annual ACM Symposium on Principles of Distributed Computing (PODC)}, 1983, pp. 27--30.

\bibitem{cosmos2019}
J.~Kwon, ``Cosmos governance framework,'' Cosmos Network Documentation, 2019.

\bibitem{ouroboros2017}
A.~Kiayias, A.~Russell, B.~David, and R.~Oliynykov, ``Ouroboros: A provably secure proof-of-stake blockchain protocol,'' in \emph{Advances in Cryptology -- CRYPTO 2017}, 2017, pp. 357--388.

\bibitem{ethereumpos2022}
V.~Buterin, D.~Herman, and D.~Reijsbergen, ``Ethereum proof-of-stake consensus specification (casper ffg + lmd ghost),'' Ethereum Foundation, Tech. Rep., 2022.

\bibitem{frd2022}
Z.~Naderi, S.~P. Shariatpanahi, and B.~Bahrak, ``Approach to alleviate wealth compounding in proof-of-stake cryptocurrencies,'' \emph{arXiv preprint arXiv:2207.11714}, 2022.

\bibitem{pos_reward_governance2024}
G.~Birmpas, P.~Lazos, E.~Markakis, and P.~Penna, ``Reward schemes and committee sizes in proof of stake governance,'' \emph{arXiv preprint arXiv:2406.10525}, 2024.

\bibitem{pos_concentration2022}
W.~Tang, ``Stability of shares in the proof of stake protocol -- concentration and phase transitions,'' \emph{arXiv preprint arXiv:2206.02227}, 2022.

\bibitem{gini1921}
C.~Gini, ``Measurement of inequality of incomes,'' \emph{The Economic Journal}, vol.~31, no. 121, pp. 124--126, 1921.

\bibitem{cukier1973}
R.~I. Cukier, C.~M. Fortuin, K.~E. Shuler, A.~G. Petschek, and J.~H. Schaibly, ``Study of the sensitivity of coupled reaction systems to uncertainties in rate coefficients. i. theory,'' \emph{Journal of Chemical Physics}, vol.~59, no.~8, pp. 3873--3878, 1973.

\bibitem{saltelli2010}
A.~Saltelli, M.~Ratto, T.~Andres, F.~Campolongo, J.~Cariboni, D.~Gatelli, M.~Saisana, and S.~Tarantola, \emph{Global Sensitivity Analysis: The Primer}.\hskip 1em plus 0.5em minus 0.4em\relax Wiley, 2010.

\bibitem{sobol2001}
I.~M. Sobol', ``Global sensitivity indices for nonlinear mathematical models and their monte carlo estimates,'' \emph{Mathematics and Computers in Simulation}, vol.~55, no. 1--3, pp. 271--280, 2001.

\bibitem{cachin2011}
C.~Cachin, R.~Guerraoui, and L.~Rodrigues, \emph{Introduction to Reliable and Secure Distributed Programming}, 2nd~ed.\hskip 1em plus 0.5em minus 0.4em\relax Springer, 2011.

\bibitem{dls1988}
C.~Dwork, N.~Lynch, and L.~Stockmeyer, ``Consensus in the presence of partial synchrony,'' \emph{Journal of the ACM}, vol.~35, no.~2, pp. 288--323, 1988.

\bibitem{douceur2002}
J.~R. Douceur, ``The sybil attack,'' in \emph{Proceedings of the First International Workshop on Peer-to-Peer Systems (IPTPS)}, 2002, pp. 251--260.

\bibitem{polkadot2020}
G.~Wood, ``Polkadot: Vision for a heterogeneous multi-chain framework,'' Whitepaper, 2020.

\bibitem{hedera2016}
L.~Baird, M.~Harmon, and P.~Madsen, ``Hedera: A public hashgraph network for fair ordering of transactions,'' Technical Report, 2016.

\bibitem{solana2020}
A.~Yakovenko, ``Solana: A new architecture for a high performance blockchain,'' Whitepaper, 2020.

\bibitem{eth_builder2024}
B.~{\"O}z, D.~Sui, T.~Thiery, and F.~Matthes, ``Who wins ethereum block building auctions and why?'' in \emph{Advances in Financial Technologies (AFT 2024)}, 2024.

\end{thebibliography}

\clearpage
\onecolumn

\setlength{\intextsep}{6pt}
\setlength{\floatsep}{6pt}
\setlength{\textfloatsep}{6pt}
\setlength{\abovecaptionskip}{4pt}
\setlength{\belowcaptionskip}{4pt}
\section*{APPENDIX}
\label{sec:supplementary}
\addcontentsline{toc}{section}{APPENDIX}

\subsection*{Heterogeneous Validator Profiles}

Table~\ref{tab:app_profile} lists the five validator profile types used in the fairness evaluation (Section~\ref{sec:results}) and their non-stake quality parameters. These profiles model heterogeneous validator behaviour: whale-good validators possess high stakes and near-perfect performance, while whale-bad validators have comparable stakes but degraded uptime, latency, and reputation. Small-good validators combine minimal stakes with excellent performance; small-bad validators match their stake level but with below-average quality.

\begin{table}[H]
	\centering
	\caption{Heterogeneous validator profile configuration used in the fairness analysis. Each profile specifies the nominal stake multiplier and the per-validator parameters $U, L, \mathrm{Rep}$ that enter the reward score $R(i)$ (Eq.~\ref{eq:reward_score}). The per-epoch validator set for $N=50$ comprises 3 whale-good, 3 whale-bad, 14 medium-good, 15 small-good, and 15 small-bad validators, for a total of 50.}
	\label{tab:app_profile}
	\begin{tabular}{lcccc}
		\toprule
		\textbf{Profile} & \textbf{Stake mult.} & \textbf{$U$} & \textbf{$L$} & \textbf{$\mathrm{Rep}$} \\
		\midrule
		Whale-good  & 0.95 & 0.92 & 0.92 & 0.92 \\
		Whale-bad   & 0.90 & 0.10 & 0.10 & 0.10 \\
		Medium-good & 0.50 & 0.45 & 0.45 & 0.45 \\
		Small-good  & 0.12 & 1.00 & 1.00 & 1.00 \\
		Small-bad   & 0.10 & 0.90 & 0.90 & 0.90 \\
		\bottomrule
	\end{tabular}
\end{table}

\subsection*{Experimental Configuration}

Table~\ref{tab:app_config} summarises the global experimental parameters used throughout the FairWave evaluation. With these parameters and the profile definitions in Table~\ref{tab:app_profile}, the reported results, trajectories, and metrics in Section~\ref{sec:results} are deterministically reproducible.

\begin{table}[H]
	\centering
	\caption{Global experimental parameters.}
	\label{tab:app_config}
	\begin{tabular}{lrl}
		\toprule
		\textbf{Parameter} & \textbf{Value} & \textbf{Description} \\
		\midrule
		$S_{\min}$        & 5,000     & Minimum validator stake \\
		$S_{\max}$        & 50,000    & Maximum validator stake \\
		$\alpha$          & 0.20      & Stake weight \\
		$\beta$           & 0.25      & Uptime weight \\
		$\gamma$          & 0.40      & Reputation weight \\
		$\delta$          & 0.15      & Latency weight \\
		Pareto $\alpha$   & 1.5       & Stake distribution shape \\
		Seed              & \texttt{0xC0FFEE} & Base random seed \\
		Runs per cell     & 30        & Monte Carlo replicates \\
		$N$ (Table~III)   & 50        & Validator count for fairness \\
		Epochs (longit.)  & 50,000    & Longitudinal reinvestment \\
		Reward rate       & $5\times10^{-5}$ & Per-epoch reinvestment rate \\
		\bottomrule
	\end{tabular}
\end{table}

\begin{figure}[H]
	\centering
	\begin{subfigure}[t]{0.48\columnwidth}
		\includegraphics[width=\textwidth]{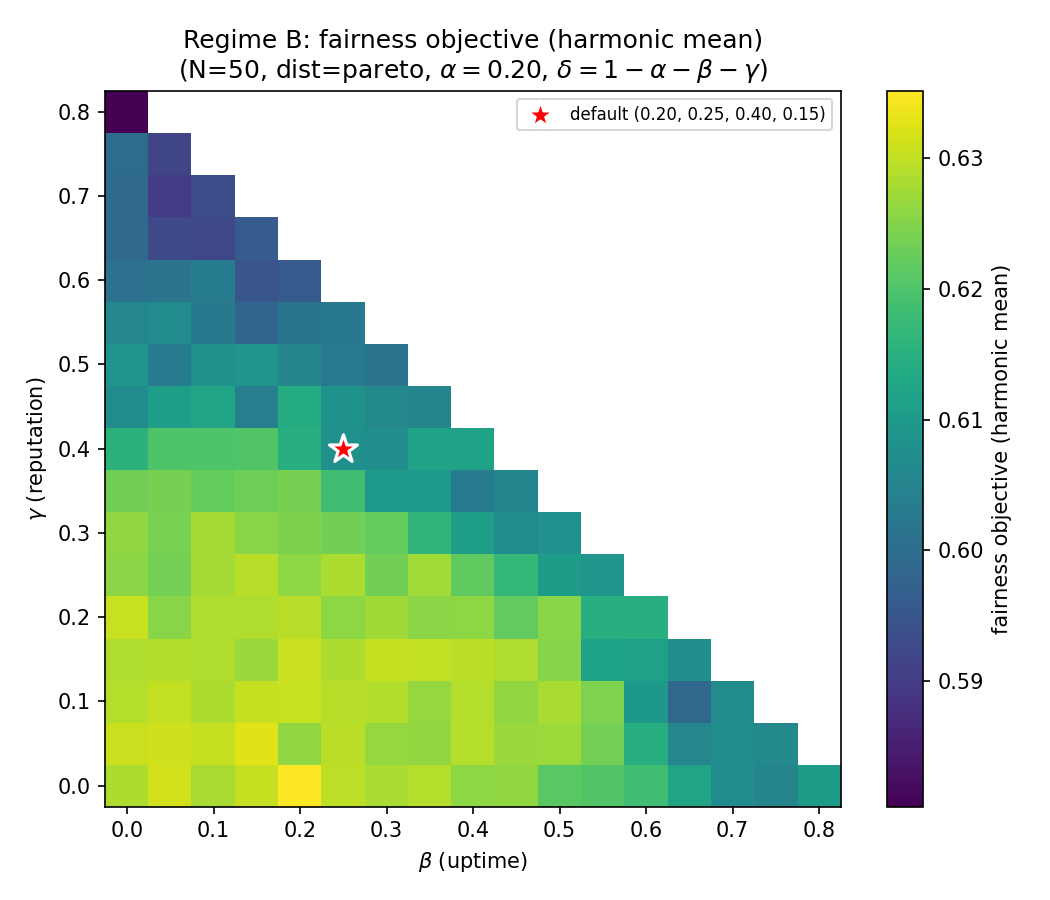}
		\caption{Objective landscape}
		\label{fig:simplex}
	\end{subfigure}
	\hfill
	\begin{subfigure}[t]{0.48\columnwidth}
		\includegraphics[width=\textwidth]{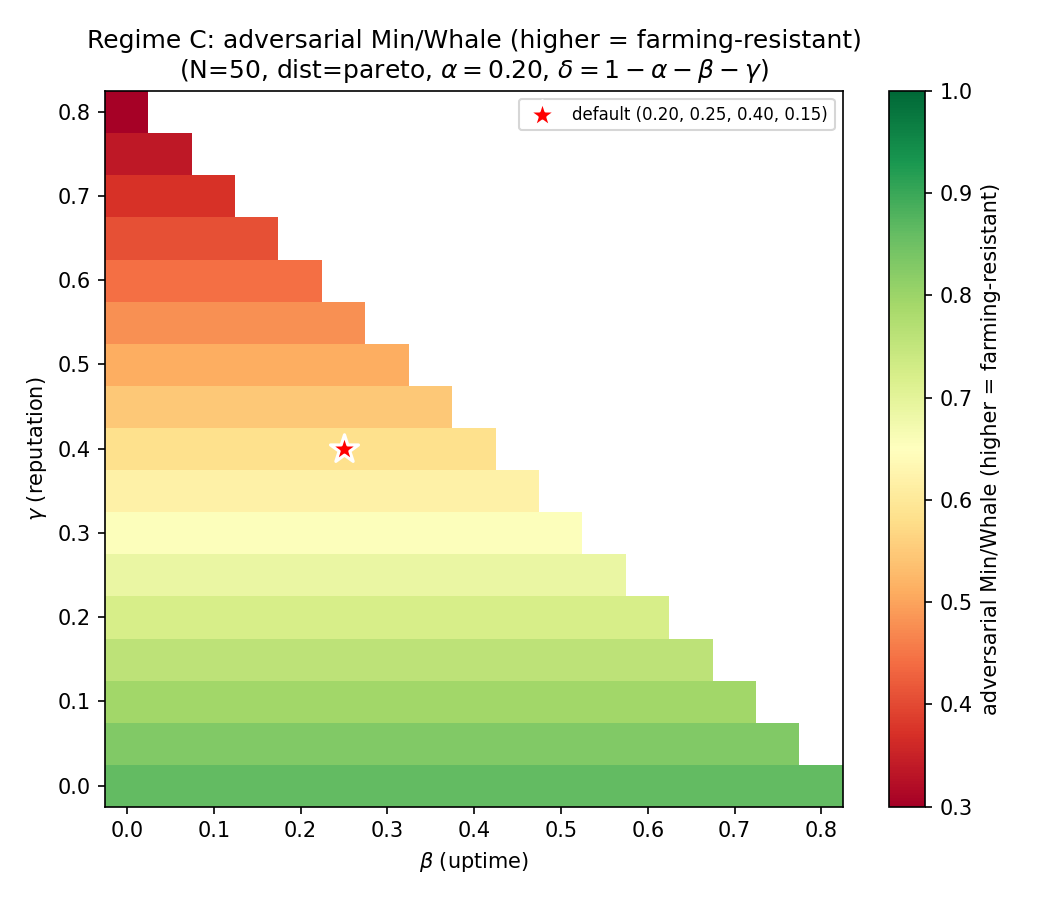}
		\caption{Farming resistance}
		\label{fig:simplex_farming}
	\end{subfigure}
	\caption{Parameter space characterization: (a)~objective function landscape on 4-D weight simplex with feasible region highlighted; (b)~reputation farming resistance versus $\gamma$.}
\end{figure}

\begin{figure}[H]
	\centering
	\includegraphics[width=0.70\columnwidth]{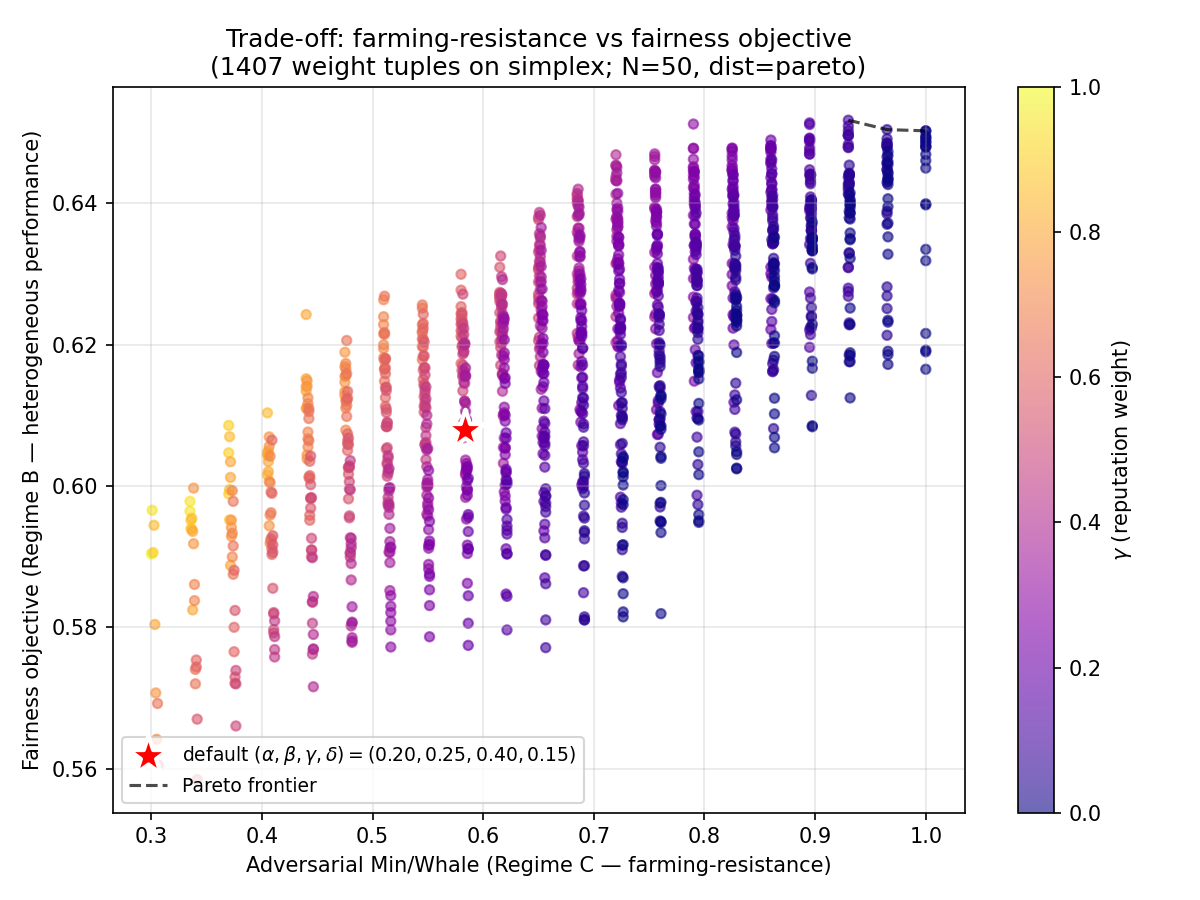}
	\caption{Pareto frontier: Min/Whale ratio versus farming resistance across 1,407 simplex evaluation points. Default parameterization is feasible but intentionally not Pareto-optimal for robustness.}
	\label{fig:pareto}
\end{figure}

\begin{figure}[H]
	\centering
	\begin{subfigure}[t]{1\columnwidth}
		\includegraphics[width=\textwidth]{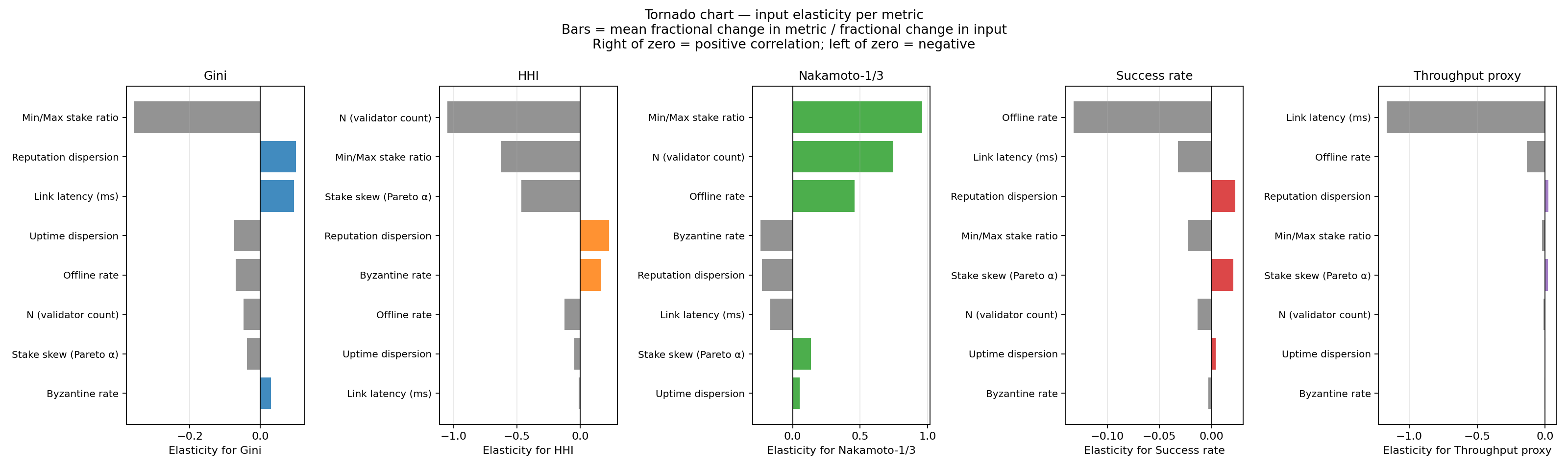}
		\caption{OAT elasticity (tornado)}
		\label{fig:tornado}
	\end{subfigure}
	\hfill
	\begin{subfigure}[t]{1\columnwidth}
		\includegraphics[width=\textwidth]{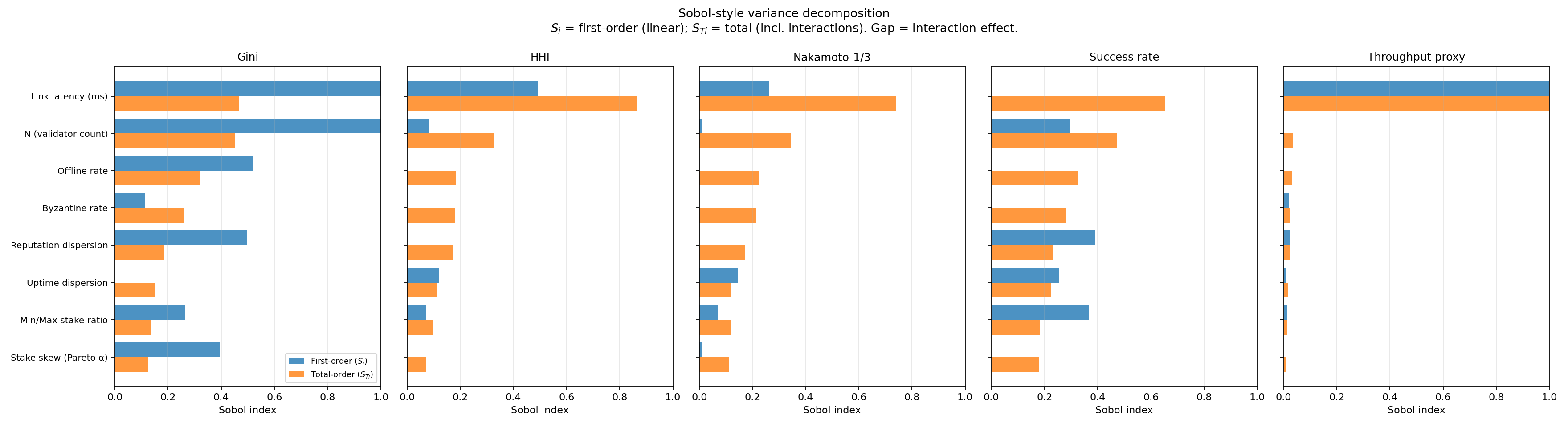}
		\caption{Sobol indices}
		\label{fig:sobol}
	\end{subfigure}
	\caption{Sensitivity analysis: (a)~One-At-a-Time elasticities---all success\_rate elasticities $|e| < 0.25$; (b)~first-order and total Sobol indices showing throughput dominated by link\_latency.}
\end{figure}

\begin{figure}[H]
	\centering
	\begin{subfigure}[t]{0.495\columnwidth}
		\includegraphics[width=\textwidth]{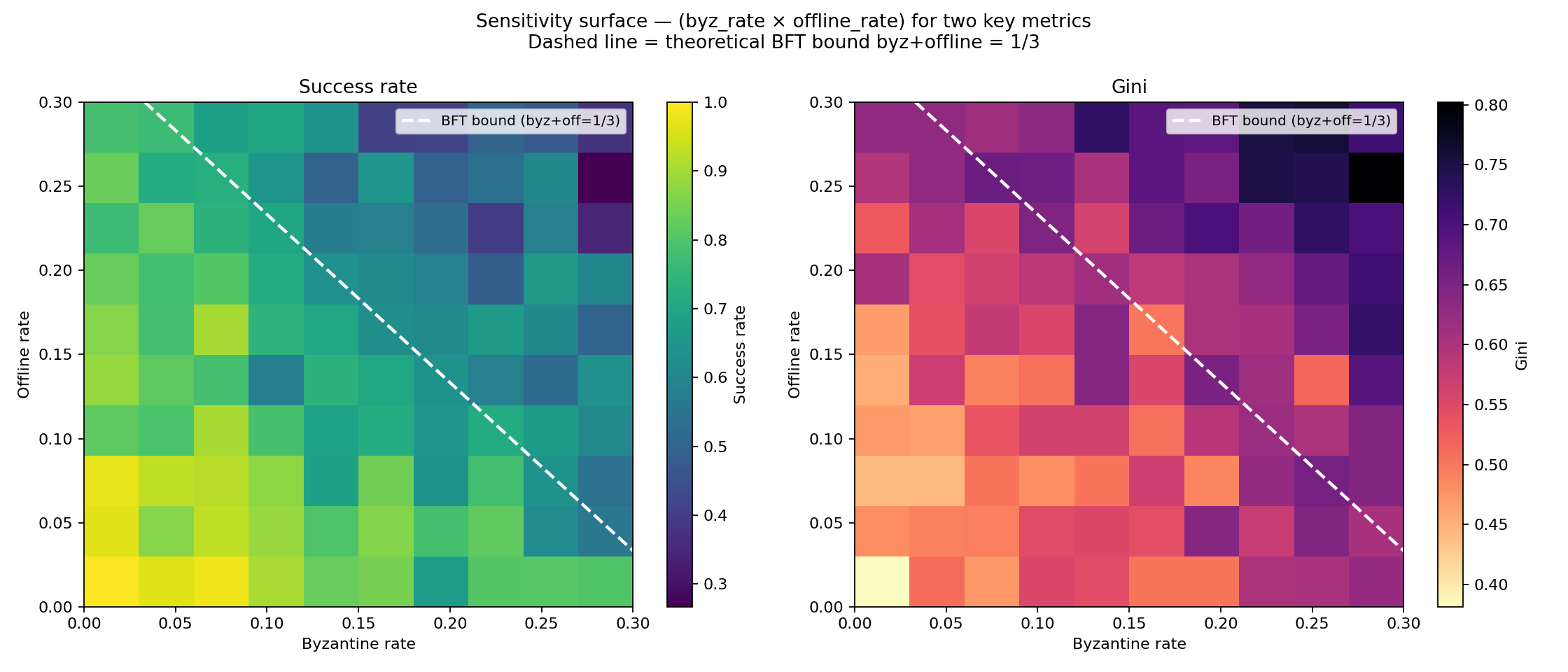}
		\caption{Sensitivity surface}
		\label{fig:surface}
	\end{subfigure}
	\hfill
	\begin{subfigure}[t]{0.495\columnwidth}
		\includegraphics[width=\textwidth]{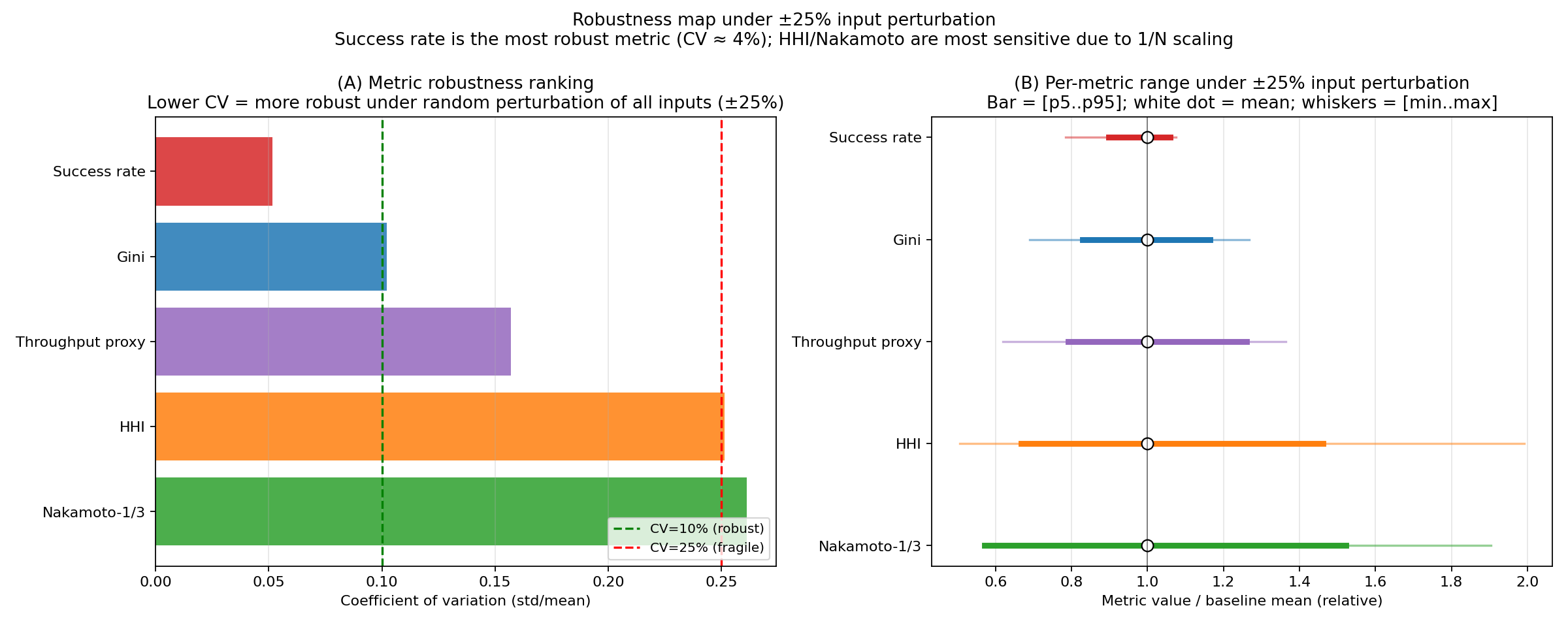}
		\caption{Robustness map}
		\label{fig:robustness}
	\end{subfigure}
	\caption{(a)~Success\_rate surface as a function of (byz\_rate, and offline\_rate)---BFT cliff at $1/3$ visible; (b)~robustness map under combined perturbation of $\pm 25\%$ (300 samples, CV $= 5.2\%$).}
\end{figure}

\begin{figure}[H]
	\centering
	\begin{subfigure}[t]{0.48\columnwidth}
		\includegraphics[width=\textwidth]{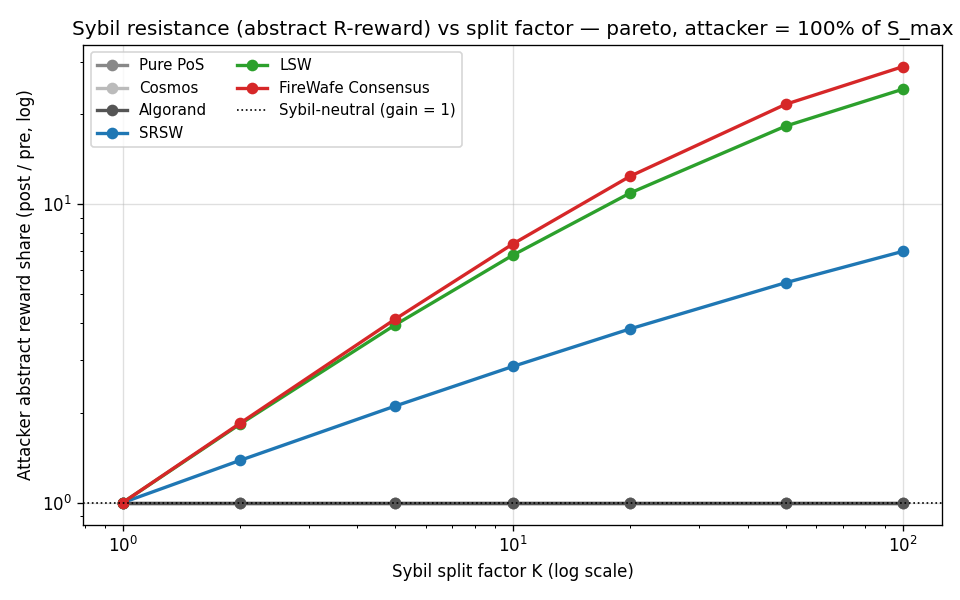}
		\caption{Sybil gain R-channel}
		\label{fig:sybil_reward}
	\end{subfigure}
	\hfill
	\begin{subfigure}[t]{0.48\columnwidth}
		\includegraphics[width=\textwidth]{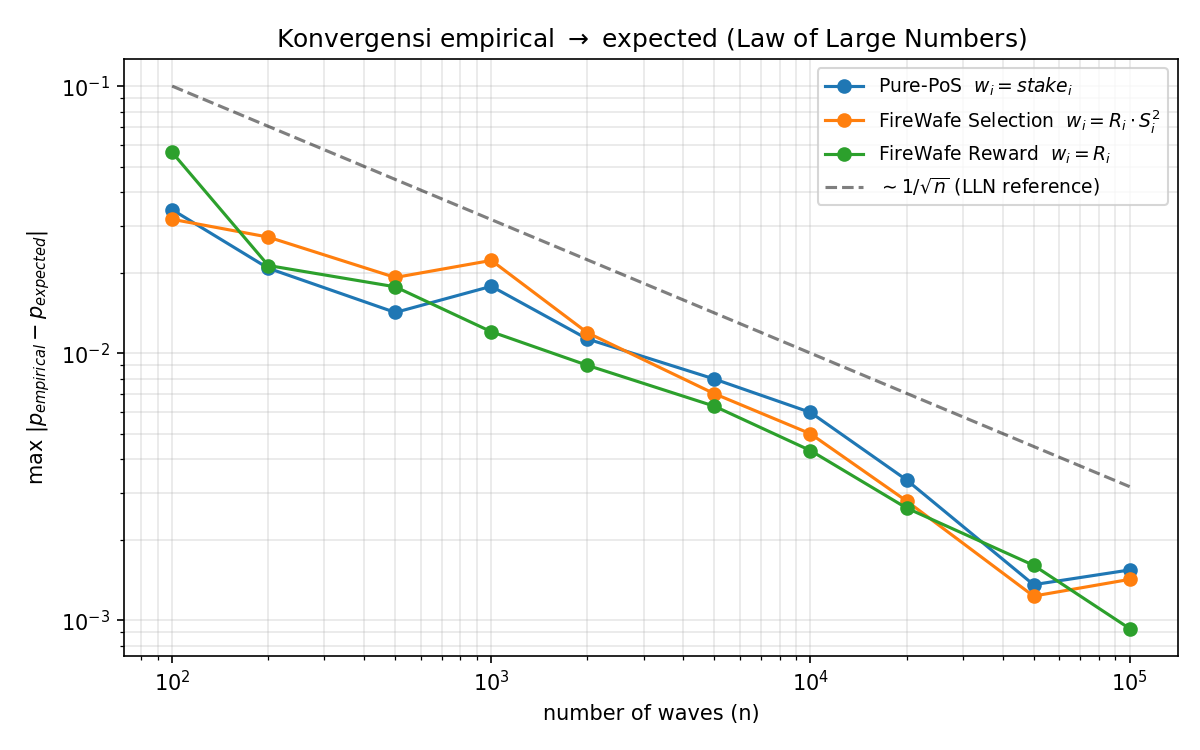}
		\caption{Anchor convergence}
		\label{fig:anchor_convergence}
	\end{subfigure}
	\caption{(a)~Reward-channel Sybil gain (isolated, $g_R > 1$; mitigated by selection channel in production); (b)~anchor share convergence at $O(1/\sqrt{T})$.}
\end{figure}

\begin{figure}[H]
	\centering
	\includegraphics[width=0.7\columnwidth]{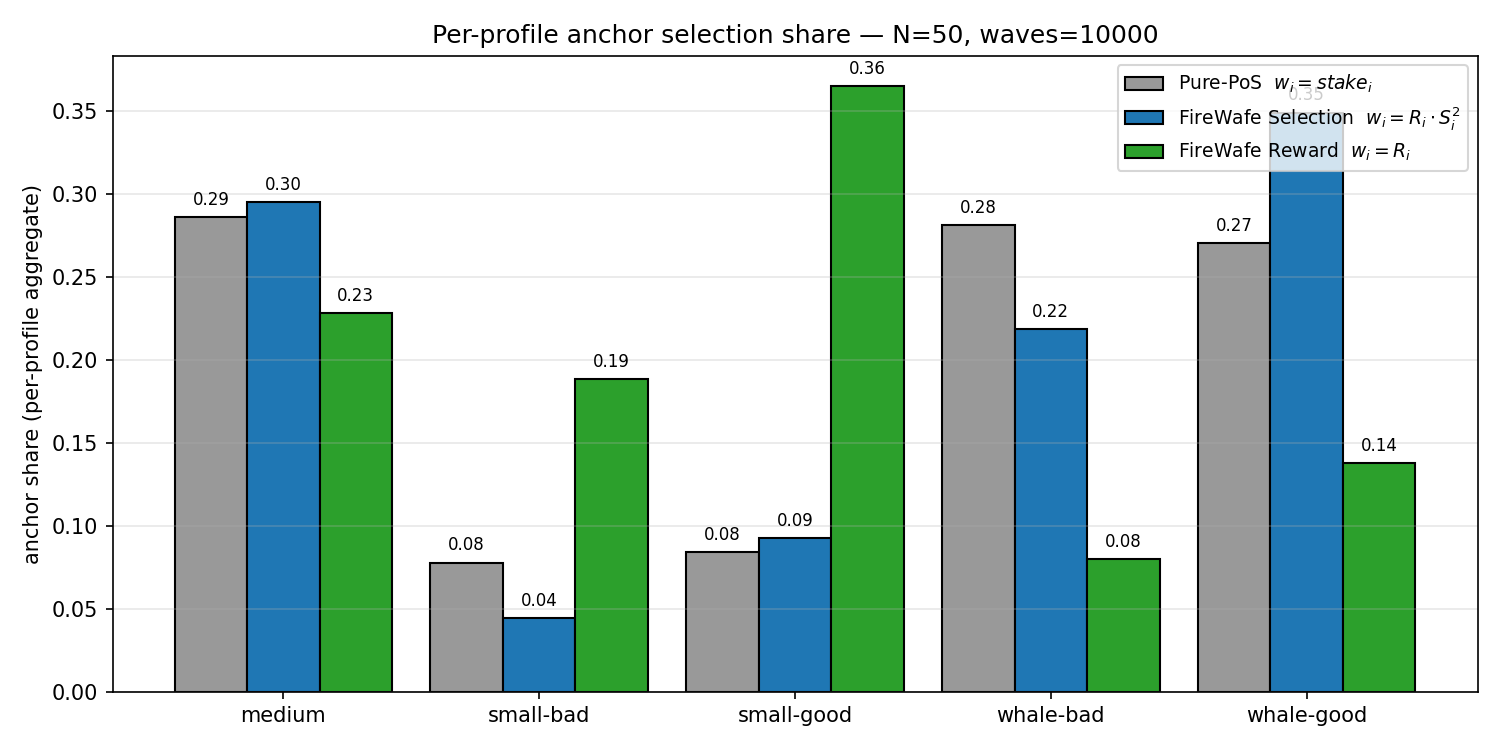}
	\caption{Per-profile anchor share allocation. FairWave selection reinforces whale-good (35\%) while penalizing whale-bad (22\%).}
	\label{fig:anchor_profile}
\end{figure}

\subsection*{Liveness Sensitivity Analysis}
\label{sec:appendix_liveness}

Table~\ref{tab:liveness_sensitivity} reports commit rate at key Byzantine fractions for each of the five liveness model parameters varied independently by $\pm 50\%$. All $9$ parameterizations preserve the three qualitative properties: monotone degradation (commit rate strictly decreasing with $b$), predicted absence of discontinuous cliff (drop between consecutive $b$ values $< 20$ pp), and commit rate $> 60\%$ at the theoretical boundary $b = 0.33$.

\begin{table}[H]
	\centering
	\caption{Liveness sensitivity to model parameters ($N=50$, $200{,}000$ waves).}
	\label{tab:liveness_sensitivity}
	\begin{tabular}{lcrrrl}
		\toprule
		\textbf{Parameter} & \textbf{Factor} & \textbf{$b=0.20$} & \textbf{$b=0.33$} & \textbf{Shape} \\
		\midrule
		$h_{\max}$ (max holdback) & $\times 0.5$ & $94.1\%$ & $74.1\%$ & OK \\
		$h_{\max}$ (max holdback) & $\times 1.0$ & $94.1\%$ & $74.1\%$ & OK \\
		$h_{\max}$ (max holdback) & $\times 1.5$ & $94.1\%$ & $74.1\%$ & OK \\
		\midrule
		$q_0$ (Byz. anchor block base) & $\times 0.5$ & $97.1\%$ & $78.9\%$ & OK \\
		$q_0$ (Byz. anchor block base) & $\times 1.0$ & $94.1\%$ & $74.1\%$ & OK \\
		$q_0$ (Byz. anchor block base) & $\times 1.5$ & $91.2\%$ & $69.4\%$ & OK \\
		\midrule
		$q_{\text{amp}}$ (block amplification) & $\times 0.5$ & $94.1\%$ & $82.3\%$ & OK \\
		$q_{\text{amp}}$ (block amplification) & $\times 1.0$ & $94.1\%$ & $74.1\%$ & OK \\
		$q_{\text{amp}}$ (block amplification) & $\times 1.5$ & $94.1\%$ & $68.1\%$ & OK \\
		\bottomrule
	\end{tabular}
\end{table}

\FloatBarrier
\end{document}